\documentclass[11pt,a4paper]{article}
\usepackage{jheppub}
\usepackage{amsthm}
\usepackage[T1]{fontenc}
\usepackage[utf8]{inputenc}
\usepackage{mathtools}   
\usepackage{amsfonts}
\usepackage{physics}
\usepackage{mathrsfs}
\usepackage{natbib}
\usepackage{tikz} 
\usepackage{enumitem}
\usepackage{subcaption}
\usepackage{graphicx}

\title{Blocks and Vortices in the 3d ADHM Quiver Gauge Theory}

\author{Samuel Crew,}
\author{Nick Dorey,}
\author{Daniel Zhang}
\affiliation{Department of Applied Mathematics and Theoretical Physics, University of Cambridge\\Cambridge, CB3 0WA, UK}

\emailAdd{s.c.crew@damtp.cam.ac.uk}
\emailAdd{n.dorey@damtp.cam.ac.uk}
\emailAdd{d.zhang@damtp.cam.ac.uk}

\abstract{We study the hemisphere partition function of a three-dimensional $\mathcal{N}=4$ supersymmetric $U(N)$ gauge theory with one adjoint and one fundamental hypermultiplet---the ADHM quiver theory. In particular, we propose a distinguished set of UV boundary conditions which yield Verma modules of the quantised chiral rings of the Higgs and Coulomb branches. In line with a recent proposal by two of the authors in collaboration with M. Bullimore, we show explicitly that the hemisphere partition functions recover the characters of these modules in two limits, and realise blocks gluing exactly to the partition functions of the theory on closed three-manifolds. We study the geometry of the vortex moduli space and investigate the interpretation of the vortex partition functions as equivariant indices of quasimaps to the Hilbert scheme of points in $\mathbb{C}^2$. We also investigate half indices of the ADHM quiver gauge theory in the presence of a line operator and discuss their geometric interpretation. Along the way we find interesting relations between our hemisphere blocks and related quantities in topological string theory and  equivariant quantum K-theory.}

\theoremstyle{definition}

\newtheorem{lemma}{Lemma}[section]

\begin{document}
\maketitle


\section{Introduction}
Three dimensional gauge theories with ${\mathcal N}=4$ supersymmetry sit at the centre of a remarkable web of connections between physics and mathematics \cite{Bullimore:2016hdc, Aganagic:2017smx, Bullimore:2015lsa, Rozansky:1996bq, Nakajima:2015txa, Braverman:2016wma, Cremonesi:2013lqa, Bullimore:2018jlp, Bullimore:2019qnt, Costello:2018swh}. On the physics side, these theories flow to strongly interacting conformal fixed points in the IR and, in many cases, the resulting conformal theories coincide with worldvolume theories of M-theory membranes \cite{Kapustin:2010xq, Gang:2011xp}. These CFTs also have gravitational duals with supersymmetric black hole solutions. Accounting for the entropy of these black holes from the perspective of the dual gauge theory is an active subject of research \cite{Benini:2015eyy,Benini:2016rke,Choi:2019dfu,Choi:2019zpz}.   

On the mathematical side, the vacuum moduli spaces of a large class of these theories coincide with Nakajima quiver varieties \cite{Nakajima:1994nid}. The quantised coordinate rings of these spaces give rise to interesting non-commutative algebras \cite{Braden:2014iea,MR3594663}. The actions of quantum groups and algebras also arise in the K-theory of these spaces \cite{nakajima1998} and in the enumerative geometry of curves (quantum K-theory) \cite{Aganagic:2017gsx,Smirnov:2016cqz,Koroteev:2017nab}. Three dimensional $\mathcal{N}=4$ theories are equipped with various protected observables which can be evaluated in the IR and are often characterised in terms of the geometry of quiver varieties and the representation theory of the associated algebras.  

Among the protected observables of 3d theories with $\mathcal{N}\ge 2$
supersymmetry are partition functions computed on certain closed three manifolds. Remarkably, a wide class of these partition functions can be constructed from a common constituent, the holomorphic blocks of \cite{Beem:2012mb}. The factorisation into blocks has been demonstrated in a number of examples including the three-sphere partition function \cite{Pasquetti:2011fj}, the superconformal index \cite{Hwang:2012jh} and topologically twisted indices \cite{Cabo-Bizet:2016ars,Crew:2020jyf}. The factorisation can also be understood from the Higgs branch localisation perspective \cite{Benini:2013yva,Fujitsuka:2013fga}. The blocks are interesting quantities in their own right which receive both perturbative quantum corrections and non-perturbative contributions from the vortices of the 3d theory. For theories with an AdS dual, recent work \cite{Choi:2019zpz,Choi:2019dfu} suggests that in a limit of large angular momentum a single ``Cardy block'' dominates the thermodynamic ensemble relevant for calculating black hole entropy. 

Work by two of the authors with M. Bullimore \cite{Bullimore:2020jdq} has provided a first principles construction of these fundamental blocks as hemisphere partition functions on $S^1\times D$ with exceptional Dirichlet UV boundary conditions. Geometrically, these boundary conditions flow to thimble branes in the IR Rozansky-Witten $\sigma$-model. The state-operator correspondence relates the hemisphere partition functions to a half-index counting local operators inserted at the origin of $\Omega$-deformed $\mathbb{R}_{\geq 0}\times \mathbb{R}^2$ and in this picture, as detailed in \cite{Bullimore:2016nji}, these particular UV boundary conditions are associated with Verma modules of the quantised Higgs and Coulomb branch chiral rings.

In this paper we elucidate these ideas for a particularly interesting example: the ADHM quiver theory.  This theory is realised as the three dimensional worldvolume theory of $N$ $D2$ branes on top a single $D6$ brane in type IIA string theory and has a Lagrangian quiver description (see figure \ref{fig:adhm}) with one adjoint and one fundamental hupermultiplet.\footnote{The name ``ADHM quiver'' comes from the fact that the hyperk\"{a}hler quotient description of the Higgs branch of this theory coincides with the ADHM construction \cite{Atiyah:1978ri} of the moduli space of $N$ non-commutative instantons \cite{Nekrasov:1998ss} in a $U(1)$ gauge theory.} The ADHM theory flows in the IR to the ABJM theory describing the worldvolume of $N$ M2 branes and the ABJM theory is in turn holographically dual to M-theory on $\text{AdS}_{4}\times S^{7}$ \cite{Aharony:2008ug}. In interesting recent work \cite{Choi:2019zpz}, the Cardy block of the ADHM quiver theory has been evaluated in the large $N$ limit and found to reproduce the entropy of supersymmetric asymptotically $\text{AdS}_{4}$ black holes. The ADHM quiver theory also appears in recent attempts to formulate a topologically twisted version of the AdS/CFT correspondence \cite{Costello:2017fbo}. 

In the present work we construct the blocks, realised as hemisphere partition functions, for the ADHM theory and verify explicitly that they glue to reproduce the superconformal index, the $S^{3}$ partition function and the $A$- and $B$-twisted indices. The theory is self-dual under 3d mirror symmetry and both the Coulomb and Higgs branch coincide with the Hilbert scheme of $N$ points on $\mathbb{C}^{2}$. The corresponding quantised coordinate ring has been identified by Kodera and Nakajima as a cyclotomic rational Cherednik algebra \cite{Kodera:2016faj} and we show that specialised limits of these blocks reproduce Verma characters of this algebra in accordance with the general theory of \cite{Bullimore:2020jdq}. Along the way, we discuss the implications of 3d mirror symmetry in the geometric setup and discuss the relationship between the twisted indices and the Hilbert series of the Hilbert scheme of $N$ points in $\mathbb{C}^2$.

In the following we provide several different perspectives on the block of the ADHM theory. We first show that the vortex contributions coincide with the equivariant K-theoretic vertex function of quasi-maps to the Hilbert scheme of points---a similar relation appears for linear quivers in the q-Langlands correspondence \cite{Aganagic:2017smx}. We also give a geometric characterisation of the blocks in certain limits as the Poincar\'e polynomials of quasi-map or vortex moduli spaces. 
Finally, we discuss a connection between the 3d blocks of the ADHM theory and the 1-leg K-theoretic PT vertex. In particular, we use holomorphic factorisation to understand the twisted index of the 3d ADHM theory as the $N$ $D2$ brane sector of the conifold amplitude in PT theory.

\paragraph{Outline}
We begin in section \ref{sec:background} with an overview of the UV description of the theory and discuss combinatorial aspects of fixed points on the vacuum moduli space.

In section \ref{sec:vpfn}, we construct the hemisphere partition function of the 3d ADHM theory with boundary conditions associated to vacua. We discuss geometric aspects of two supersymmetry enhancing limits, namely the Verma character limit and the limit in which the vortex partition function is expected to coincide with a generating function of Poincar\'e polynomials of the vortex moduli space. We find combinatorial expressions for the vortex contributions to the hemisphere partition functions and show that 3d mirror symmetry implies interesting identities for generating functions of reverse plane partitions.

In section \ref{sec:factorisation} we turn to holomorphic factorisation of the 3d ADHM theory. We explicitly demonstrate the exact factorisation into hemisphere partition functions of the $A$- and $B$-twisted indices and discuss the connection to the Hilbert series of the Hilbert scheme of points in $\mathbb{C}^2$. Using the geometric interpretation of section \ref{subsec:quasimaps}, we then relate the Hilbert series to the Poincar\'e polynomial of the vortex moduli space and use Macdonald polynomial methods \ref{appendix:symmetricfuns} to compute the large $N$ limit.

In this work we also study an alternative (Neumann) choice of boundary condition \cite{Yoshida:2014ssa}. The corresponding half index in the presence of a line operator can be expressed as a contour integral. In section \ref{sec:qm}, we focus on a particularly simple Neumann boundary condition and show that the index can be realised as counting gauge invariant states in the matrix model for a certain Chern-Simons quantum mechanics. We discuss a geometric interpretation of this boundary condition as an equivariant Euler characteristic counting sections of holomorphic line bundles over a distinguished Lagrangian in the Higgs branch of the ADHM theory. Finally, we compute the Euler characteristic in terms of Milne symmetric polynomials and conjecture that the matrix model yields simple modules of the ADHM Coulomb branch algebra $\mathcal{A}^C_N$.

We include detailed appendices covering conventions for partitions and symmetric functions \ref{appendix:combinatorics} as well as novel results on evaluating Molien integrals using Macdonald polynomial methods \ref{appendix:symmetricfunctionmethods}.


\section{Background}\label{sec:background}
We focus on a particular 3d gauge theory, denoted 3d ADHM, living on N D2-branes on top a single D6-brane in type IIA string theory. In the IR this theory is expected to flow to an $\mathcal{N}=8$ SCFT, the ABJM theory, living on the worldvolume of $N$ M2-branes in flat spacetime. In the UV the theory has a Lagrangian description as a 3d $\mathcal{N}=4$ theory with $G = U(N)$ gauge symmetry.

Various partition conventions used in this section are summarised in appendix \ref{appendix:conventions}.

\paragraph{UV description}
The field content is summarised by the Jordan quiver in figure \ref{fig:adhm}---we refer to this theory as 3d ADHM with one flavour because of the role of this quiver in the ADHM construction \cite{Atiyah:1978ri} of the instanton moduli space. The theory has a vector multiplet, one hypermultiplet in the fundamental representation $(I,J)$ and an adjoint hypermultiplet $(A,B)$. 

\begin{figure}
\centering

\tikzset{every picture/.style={line width=0.75pt}} 

\begin{tikzpicture}[x=0.75pt,y=0.75pt,yscale=-1,xscale=1]

\draw   (265,115) .. controls (265,101.19) and (276.19,90) .. (290,90) .. controls (303.81,90) and (315,101.19) .. (315,115) .. controls (315,128.81) and (303.81,140) .. (290,140) .. controls (276.19,140) and (265,128.81) .. (265,115) -- cycle ;
\draw   (270,190) -- (310,190) -- (310,227.05) -- (270,227.05) -- cycle ;
\draw    (285,140) -- (285,187) ;
\draw [shift={(285,190)}, rotate = 270] [fill={rgb, 255:red, 0; green, 0; blue, 0 }  ][line width=0.08]  [draw opacity=0] (8.93,-4.29) -- (0,0) -- (8.93,4.29) -- cycle    ;
\draw    (310,100) .. controls (389.53,81.02) and (302.12,12.55) .. (300.04,87.68) ;
\draw [shift={(300,90)}, rotate = 270.24] [fill={rgb, 255:red, 0; green, 0; blue, 0 }  ][line width=0.08]  [draw opacity=0] (8.93,-4.29) -- (0,0) -- (8.93,4.29) -- cycle    ;
\draw    (295,190) -- (295,143) ;
\draw [shift={(295,140)}, rotate = 450] [fill={rgb, 255:red, 0; green, 0; blue, 0 }  ][line width=0.08]  [draw opacity=0] (8.93,-4.29) -- (0,0) -- (8.93,4.29) -- cycle    ;
\draw    (270,100) .. controls (191.46,80.69) and (278.56,12.55) .. (279.98,87.68) ;
\draw [shift={(280,90)}, rotate = 270.24] [fill={rgb, 255:red, 0; green, 0; blue, 0 }  ][line width=0.08]  [draw opacity=0] (8.93,-4.29) -- (0,0) -- (8.93,4.29) -- cycle    ;

\draw (284,200.4) node [anchor=north west][inner sep=0.75pt]    {$1$};
\draw (281,107.4) node [anchor=north west][inner sep=0.75pt]    {$N$};
\draw (251,66.4) node [anchor=north west][inner sep=0.75pt]    {$A$};
\draw (315.5,65.4) node [anchor=north west][inner sep=0.75pt]    {$B$};
\draw (302,153.9) node [anchor=north west][inner sep=0.75pt]    {$I$};
\draw (264.5,154.4) node [anchor=north west][inner sep=0.75pt]    {$J$};

\end{tikzpicture}

		\caption{The Jordan quiver specifying the field content of 3d ADHM with one flavour. }\label{fig:adhm}
\end{figure}
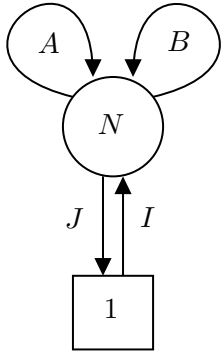

In a fixed $\mathcal{N}=2$ subalgebra the theory has R-symmetries $R_H=2 U(1)_H$ and $R_C=2 U(1)_C$ (normalised to have integer charges) acting on the vectormultiplet and hypermultiplet scalars respectively and the flavour symmetries acting on hypermultiplets and monopole operators are $G_H \cong U(1)$ and $G_C \cong U(1)$ respectively. We summarise the charges of the scalar components of the hypermultiplets below.
\begin{center}
    \begin{tabular}{c|c|c|c|c}
            & $R_H$ & $R_C$ & $G$ & $G_H$  \\ \hline
        $I$ & $+1$ & $0$ & $\square$ & $0$ \\
        $J$ & $+1$ & $0$ & $\overline{\square}$ & $0$ \\
        $A$ & $+1$ & $0$ & $\text{adj}$ & $+1$ \\
        $B$ & $+1$ & $0$ & $\text{adj}$ & $-1$
    \end{tabular}
\end{center}
The scalars can be regarded as linear maps:
\begin{equation}
    A,B \in \text{Hom}(V,V),\quad I \in \text{Hom}(W,V),\quad J \in \text{Hom}(V,W),
\end{equation}
where $V = \mathbb{C}^N$, $W = \mathbb{C}$.

In the presence of a real FI parameter $\xi$ and mass parameter $m$, the space of classical supersymmetric vacua consists of solutions to:
\begin{equation}\label{eq:ADHMmomentmaps}
   \mu_{\mathbb{C}} = [A,B]+I J = 0,\qquad \mu_{\mathbb{R}}=[A,A^\dagger]+[B,B^\dagger]+ II^\dagger - J^\dagger J = \xi,
\end{equation}
and:
\begin{equation}\label{eq:scalargroupaction}
   [\sigma,A]+mA = 0, \qquad [\sigma,B]-m B = 0, \qquad \sigma I = 0, \qquad -J \sigma = 0,
\end{equation}
modulo gauge transformations:
\begin{equation}
    (A,B,I,J) \mapsto \left(g A g^{-1}, g B g^{-1}, g I, J g^{-1}\right),
\end{equation}
with $g \in U(V)$ and $\sigma$ is the real scalar. $\mu_{\mathbb{R}}$ and $\mu_{\mathbb{C}}$ are real and complex moment maps for the $G$ action and, in the language of 3d $\mathcal{N}=2$ supersymmetry, correspond to $D$ and $F$ terms respectively. Throughout this work, we do not consider turning on a complex mass and FI parameter. Assuming $\xi >0$ and setting $m=0$, these equations require $\sigma=0$ and we recover the hyperk\"ahler quotient description of the Higgs branch $\mathcal{M}_H$ as the Hilbert scheme of points in the plane  -- $\text{Hilb}^N(\mathbb{C}^2)$:
\begin{equation}
    \mathcal{M}_{H} = \text{Hilb}^N(\mathbb{C}^2) = \mu_{\mathbb{R}}^{-1}(\xi)\cap \mu_{\mathbb{C}}^{-1}(0) / U(V).
\end{equation}
The FI parameter $\xi$ is a resolution parameter for $\text{Hilb}^N(\mathbb{C}^2) \to \text{Sym}^N(\mathbb{C}^2)$. Alternatively, $\mathcal{M}_H$ can be realised as a complex symplectic quotient by substituting the real moment map for a stability condition and performing the quotient by $\text{GL}(V)$ as in \cite{nakajima1999lectures}:
\begin{equation}
    \text{Hilb}^N(\mathbb{C}^2) = \left. \left\{ (A,B,I,J)\, \left\lvert\,\,
    \begin{aligned}
    &\mu_{\mathbb{C}} = [A,B]+I J = 0,\\
    &\nexists \text{ proper } S\subsetneq V \text{ s.t. } A(S)\subset S,\\
    &\quad B(S)\subset S, \text{im}I \subset S 
    \end{aligned}\right. \right\}\middle/GL(V)\right. .
\end{equation}

\paragraph{Group action and fixed points}
There is a natural $T^2=T_1\times T_2$ action on $ \text{Hilb}^N(\mathbb{C}^2)$ induced by the following action on the linear data:
\begin{equation}
    (t_1,t_2):\,(A,B,I,J) \mapsto (t_1 A, t_2 B, I , t_1 t_2 J),
\end{equation}
where $(t_1,t_2)\in \mathfrak{t}^2$. Further, it is shown in \cite{nakajima1999lectures} that enforcing the complex moment map and stability condition, or alternatively both moment maps (\ref{eq:ADHMmomentmaps}), forces $J=0$. Setting $t_1 = z t^{\frac{1}{2}} $ and $t_2 = z^{-1} t^{\frac{1}{2}}$ we have that, up to gauge transformations, $z$ and $t$ are fugacities for $G_H$ and $(R_H-R_C)/2$ respectively. We can regard these group actions as generated by mass deformations $m$ and $\tau$ with $z=e^{-m}$ and $t=e^{-\tau}$. Note that turning on the mass deformation $\tau$ softly the breaks the supersymmetry to $\mathcal{N}=2^*$. For later use we also define the exponentiated fugacity $\zeta = e^{-\xi}$ for the $G_C$ topological symmetry.

We now consider the fixed points under these group actions and briefly recap the results of \cite{nakajima1999lectures}. The fixed points coincide with isolated supersymmetric vacua in the presence of mass deformations -- we note that a non-zero $m$ is itself enough to give isolated vacua (see e.g. \cite{Smirnov:2018drm}). The fixed points are described by linear data $(A,B,I)$ such that 
\begin{equation}
    t_1 A = \lambda(t_1,t_2)^{-1} A \lambda(t_1,t_2)\,,\quad 
    t_2 B = \lambda(t_1,t_2)^{-1} B \lambda(t_1,t_2)\,,\quad 
    I = \lambda(t_1,t_2)^{-1} I,
\end{equation}
where $\lambda: T^2\rightarrow U(V)$.\footnote{The vacuum equations involving the adjoint scalar (\ref{eq:scalargroupaction}) are precisely the infinitesimal version of these equations.} $V$ can be decomposed with respect to the eigenspaces of $\lambda$ as follows: 
\begin{equation}
    V = \bigoplus_{k, l} V(k,l)\,,\quad V(k,l)=\left\{v \in V\,\rvert\, \lambda(t_1,t_2)\cdot v = t_1^{k}t_2^{l}v\right\}
\end{equation}
and, abusing notation, $\lambda$ denotes the homomorphism as well as the fixed point itself. Indeed, $\lambda$ represents the Young diagram of a partition of weight $N$ with the box $s=(i,j)$ representing an eigenspace $V_s \equiv V(-i_s+1,-j_s+1)$. The components of scalars $(A,B,I,J)$ in the vacuum $\lambda$ are given by
\begin{align}\label{eq:ADHMvacuum}
       A_{ab} =
    \begin{cases}
      1 & \text{if} \quad i_a = i_b+1 \quad \text{and} \quad j_a = j_b\\
      0 & \text{otherwise}
    \end{cases} ,
\nonumber\\
      B_{ab} =
    \begin{cases}
      1 & \text{if} \quad i_a = i_b \quad \text{and} \quad j_a=j_b+1\\
      0 & \text{otherwise}
    \end{cases} ,
\end{align}
\begin{equation}
      I_a = \delta_{a,1},
\qquad 
      J_a = 0 ,\nonumber
\end{equation}
with
\begin{equation}
    \sigma_{ab} = \delta_{ab}(j_a-i_a)m \, ,
\end{equation}
modulo Weyl transformations. Here indices $a,b=1,\ldots,N$ are for the fundamental representation of $U(V)$ and equivalently label boxes $a,b \in \lambda$ or their corresponding $\lambda(t_1,t_2)$ eigenspaces $V_{a}$ and $V_b$. Also $(i_s,j_s)$ denote the coordinates of the $s^{\text{th}}$ box of $\lambda$. 

As explained in more detail in appendix \ref{appendix:localisation}, the fixed points can be realised as critical values of the Morse function $h_H(\lambda) = m\cdot \mu_{H,\mathbb{R}}$ where $\mu_{H,\mathbb{R}}$ is the real moment map for the Hamiltonian action of $G_H$. The Morse flow with respect to this function provides an ordering on vacua
\begin{equation}
    \nu\in \overline{\mathcal{L}_{\lambda}} \Rightarrow \nu \leq \lambda \,,
\end{equation}
where $\mathcal{L}_\lambda$ is the holomorphic Lagrangian attracting submanifold of the vacuum $\lambda$ under upwards gradient flow. This Lagrangian locally coincides with the positive (with respect to $z$) tangent weight space at $\lambda$. The critical value at the fixed point is, see proposition (5.13) of \cite{nakajima1999lectures}:\footnote{Note we have adapted the notation in \cite{nakajima1999lectures} to the standard Macdonald \cite{macdonald1998symmetric} notation.}
\begin{equation}
\begin{split}
     m \cdot \mu_{H,\mathbb{R}}(\lambda) &=  m\, \xi \bigg(\sum_{i} (i-1)\lambda_i - \sum_{j}(j-1)\lambda_j^{\vee}  \bigg) \\
     &= - m \, \xi \sum_{s\in\lambda} c_{\lambda}(s).
\end{split}
\end{equation}
Consequently the ordering on the vacua is related to a partial order on Young diagrams (\ref{eq:dominanceordering}).

\paragraph{3d mirror symmetry}
The theory has a Coulomb branch algebra $\mathbb{C}[\mathcal{M}_C]$ generated by vevs of monopole operators. The Coulomb branch $\mathcal{M}_C$ is defined by the spectrum of this algebra \cite{Bullimore:2015lsa, Nakajima:2015txa} and turning on an FI parameter $\xi$ and real mass $m$ the resolved Coulomb branch is given by
\begin{equation}
    \mathcal{M}_C = \text{Hilb}^N(\mathbb{C}^2).
\end{equation}
The ADHM theory with one flavour is self-dual under 3d mirror symmetry \cite{deBoer:1996mp,Porrati:1996xi} and the duality exchanges
\begin{equation}\label{eq:mirrormap}
\begin{split}
    \mathcal{M}_H &\leftrightarrow \mathcal{M}_C, \\
    \zeta = e^{-\xi} &\leftrightarrow z=e^{-m}, \\
    R_H &\leftrightarrow R_C.
\end{split}
\end{equation}
3d mirror symmetry is an IR duality of gauge theories \cite{Intriligator:1996ex} and as well as providing a duality between Higgs and Coulomb branch geometry, supersymmetric observables are also identified in mirror dual theories.


\section{Hemisphere Partition Function}\label{sec:vpfn}
In this section we introduce the main observable of interest in this work: the hemisphere partition function of the 3d ADHM theory. We construct a UV boundary condition for the hemisphere partition function on $S^1\times D$ that realises an exact factorisation of partition functions on various closed three-manifolds.  We then discuss the relationship between the partition function and the geometry of quasimaps to the Hilbert scheme $\text{Hilb}^N(\mathbb{C}^2)$ and demonstrate that certain specialised limits realise Verma characters of the ADHM Higgs and Coulomb quantised chiral rings.

\subsection{Boundary condition and localisation}\label{sec:hemispherelocalisation}
We first define the half superconformal index of the theory on  $\Omega$-deformed $\mathbb{R}_{\ge 0} \times \mathbb{R}^2$ in the presence of a boundary condition $\mathcal{B}$ on the plane $x^1 = 0$. As a trace over local operators, the half index is defined by
\begin{equation}
    \mathcal{I}(\mathcal{B}) = \text{Tr}\,(-)^F q^{J+\frac{R_H+R_C}{4}}t^{\frac{R_H-R_C}{2}}z^{F_H}\zeta^{F_C}\,,
\end{equation}
where $J$ is the generator of rotations in the $\mathbb{R}^2$ plane, and $F_H$ and $F_C$ are generators of $G_H$ and $G_C$ respectively. In this work we consider a set of $\mathcal{N}=(2,2)$ exceptional Dirichlet boundary conditions denoted $\mathcal{B}_{\lambda}$. These boundary conditions are associated to the massive vacua of the theory $\lambda$ and are expected to flow to thimble boundary conditions in the IR sigma model specified by Lagrangian submanifolds $\mathcal{L}_{\lambda}$ in $\text{Hilb}^N(\mathbb{C}^2)$. In this way, $\mathcal{B}_{\lambda}$ mimics a vacuum $\lambda$ at infinity. We give a more detailed overview of the main aspects of this setup in appendix \ref{appendix:localisation}. The boundary condition $\mathcal{B}_{\lambda}$ preserves a subset of the $\mathfrak{osp}(4|4)$ superconformal algebra generated by the supercharges: $Q_{+}^{1\dot{1}}, \enspace Q_{-}^{1 \dot{2}}, \enspace Q_{-}^{2 \dot{1}}, \enspace Q_{+}^{2\dot{2}}$ and their superconformal conjugates. Our conventions for the 3d $\mathcal{N}=4$ supersymmetry algebra match those of \cite{Bullimore:2020jdq}.

The state-operator correspondence relates the local operator count $\mathcal{I}(\mathcal{B_{\lambda}})$ to a partition function $\mathcal{Z}^{\lambda}_{S^1\times D}$ on the hemisphere with a boundary condition on $\partial(S^1\times D) = T^2$. The $\Omega$-deformation enters as an angular momentum refinement on the hemisphere\footnote{In practice the angular momentum refinement is implemented by a twisted boundary condition.} and we have:
\begin{equation}
    \mathcal{Z}^{\lambda}_{S^1 \times D}=e^{\phi_{\lambda}}\mathcal{I} (\mathcal{B}_{\lambda}).
\end{equation}
where $\phi_{\lambda}$ corresponds to the equivariant integral of the anomaly polynomial encoding boundary mixed 't Hooft anomalies, which are dependent on the boundary condition $\mathcal{B}_{\lambda}$. This term determines the Casimir energy of the vacuum \cite{Bobev:2015kza}. For the exceptional Dirichlet boundary condition it includes a central charge term $\kappa_{\lambda}$:
\begin{equation}
    \kappa_{\lambda}\,:\quad  \mathfrak{t}_H \times \mathfrak{t}_C \rightarrow \mathbb{R},
\end{equation} 
where $\mathfrak{t}_{H}$ and $\mathfrak{t}_C$ are the Lie algebras of the flavour symmetries.  The central charge is such that:
\begin{equation}
    \kappa_{\lambda}(m,\xi) = m \cdot \mu_{H,\mathbb{R}}(\lambda) = \xi \cdot \mu_{C,\mathbb{R}}(\lambda).
\end{equation}
where $\mu_{C,\mathbb{R}}$ is the moment map of the $G_C$ action on the Coulomb branch. This coincides with the value of the effective $G_H\times G_C$ mixed Chern-Simons coupling in the vacuum $\lambda$. As elucidated in \cite{Bullimore:2020jdq}, hemisphere partition functions with such boundary conditions can be interpreted as holomorphic blocks $\mathbb{B}_{\lambda}$. To our knowledge the present work is the first instance of a first principles definition and computation of the exact holomorphic block for a theory with $\mathcal{N}=4$ adjoint matter.

The hemisphere partition function can be computed using localisation and decomposes into classical, 1-loop and vortex contributions:
\begin{equation}
    \mathcal{Z}_{S^1 \times D}^{\lambda} = \mathcal{Z}_{\text{Classical}}^{\lambda}\mathcal{Z}_{\text{1-loop}}^{\lambda}\mathcal{Z}_{\text{Vortex}}^{\lambda} \equiv \mathbb{B}_{\lambda}(q,t), 
\end{equation}
where $\mathcal{Z}_{\text{Classical}}^{\lambda} = e^{\phi_{\lambda}}$ and $\mathcal{Z}_{\text{1-loop}}^{\lambda}\mathcal{Z}_{\text{Vortex}}^{\lambda}=\mathcal{I}(\mathcal{B}_{\lambda}) $.

\paragraph{Boundary condition}
We now specify the $\mathcal{N}=(2,2)$ boundary condition $\mathcal{B}_{\lambda}$. We use the 3d $\mathcal{N}=2$ language of \cite{Dimofte:2017tpi}, referring therein for detailed forms of the boundary conditions in the half-space picture, and to \cite{Bullimore:2020jdq} for the corresponding boundary conditions for $S^1\times D$. We prescribe a Dirichlet and Neumann boundary condition for the $\mathcal{N}=2$ vector and adjoint chiral multiplets comprising the $\mathcal{N}=4$ vector multiplet respectively. The Dirichlet boundary condition for the vector multiplet supports boundary monopole configurations resulting in a sum over abelian flux sectors in the index. 

It remains to specify the boundary condition for the hypermultiplets. We work in the basis of $\mathbb{C}^N$ given by Nakajima's fixed point description, where the indices on the linear data label boxes in the Young diagram specified by the vacuum. In this basis $\mathcal{B}_{\lambda}$ prescribes the boundary condition for the hypermultiplet scalars given in table \ref{table:exceptionalDirichletBCADHM}.
\begin{table}[h!]
\begin{center}
\begin{tabular}{ |c|c|c|c| } 
\hline
Boxes & $\mathcal{D}$ & $\mathcal{N}$\\
\hline
$\forall\,a,b\in\lambda \text{ s.t. }  (b \in \lambda^{B} )\cap( i_{a}> i_b) \text{ or } (b \notin \lambda^{B})$ & $A_{ab}$ & $B_{ba}$ \\ 
$\forall\,a,b\in\lambda\text{ s.t. }  (a \in \lambda^{B} )\cap( i_{b} \leq i_a)$ & $B_{ab}$ & $A_{ba}$ \\ 
$\forall a$ & $I_a$ & $J_a$ \\ 
\hline
\end{tabular}
\caption{Exceptional Dirichlet boundary condition associated to vacuum $\lambda$.}\label{table:exceptionalDirichletBCADHM}
\end{center}
\vspace{-5mm}
\end{table}
The scalars in $\mathcal{D}$ have their values fixed at the boundary, whilst those in $\mathcal{N}$ are allowed to fluctuate. The boundary conditions for the rest of the fields are fixed by supersymmetry.  $\lambda^B$ denotes the set of boxes on the bottom-most edge of the Young diagram, i.e. those $s\in\lambda$ such that $i_s = \lambda_{j_s}^{\vee}$. As part of the boundary condition data, we fix non-zero values for the following scalars in $\mathcal{D}$:
\begin{itemize}
    \item $I_{(1,1)}$ where (1,1) is the top-left box in $\lambda$.
    \item $A_{ab}$ whenever $i_a = i_b+1$ and $j_a = j_b$.
    \item $B_{ab}$ whenever $a\in \lambda^B$, $i_a = i_b$ and $j_a=j_b+1$.
\end{itemize}
We set these scalars equal to their values at the fixed point which we can normalise to $1$ by the action of $\prod_{a\in \lambda} GL(V_a)$. This breaks the $U(N)$ gauge symmetry at the boundary (itself already only a flavour symmetry due to the Dirichlet boundary condition for the vector multiplet) completely, whilst preserving the flavour and R-symmetries up to a boundary gauge transformation.  On the Young diagram this corresponds to a particular tree\footnote{We expect this to be closely related to the elliptic stable envelope for $\text{Hilb}^N(\mathbb{C}^2)$ recently constructed by Smirnov \cite{Smirnov:2018drm}, and hope to explore this connection in future work. Also see loc. cit. for a definition of a tree.}  $T_{\lambda}$, an example of which is shown in figure \ref{fig:treeforboundaryVEV}. The vertical arrows correspond to those components of $A$ fixed to non-zero values at the boundary, and horizontal arrows to those of $B$. In fact, any configuration of non-zero values corresponding to the edges in a tree would have the above property, however only $T_{\lambda}$ is compatible with table \ref{table:exceptionalDirichletBCADHM}.

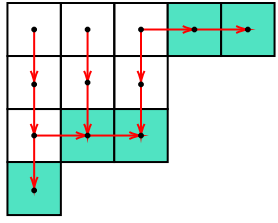
\begin{figure}
\centering
\tikzset{every picture/.style={line width=0.75pt}} 
\begin{tikzpicture}[x=0.4pt,y=0.4pt,yscale=-1,xscale=1]

\draw   (146,64) -- (196,64) -- (196,114) -- (146,114) -- cycle ;
\draw   (146,114) -- (196,114) -- (196,164) -- (146,164) -- cycle ;
\draw   (196,64) -- (246,64) -- (246,114) -- (196,114) -- cycle ;
\draw   (196,114) -- (246,114) -- (246,164) -- (196,164) -- cycle ;
\draw   (246,64) -- (296,64) -- (296,114) -- (246,114) -- cycle ;
\draw   (146,164) -- (196,164) -- (196,214) -- (146,214) -- cycle ;
\draw  [fill={rgb, 255:red, 80; green, 227; blue, 194 }  ,fill opacity=1 ] (146,214) -- (196,214) -- (196,264) -- (146,264) -- cycle ;
\draw  [fill={rgb, 255:red, 80; green, 227; blue, 194 }  ,fill opacity=1 ] (196,164) -- (246,164) -- (246,214) -- (196,214) -- cycle ;
\draw  [fill={rgb, 255:red, 80; green, 227; blue, 194 }  ,fill opacity=1 ] (246,164) -- (296,164) -- (296,214) -- (246,214) -- cycle ;
\draw   (246,114) -- (296,114) -- (296,164) -- (246,164) -- cycle ;
\draw  [fill={rgb, 255:red, 80; green, 227; blue, 194 }  ,fill opacity=1 ] (296,64) -- (346,64) -- (346,114) -- (296,114) -- cycle ;
\draw  [fill={rgb, 255:red, 80; green, 227; blue, 194 }  ,fill opacity=1 ] (346,64) -- (396,64) -- (396,114) -- (346,114) -- cycle ;
\draw [color={rgb, 255:red, 255; green, 0; blue, 0 }  ,draw opacity=1 ]   (171,89) -- (171,137) ;
\draw [shift={(171,139)}, rotate = 270] [color={rgb, 255:red, 255; green, 0; blue, 0 }  ,draw opacity=1 ][line width=0.75]    (10.93,-3.29) .. controls (6.95,-1.4) and (3.31,-0.3) .. (0,0) .. controls (3.31,0.3) and (6.95,1.4) .. (10.93,3.29)   ;
\draw [color={rgb, 255:red, 255; green, 0; blue, 0 }  ,draw opacity=1 ]   (171,139) -- (171,187) ;
\draw [shift={(171,189)}, rotate = 270] [color={rgb, 255:red, 255; green, 0; blue, 0 }  ,draw opacity=1 ][line width=0.75]    (10.93,-3.29) .. controls (6.95,-1.4) and (3.31,-0.3) .. (0,0) .. controls (3.31,0.3) and (6.95,1.4) .. (10.93,3.29)   ;
\draw [color={rgb, 255:red, 255; green, 0; blue, 0 }  ,draw opacity=1 ]   (171,189) -- (171,237) ;
\draw [shift={(171,239)}, rotate = 270] [color={rgb, 255:red, 255; green, 0; blue, 0 }  ,draw opacity=1 ][line width=0.75]    (10.93,-3.29) .. controls (6.95,-1.4) and (3.31,-0.3) .. (0,0) .. controls (3.31,0.3) and (6.95,1.4) .. (10.93,3.29)   ;
\draw [color={rgb, 255:red, 255; green, 0; blue, 0 }  ,draw opacity=1 ]   (221,89) -- (221,137) ;
\draw [shift={(221,139)}, rotate = 270] [color={rgb, 255:red, 255; green, 0; blue, 0 }  ,draw opacity=1 ][line width=0.75]    (10.93,-3.29) .. controls (6.95,-1.4) and (3.31,-0.3) .. (0,0) .. controls (3.31,0.3) and (6.95,1.4) .. (10.93,3.29)   ;
\draw [color={rgb, 255:red, 255; green, 0; blue, 0 }  ,draw opacity=1 ]   (221,139) -- (221,187) ;
\draw [shift={(221,189)}, rotate = 270] [color={rgb, 255:red, 255; green, 0; blue, 0 }  ,draw opacity=1 ][line width=0.75]    (10.93,-3.29) .. controls (6.95,-1.4) and (3.31,-0.3) .. (0,0) .. controls (3.31,0.3) and (6.95,1.4) .. (10.93,3.29)   ;
\draw [color={rgb, 255:red, 255; green, 0; blue, 0 }  ,draw opacity=1 ]   (171,189) -- (219,189) ;
\draw [shift={(221,189)}, rotate = 180] [color={rgb, 255:red, 255; green, 0; blue, 0 }  ,draw opacity=1 ][line width=0.75]    (10.93,-3.29) .. controls (6.95,-1.4) and (3.31,-0.3) .. (0,0) .. controls (3.31,0.3) and (6.95,1.4) .. (10.93,3.29)   ;
\draw [color={rgb, 255:red, 255; green, 0; blue, 0 }  ,draw opacity=1 ]   (221,189) -- (269,189) ;
\draw [shift={(271,189)}, rotate = 180] [color={rgb, 255:red, 255; green, 0; blue, 0 }  ,draw opacity=1 ][line width=0.75]    (10.93,-3.29) .. controls (6.95,-1.4) and (3.31,-0.3) .. (0,0) .. controls (3.31,0.3) and (6.95,1.4) .. (10.93,3.29)   ;
\draw [color={rgb, 255:red, 255; green, 0; blue, 0 }  ,draw opacity=1 ]   (271,89) -- (319,89) ;
\draw [shift={(321,89)}, rotate = 180] [color={rgb, 255:red, 255; green, 0; blue, 0 }  ,draw opacity=1 ][line width=0.75]    (10.93,-3.29) .. controls (6.95,-1.4) and (3.31,-0.3) .. (0,0) .. controls (3.31,0.3) and (6.95,1.4) .. (10.93,3.29)   ;
\draw [color={rgb, 255:red, 255; green, 0; blue, 0 }  ,draw opacity=1 ]   (321,89) -- (369,89) ;
\draw [shift={(371,89)}, rotate = 180] [color={rgb, 255:red, 255; green, 0; blue, 0 }  ,draw opacity=1 ][line width=0.75]    (10.93,-3.29) .. controls (6.95,-1.4) and (3.31,-0.3) .. (0,0) .. controls (3.31,0.3) and (6.95,1.4) .. (10.93,3.29)   ;
\draw [color={rgb, 255:red, 255; green, 0; blue, 0 }  ,draw opacity=1 ]   (271,89) -- (271,137) ;
\draw [shift={(271,139)}, rotate = 270] [color={rgb, 255:red, 255; green, 0; blue, 0 }  ,draw opacity=1 ][line width=0.75]    (10.93,-3.29) .. controls (6.95,-1.4) and (3.31,-0.3) .. (0,0) .. controls (3.31,0.3) and (6.95,1.4) .. (10.93,3.29)   ;
\draw [color={rgb, 255:red, 255; green, 0; blue, 0 }  ,draw opacity=1 ]   (271,139) -- (271,187) ;
\draw [shift={(271,189)}, rotate = 270] [color={rgb, 255:red, 255; green, 0; blue, 0 }  ,draw opacity=1 ][line width=0.75]    (10.93,-3.29) .. controls (6.95,-1.4) and (3.31,-0.3) .. (0,0) .. controls (3.31,0.3) and (6.95,1.4) .. (10.93,3.29)   ;
\draw  [fill={rgb, 255:red, 0; green, 0; blue, 0 }  ,fill opacity=1 ] (169.25,89) .. controls (169.25,88.03) and (170.03,87.25) .. (171,87.25) .. controls (171.97,87.25) and (172.75,88.03) .. (172.75,89) .. controls (172.75,89.97) and (171.97,90.75) .. (171,90.75) .. controls (170.03,90.75) and (169.25,89.97) .. (169.25,89) -- cycle ;
\draw  [fill={rgb, 255:red, 0; green, 0; blue, 0 }  ,fill opacity=1 ] (169.25,140.75) .. controls (169.25,139.78) and (170.03,139) .. (171,139) .. controls (171.97,139) and (172.75,139.78) .. (172.75,140.75) .. controls (172.75,141.72) and (171.97,142.5) .. (171,142.5) .. controls (170.03,142.5) and (169.25,141.72) .. (169.25,140.75) -- cycle ;
\draw  [fill={rgb, 255:red, 0; green, 0; blue, 0 }  ,fill opacity=1 ] (169.25,189) .. controls (169.25,188.03) and (170.03,187.25) .. (171,187.25) .. controls (171.97,187.25) and (172.75,188.03) .. (172.75,189) .. controls (172.75,189.97) and (171.97,190.75) .. (171,190.75) .. controls (170.03,190.75) and (169.25,189.97) .. (169.25,189) -- cycle ;
\draw  [fill={rgb, 255:red, 0; green, 0; blue, 0 }  ,fill opacity=1 ] (169.25,240.75) .. controls (169.25,239.78) and (170.03,239) .. (171,239) .. controls (171.97,239) and (172.75,239.78) .. (172.75,240.75) .. controls (172.75,241.72) and (171.97,242.5) .. (171,242.5) .. controls (170.03,242.5) and (169.25,241.72) .. (169.25,240.75) -- cycle ;
\draw  [fill={rgb, 255:red, 0; green, 0; blue, 0 }  ,fill opacity=1 ] (219.25,89) .. controls (219.25,88.03) and (220.03,87.25) .. (221,87.25) .. controls (221.97,87.25) and (222.75,88.03) .. (222.75,89) .. controls (222.75,89.97) and (221.97,90.75) .. (221,90.75) .. controls (220.03,90.75) and (219.25,89.97) .. (219.25,89) -- cycle ;
\draw  [fill={rgb, 255:red, 0; green, 0; blue, 0 }  ,fill opacity=1 ] (219.25,189) .. controls (219.25,188.03) and (220.03,187.25) .. (221,187.25) .. controls (221.97,187.25) and (222.75,188.03) .. (222.75,189) .. controls (222.75,189.97) and (221.97,190.75) .. (221,190.75) .. controls (220.03,190.75) and (219.25,189.97) .. (219.25,189) -- cycle ;
\draw  [fill={rgb, 255:red, 0; green, 0; blue, 0 }  ,fill opacity=1 ] (219.25,139) .. controls (219.25,138.03) and (220.03,137.25) .. (221,137.25) .. controls (221.97,137.25) and (222.75,138.03) .. (222.75,139) .. controls (222.75,139.97) and (221.97,140.75) .. (221,140.75) .. controls (220.03,140.75) and (219.25,139.97) .. (219.25,139) -- cycle ;
\draw  [fill={rgb, 255:red, 0; green, 0; blue, 0 }  ,fill opacity=1 ] (269.25,189) .. controls (269.25,188.03) and (270.03,187.25) .. (271,187.25) .. controls (271.97,187.25) and (272.75,188.03) .. (272.75,189) .. controls (272.75,189.97) and (271.97,190.75) .. (271,190.75) .. controls (270.03,190.75) and (269.25,189.97) .. (269.25,189) -- cycle ;
\draw  [fill={rgb, 255:red, 0; green, 0; blue, 0 }  ,fill opacity=1 ] (269.25,140.75) .. controls (269.25,139.78) and (270.03,139) .. (271,139) .. controls (271.97,139) and (272.75,139.78) .. (272.75,140.75) .. controls (272.75,141.72) and (271.97,142.5) .. (271,142.5) .. controls (270.03,142.5) and (269.25,141.72) .. (269.25,140.75) -- cycle ;
\draw  [fill={rgb, 255:red, 0; green, 0; blue, 0 }  ,fill opacity=1 ] (269.25,89) .. controls (269.25,88.03) and (270.03,87.25) .. (271,87.25) .. controls (271.97,87.25) and (272.75,88.03) .. (272.75,89) .. controls (272.75,89.97) and (271.97,90.75) .. (271,90.75) .. controls (270.03,90.75) and (269.25,89.97) .. (269.25,89) -- cycle ;
\draw  [fill={rgb, 255:red, 0; green, 0; blue, 0 }  ,fill opacity=1 ] (319.25,89) .. controls (319.25,88.03) and (320.03,87.25) .. (321,87.25) .. controls (321.97,87.25) and (322.75,88.03) .. (322.75,89) .. controls (322.75,89.97) and (321.97,90.75) .. (321,90.75) .. controls (320.03,90.75) and (319.25,89.97) .. (319.25,89) -- cycle ;
\draw  [fill={rgb, 255:red, 0; green, 0; blue, 0 }  ,fill opacity=1 ] (369.25,89) .. controls (369.25,88.03) and (370.03,87.25) .. (371,87.25) .. controls (371.97,87.25) and (372.75,88.03) .. (372.75,89) .. controls (372.75,89.97) and (371.97,90.75) .. (371,90.75) .. controls (370.03,90.75) and (369.25,89.97) .. (369.25,89) -- cycle ;
\end{tikzpicture}
\caption{A tree permitted by the boundary condition corresponding to non-zero expectation values at the boundary. The boxes in $\lambda^B$ are highlighted.}\label{fig:treeforboundaryVEV}
\end{figure}

\paragraph{Higgs branch image} We now motivate this choice of $\mathcal{B}_{\lambda}$. Abusing notation, we also denote by $\mathcal{B}_{\lambda}$ the holomorphic Lagrangian of the affine space parametrised by the scalars $(A,B,I,J)$ with $\mathcal{N}$ boundary conditions. Explicitly:
\begin{equation}
    \mathcal{B}_{\lambda}=  \left\{ (A, B,I,J) \,\left|\,\,
    \begin{aligned}
     &A_{ab} = c_{1,ab} \text{ if } (b \in \lambda^{B} )\cap( i_{a}> i_b) \text{ or } (b \notin \lambda^{B}),\\
     &B_{ab}=  c_{2,ab}\text{ if }  (a \in \lambda^{B} )\cap( i_{b} \leq i_a),\\
     &I_a = \delta_{1,a} 
    \end{aligned}\right.\right\}
\end{equation}
where:
\begin{equation}\label{eq:matrixofconstants}
    c_{i,ab} = \begin{cases} 1 \quad \text{ if $a,b \in \lambda$ linked by an edge in } T_{\lambda}, \\ 0 \quad \text{ otherwise.}
    \end{cases}
\end{equation}
The image of this boundary condition in the bulk Higgs branch is given by:
\begin{equation}
    \mathcal{B}^{(\text{Bulk})}_{\lambda} = \left[\mathcal{B}_{\lambda}\cap \mu_{\mathbb{C}}^{-1}(0) \right]/GL(N,\mathbb{C})\,.
\end{equation}
This contains the fixed point $\lambda$ since the only scalars which are non-zero on the vacuum which are \textit{not} fixed to their values at the vacuum by the boundary condition are $B_{ab}$ with $a,b \in \lambda$ adjacent, with $a$ to the right of $b$, and $a\notin \lambda^B$. These scalars are in $\mathcal{N}$ so are free to take their vacuum values.

The proposal in \cite{Bullimore:2016nji} for such boundary conditions to correspond to thimbles for vacua, and the bulk-boundary system to yield Verma modules of the chiral ring require the following conditions to be met. Firstly, since the bulk-boundary system has vacua $\mathcal{B}_{\lambda}\cap \mu_{\mathbb{C}}^{-1}(0)$, there must be no non-trivial $GL(N,\mathbb{C})$ orbits in $\mathcal{B}^{(\text{Bulk})}_{\lambda}$, else there would be additional non-compact 2d degrees of freedom on the boundary. That is: $\mathcal{B}^{(\text{Bulk})}_{\lambda} = \mathcal{B}_{\lambda}\cap \mu_{\mathbb{C}}^{-1}(0)$. Secondly, we should have:
\begin{equation}
    \mathcal{B}_{\lambda}\cap \mu_{\mathbb{C}}^{-1}(0) = \mathcal{L}_{\lambda}.
\end{equation}
In appendix \ref{appendix:geometricboundarycondition}, we show the local version of these statements i.e. that there are no non-trivial gauge orbits in a neighbourhood of $\lambda$ in $\mathcal{B}^{(\text{Bulk})}_{\lambda}$ and that the tangent space to $\lambda$ in $\mathcal{B}^{(\text{Bulk})}_{\lambda}$ consists of the half-dimensional subspace of $T_{\lambda}\mathcal{M}_H$ with positive $z$-weights. More precisely, Nakajima \cite{nakajima1999lectures} computes the tangent space character at a fixed point $\lambda$:
\begin{equation}
\begin{split}
    T_{\lambda}\mathcal{M}_H = \sum_{s\in\lambda} z^{a_{\lambda}(s)+l_{\lambda}(s)+1} t^{\frac{1}{2}(-a_{\lambda}(s)+l_{\lambda}(s)+1)} + z^{-a_{\lambda}(s)-l_{\lambda}(s)-1} t^{\frac{1}{2}(a_{\lambda}(s)-l_{\lambda}(s)+1)}.
\end{split}
\end{equation}
This is a holomorphic Lagrangian splitting into positive and negative weight spaces for $z$ since the holomorphic symplectic form $\omega_{\mathbb
C}$ is invariant under $G_H$. For our choice of boundary condition we find:
\begin{equation}
    T_{\lambda}\mathcal{B}_{\lambda}^{\text{Bulk}} = \sum_{s\in\lambda} z^{a_{\lambda}(s)+l_{\lambda}(s)+1} t^{\frac{1}{2}(-a_{\lambda}(s)+l_{\lambda}(s)+1)} =  T_{\lambda}\mathcal{L}_{\lambda}.
\end{equation}
which is precisely the character of the positive weight space.

\paragraph{The partition function}
We now compute the hemisphere partition function for this boundary condition following \cite{Dimofte:2017tpi,Bullimore:2020jdq}. The first step is to write the partition function for $c=0$, and then deform with $c$ as in (\ref{eq:matrixofconstants}) by setting to $1$ the combination of fugacities dual to the charge of the chirals which acquire non-zero vevs. Using the 1-loop determinants and classical contributions from the localisation computation in \cite{Bullimore:2020jdq} together with the boundary conditions in table \ref{table:exceptionalDirichletBCADHM} we find, before deformation:
\begin{equation}\label{eq:beforedeformation}
\begin{split}
    \tilde{\mathcal{Z}}^{ \lambda}_{S^1\times D} &= 
    \sum_{k\in \mathbb{Z}^N} e^{\frac{\log\zeta}{\log q}\left(\sum_{a\in\lambda} \log\left(s_aq^{k_a}\right)\right)} \prod_{a,b\,\in \lambda}  \frac{\left(u^2\frac{s_a}{s_b}q^{k_a-k_b};q\right)'_{\infty}}{\left(q\frac{s_a}{s_b}q^{k_a-k_b};q\right)'_{\infty}}
    \prod_{a\in \lambda }  \frac{\left(q u^{-1} s_a q^{k_a};q\right)'_{\infty}}{\left(us_a q^{k_a};q\right)'_{\infty}}\\
    &\prod_{\substack{a,b\in \lambda \text{ s.t. } \\  (b \in \lambda^{B} )\cap( i_{a}> i_b) \\ \text{ or } (b \notin \lambda^{B})}}  \frac{\left(qz^{-1}u^{-1}\frac{s_a}{s_b}q^{k_a-k_b};q\right)'_{\infty}}{\left(z^{-1}u\frac{s_a}{s_b}q^{k_a-k_b};q\right)'_{\infty}}
    \prod_{\substack{a,b\in \lambda \text{ s.t. } \\ (a \in \lambda^{B} )\cap( i_{b} \leq i_a)}}  \frac{\left(qzu^{-1}\frac{s_a}{s_b}q^{k_a-k_b};q\right)'_{\infty}}{\left(zu\frac{s_a}{s_b}q^{k_a-k_b};q\right)'_{\infty}}\\
\end{split}
\end{equation}
where $z,t,\zeta$ are fugacities defined in section \ref{sec:background}. We define $u=t^{\frac{1}{2}}q^{\frac{1}{4}}$, $v=t^{-\frac{1}{2}}q^{\frac{1}{4}}$ as the $U(1)_{H,C}$ R-symmetry fugacities respectively and $s_a^{-1}$ as the fugacity for the abelian factor rotating the eigenspace $V_a \subset V$. The sum is over $k\in \text{cochar}(U(N)) \simeq \mathbb{Z}^N$ i.e. the abelian flux sectors for boundary monopole operators. The function $(a;q)'$ is related to the usual q-Pochhammer function (\ref{eq:qpochhammer}) by:
\begin{equation}
    (a;q)'_{\infty} = e^{-\mathcal{E}[-\log(a)]} (a;q)_{\infty}.
\end{equation}
The factors $\mathcal{E}[x]$ arise as a zeta function regularisation of zero point energies, explicitly:
\begin{equation}
    \mathcal{E}[x]=\frac{\beta}{12}-\frac{x}{4}+\frac{x^2}{8\beta}
\end{equation}
with $q=e^{-2\beta}$. These factors are a crucial ingredient in \cite{Bullimore:2020jdq} to obtain the correct weight of the vacuum state.

Now we deform to the exceptional Dirichlet boundary condition specified above. The chirals with non-zero values at the boundary are the components of $(A,B,I)$ given by the particular tree $T_{\lambda}$. This prescription completely breaks the boundary gauge group (which was already broken to a flavour symmetry by the boundary condition). Specifically, this sets:
\begin{equation}
\begin{split}
    s_{(1,1)} = u\,,\quad s_a = u (z u)^{i_a-1} (z^{-1} u )^{j_a-1} \equiv u^{-1} v_{a}
\end{split}
\end{equation}
and gives the hemisphere partition function $\mathcal{Z}_{S^1\times D}^{\lambda}$. We note that these gauge fugacity values are precisely the Grothendieck roots appearing in \cite{Dinkins:2019pwj} in their computation of the vortex partition function. There are intricate combinatoric cancellations in the perturbative part of the partition function which we describe in appendix \ref{appendix:holoblockdetailedcomputation}. For now we state the final result
\begin{equation}\label{eq:ADHMblock}
    \mathcal{Z}^{\lambda}_{S^1\times D} =  \mathcal{Z}^{\lambda}_{\text{Classical}}  \mathcal{Z}^{\lambda}_{\text{1-loop}} \mathcal{Z}^{\lambda}_{\text{Vortex}}
\end{equation}
where the classical contribution is given by:
\begin{equation}
    \mathcal{Z}^{\lambda}_{\text{Classical}} = e^{-\left[\sum_{s\in\lambda}c(s) \right]\frac{\log\zeta\log z}{\log q}}
    e^{\left[\sum_{s\in\lambda}h(s) \right]\frac{\log v \log z}{\log q}}
    e^{\left[\sum_{s\in\lambda}h(s) \right]\frac{\log u \log \zeta}{\log q}}
    e^{-\left[\sum_{s\in\lambda}c(s) \right]\frac{\log u \log v}{\log q}},\nonumber
\end{equation}
and we again we use the shorthand $u=t^{\frac{1}{2}}q^{\frac{1}{4}}$ and $v=t^{-\frac{1}{2}}q^{\frac{1}{4}}$. We note that, in line with the proposals in \cite{Bullimore:2020jdq,Bobev:2015kza}, the classical piece is precisely the equivariant integral of the boundary 't Hooft anomaly, and includes the mixed central charge term $-\log\zeta\log z \sum_{s\in\lambda}c_{\lambda}(s)$. For the 1-loop piece we find:
\begin{equation}\label{eq:oneloop}
    \mathcal{Z}_{\text{1-loop}}^{\lambda} = \prod_{s\in\lambda}\frac{(q z^{a_{\lambda}(s)+l_{\lambda}(s)+1}u^{-a_{\lambda}(s)+l_{\lambda}(s)-1};q)_{\infty}}{(z^{a_{\lambda}(s)+l_{\lambda}(s)+1}u^{-a_{\lambda}(s)+l_{\lambda}(s)+1};q)_{\infty}},
\end{equation}
and the vortex contributions are given by:
\begin{equation}\label{eq:vortexsum}
    \mathcal{Z}^{\lambda}_{\text{Vortex}} = \sum_{\pi \in \text{RPP}(\lambda)} \left( \zeta t^{\frac{1}{2}}q^{-\frac{1}{4}}\right)^{|\pi|} \prod_{s\in \lambda} \frac{\left(u^2 v_s^{-1};q\right)_{-\pi_s}}{\left(q v_s^{-1};q\right)_{-\pi_s}} \prod_{\substack{s,t\in \lambda \\ s \neq t}} \frac{\left(qu^{-2}\frac{v_t}{v_s};q\right)_{\pi_t-\pi_s}}{\left(\frac{v_t}{v_s};q\right)_{\pi_t-\pi_s}}\frac{\left(zu\frac{v_t}{v_s};q\right)_{\pi_t-\pi_s}}{\left(qzu^{-1}\frac{v_t}{v_s};q\right)_{\pi_t-\pi_s}}.
\end{equation}
where in the above we set $v_s=z^{i_s-j_s}u^{i_s+j_s}$. \textit{A priori}, the vortex sum should be taken over all integers $\{k_s\}$ as in (\ref{eq:beforedeformation}) however, as shown in appendix \ref{appendix:holoblockdetailedcomputation}, the summand vanishes unless $\{k_s\}$ form a reverse plane partition (RPP) -- we thus write the vortex sum in terms of $\pi \in \text{RPP}(\lambda)$.

Our hemisphere partition function provides a first principles UV derivation of the vortex contributions to the holomorphic block as recently derived via factorisation in \cite{Choi:2019zpz}. The vortex partition function also coincides with the quasimap index of \cite{Dinkins:2019pwj} and we discuss this geometric connection in detail in the following section. The present work differs in approach from \cite{Dinkins:2019pwj} in that it is important for us to include perturbative contributions, as these terms are crucial to understanding the representation theory of quantised Coulomb and Higgs branch algebras and the \emph{exact} factorisation of partition functions on closed three manifolds. The `boundary condition' implicitly used in \cite{Dinkins:2019pwj} assigns all of $A$ Dirichlet and all of $B$ Neumann, and the vortex partition function is normalised by dividing by terms with simple poles. In contrast, our perturbative piece is finite. The vortex contribution is insensitive to the choice of polarisation provided the same Grothendieck roots are used, although such a computation does not manifestly have a physical interpretation in the gauge theory. Instead in our work, the Grothendieck roots are recovered via the procedure for computing the exceptional Dirichlet boundary condition. We leave to future work \cite{hunterpending} an investigation of the geometric interpretation of the perturbative contributions.

\paragraph{Superconformal index}
We note briefly that the block derived above fuses \textit{exactly} to the superconformal index of the ADHM theory computed in \cite{Choi:2019zpz}.  There the vortex partition function was obtained by factorising the $S^1\times S^2$ index, but only the block corresponding to the column partition was derived from first principles via a localisation calculation on $S^1 \times D$. Our work completes the derivation of the complete set of blocks for this theory. The 1-loop perturbative contribution to the superconformal index also undergoes drastic cancellations analogous to those in the hemisphere partition function (the details are given in \ref{appendix:holoblockdetailedcomputation}), and we obtain:
\begin{equation}
    \mathcal{Z}^{\text{S.C.}}_{S^1\times S^2} = \sum_{\lambda} \left\lVert\mathcal{Z}^{\lambda}_{S^1\times D}\right\rVert_{\text{S.C.}}^2
\end{equation}
where the sum is taken over partitions of weight $N$ and the gluing is given by:
\begin{equation}\label{eq:SCindexgluing}
    q\rightarrow q^{-1},\quad t\rightarrow t^{-1},  \quad z\rightarrow z^{-1}, \quad \zeta \rightarrow \zeta^{-1}  .  
\end{equation}
Fluxes for flavour symmetries through $S^2$ can be included easily, via shifting fugacities appearing in the blocks \cite{Beem:2012mb}, although we omit this for the purposes of brevity. The factorisation remains exact in the presence of background flux.

\paragraph{$A$- and $B$-shifts}
We define the $A$- and the $B$-shifted hemisphere partition functions with an R-symmetry parameter redefinition $t\to q^{\pm\frac{1}{2}}t$. 
\begin{equation}\label{eq:ABtwists}
    \mathcal{Z}_{S^1\times D}^{\lambda}(t) = \mathcal{Z}_{S^1\times D}^{A,\lambda}(tq^{-\frac{1}{2}}) = \mathcal{Z}_{S^1\times D}^{B,\lambda}(tq^{\frac{1}{2}})\,.
\end{equation}
We make use of these shifted partition functions later when studying Verma modules and factorising the $A$- and $B$- twisted indices. Shifted indices $\mathcal{I}^A(\mathcal{B}_{\lambda})$ and $\mathcal{I}^B(\mathcal{B}_{\lambda})$ are defined similarly, as are the separate perturbative and non-perturbative contributions.

The mirror self-duality exchanges $R_H \leftrightarrow R_C$ and thus exchanges the $A$ and $B$ shifts. Whilst partition functions on closed three-manifolds are preserved under the mirror map, it is expected that 3d mirror symmetry acts non-trivially on the boundary conditions $\mathcal{B}_{\lambda}$ so that the holomorphic blocks transform linearly amongst themselves under mirror symmetry.
In this work we are concerned with two specialised limits where the mirror map is particularly tractable and we explicitly verify that the 1-loop and vortex contributions in these limits are exchanged by mirror symmetry.

\subsection{Interpretation as a quasimap index}\label{subsec:quasimaps}
The vortex contribution to the hemisphere partition function in a particular vacuum can be interpreted as the equivariant K-theoretic vertex function of the Jordan quiver variety \cite{Dinkins:2019pwj,Okounkov:2015spn} evaluated at a fixed point $\lambda \in \text{Hilb}^N(\mathbb{C}^2)$.

\paragraph{ADHM vortices}
The moduli space of vortices in the 3d ADHM theory can be identified with the moduli space $\textsf{QM}^{\boldsymbol{d}}_{\lambda}$ of quasimaps $f: \mathbb{P}^1 \to \mathcal{M}_H=\text{Hilb}^N(\mathbb{C}^2)$ with degree $\boldsymbol{d}\in \mathbb{Z}^N$ and $f(\infty) = \lambda$. The space admits a $\mathbb{C}^*$ action of  rotations of the domain $\mathbb{P}^1$ and a maximal torus $T^2$ action on $\text{Hilb}^N(\mathbb{C}^2)$. The space of quasimaps to the Hilbert scheme of points is studied from the perspective of quantum cohomology in \cite{CIOCANFONTANINE2012268} and in K-theory \cite{Dinkins:2019pwj,Smirnov:2016vaw,Okounkov:2016sya}.

The vortex contributions to the $A$-shifted partition function\footnote{We could equally well use the $B$-shift here. Generally, the $B$-shift should correspond to quasimaps to the Coulomb branch $\mathcal{M}_C$ but since the theory we consider is mirror self-dual the difference is merely a parameter redefinition.} coincide with the equivariant Euler characteristic of the symmetrised virtual structure sheaf $\hat{\mathcal{O}}^{\boldsymbol{d}}_{\text{Vir.}}$ on $\textsf{QM}^{\boldsymbol{d}}_{\lambda}$:
\begin{equation}\label{eq:eulercharacteristics}
\mathcal{Z}_{\text{Vortex}}^{A,\lambda} = \sum_{\boldsymbol{d}} \left(\zeta t^{\frac{1}{2}}\right)^{\boldsymbol{d}}\chi\left(\hat{\mathcal{O}}^{\boldsymbol{d}}_{\text{Vir.}}\right)\,.
\end{equation}
The Euler characteristic can be computed in localised K-theory as the vertex function evaluated at $\lambda \in \text{Hilb}^N(\mathbb{C}^2)$ in the fixed point basis of $K_{T^2}(\text{Hilb}^N(\mathbb{C}^2))$. This yields an expression:\footnote{$\hat{a}(x) = x^{\frac{1}{2}}-x^{-\frac{1}{2}}$ and is extended multiplicatively on weights $\hat{a}(x+y)=\hat{a}(x)\hat{a}(y)$.}
\begin{equation}
    \mathcal{Z}_{\text{Vortex}}^{A,\lambda} = \sum_{\boldsymbol{d}}\left( \zeta t^{\frac{1}{2}} \right)^{\boldsymbol{d}}\sum_{\text{f.p.} \in \textsf{QM}^{\mathbb{C}^*\times T^2}}  \hat{a}\left(T_{\text{f.p.}}^{\text{Vir.}} \textsf{QM}^{\boldsymbol{d}}_{\lambda}\right)\,.
\end{equation}
Consequently, the reverse plane partitions $\pi$ in the vortex sum (\ref{eq:vortexsum}) coincide with fixed points on the quasimap moduli space with fixed degree $\boldsymbol{d}=|\pi|$ and the summand is the character of the virtual tangent space. The gauge theory and geometry parameters are related by:
\begin{center}
\begin{tabular}{ c c  }
$\mathcal{N}=2^*$ mass, $t$ & Fugacity for $\mathbb{C}^* \subset T^2$ \\ 
 Flavour symmetry, $z$ & Fugacity for $\mathbb{C}^* \subset T^2$ \\
 Angular momentum refinement, $q$ & Rotations of $\mathbb{P}^1$ \\
 Vortex number & Quasimap degree $\boldsymbol{d}$
\end{tabular}
\end{center}
\paragraph{K-theoretic PT vertex}
Quasimaps to the Hilbert scheme of $N$ points in $\mathbb{C}^2$ correspond to points in the K-theoretic PT moduli space of $\mathbb{C}^3$ \cite{Okounkov:2018huu,Okounkov:2016sya}. The equivariant Euler characteristics (\ref{eq:eulercharacteristics}) are then identified with the bare 1-leg vertex and following the notation of e.g. \cite{Kononov:2019fni} we have:
\begin{equation}
    \mathcal{Z}^{A,\lambda}_{\text{Vortex}} = V_{\text{PT}}^{\lambda,\emptyset,\emptyset}
\end{equation}
In following subsections we verify this correspondence in two cases. We first study the Poincar\'e polynomial limit of the vertex in section \ref{subsec:pplimit} and show that we recover the refined topological vertex \cite{Iqbal:2007ii}. Later, in section \ref{subsec:hilbertseries}, we show that the holomorphic block gluing for the twisted index can be identified with the vertex gluing in PT theory and indeed, taking the generating function over the gauge rank, we recover the partition function of the resolved conifold i.e. the generating function over $N$ of the Hilbert series of $\text{Hilb}^N(\mathbb{C}^2)$. This correspondence has also been studied from the 2d/homological perspective in the works \cite{Bonelli:2013mma,Bonelli:2013rja}.

\subsection{Verma character limit}\label{subsec:vermacharacterlimit}
We now consider two specialised limits of the hemisphere partition function. The first limit we consider is denoted the $A$-limit and corresponds to setting $t \to 1$ in the $A$-twisted hemisphere partition function. In the $A$-limit the generators in the index commute with both $Q_{-}^{1 \dot{2}}$ and $Q_{+}^{1\dot{1}}$, the index is independent of $q$ and the full $\mathcal{N}=4$ supersymmetry is restored. We have simply:
\begin{equation}
    \lim_{t\to1} I^A(\mathcal{B}_{\lambda}) = \text{Tr}_{\mathcal{H}_A} \zeta^{F_C}\,.
\end{equation}
The index receives contributions from boundary Coulomb branch bosonic operators in $\mathcal{H}_A$. In the $\Omega$ background, these operators admit an action of the bulk quantised Coulomb branch algebra \cite{Bullimore:2016nji} and the boundary condition $\mathcal{B}_{\lambda}$ leads to a Verma character for this action. The details of this limit were recently discussed in general by the authors in \cite{Bullimore:2020jdq} and in this work we compute these limits for the 3d ADHM theory in particular.

In the $A$-limit only the vortex contributions survive and we have:
\begin{equation}
    \lim_{t\to1}I^A(\mathcal{B}_{\lambda}) = \lim_{t\to 1}\mathcal{Z}_{\text{Vortex}}^{A,\lambda} = \sum_{\pi \in \text{RPP}(\lambda)} \zeta^{|\pi|}\,.
\end{equation}
Geometrically, this is an un-graded count of the fixed points on $\textsf{QM}^{\boldsymbol{d}}_{\lambda}$ and is thus consistent with the expectation that the equivariant homology $\sum_{\boldsymbol{d}} H^\bullet \left( \textsf{QM}^{\boldsymbol{d}}_{\lambda}\right)$, in the fixed point basis, forms a Verma module of the Coulomb branch chiral ring \cite{Bullimore:2016hdc}.

The mirror limit, which we denote the $B$-limit, corresponds to sending $t\to1$ in the $B$-shifted index, a similar supersymmetry enhancement occurs but now only the 1-loop excitations contribute to the index:
\begin{equation}
    \lim_{t\to1}I^B(\mathcal{B}_{\lambda}) = \lim_{t\to 1}\mathcal{Z}_{\text{1-loop}}^{B,\lambda} = \prod_{s\in\lambda} \frac{1}{1-z^{a_{\lambda}(s)+l_{\lambda}(s)+1}}\,.
\end{equation}
3d mirror symmetry exchanges the two limits (provided we also exchange $\zeta$ and $z$ as in the mirror map (\ref{eq:mirrormap})) and leads to the familiar generating function identity for reverse plane partitions with the 1-loop piece degenerating to the hook formula:
\begin{equation}
    \lim_{t\to1}\mathcal{Z}_{\text{1-loop}}^{B,\lambda}=\prod_{s\in\lambda} \frac{1}{1-z^{h_{\lambda}(s)}} = \sum_{\pi \in \text{RPP}(\lambda)} z^{|\pi|} = \lim_{t\to1}\mathcal{Z}_{\text{Vortex}}^{A,\lambda}\,.
\end{equation}

\paragraph{Coulomb branch algebra}
The quantum Coulomb branch algebra $\hat{\mathbb{C}}[\mathcal{M}_C]$ of the 3d ADHM theory with one flavour, hereafter denoted $\mathcal{A}^C_N$, depends on the mass parameter $m$ explicitly in the generators and the $\Omega$-deformation parameter $q$ via the quantisation. The monopole operators are built from variables $w_{k,a}$, $v_{k,a}$ and constructed as follows:
\begin{equation}
\begin{split}
E_{k,t}&=\sum_{a=0}^{N_{k}}
\frac{\prod_{b} w_{k,a}-w_{k-1,b}-m}
{\prod_{b\neq a}w_{k,a}-w_{k,b}}w_{k,a}^{t}v_{k,a},\\
F_{k,t}&=\sum_{a=0}^{N_{k}}
\frac{\prod_{b} w_{k,a}-w_{k+1,b}+m}{\prod_{b\neq a}w_{k,a}-w_{k,b}}
v_{k,a}^{-1}w_{k,a}^{t+\delta_{0,k}},
\end{split}
\end{equation}
for a decomposition $\sum_{k}N_{k}=N$. The details of this algebra and the abstract construction of the Verma modules $V_\lambda$, labelled by a vacuum $\lambda$, of $\mathcal{A}^C_N$ were recently studied in \cite{Gaiotto:2019wcc}. The character of such modules can be written:
\begin{equation}
\begin{split}
\chi^{C,\lambda}(z,\zeta;q)
&={\Tr}_{V_{\lambda}} \zeta^{\frac{1}{2 \beta}\sum w_{k,a}}\,.
\end{split}
\end{equation}
The Verma characters were computed with slightly different parameter conventions in \cite{Gaiotto:2019mmf}. As discussed in \cite{Bullimore:2020jdq}, these Verma characters coincide with the $t\to1$ limit of the full hemisphere partition function in the $A$-shift convention. The classical term determines the Casimir energy and thus the highest weight of the vacuum. In the 3d ADHM case we find:
\begin{equation}\label{eq:vermacharacter}
    \lim_{t\to1}\mathcal{Z}^{A,\lambda}_{S^1\times D} = \mathcal{Z}_{\text{Cl.}}^{A,\lambda}\mathcal{Z}_{\text{Vortex}}^{A,\lambda} = e^{\frac{\xi m}{2 \beta}\sum_{s\in\lambda}c_{\lambda}(s)+\frac{\xi}{2}\sum_{s\in\lambda}h_{\lambda}(s)}\sum_{\pi \in \text{RPP}(\lambda)} \zeta^{|\pi|} = \chi^{C,\lambda}\,.
\end{equation}
Since the theory is mirror self-dual, the Verma characters of the Higgs branch algebra $\chi^{H,\lambda}$ are expected to be functionally the same form as $\chi^{C,\lambda}$ with $\xi$ and $m$ interchanged.

\paragraph{Specialised $S^3$ partition function}
We now take a brief detour to the 3-sphere partition function and explain the connection between holomorphic factorisation and the recently proposed ``IR formulae'' of \cite{Gaiotto:2019mmf}. The partition function $\mathcal{Z}_{S^3_b}$ of the theory is defined on the squashed ellipsoid:
\begin{equation}
    b^2|z_1|^2 + \frac{1}{b^2}|z_2|^2 = 1\,.
\end{equation}
This partition function can be computed by Coulomb branch localisation \cite{Pasquetti:2011fj} and factorised into holomorphic blocks:
\begin{equation}
    \mathcal{Z}_{S^3} = \sum_{\alpha} \mathbb{B}_\alpha(z,\zeta;q,t) \mathbb{B}_{\alpha}(\bar{z},\bar{\zeta};\bar{q},\bar{t})
\end{equation}
where the parameter identifications are as follows:
\begin{align}
     q&=e^{-2\pi i b Q},
     &\qquad  t&=e^{2\pi b T},
     &\qquad   z&= e^{-2\pi b m},
     &\qquad \zeta &= e^{-2\pi b \xi},   \\
     \bar{q}&=e^{-\frac{2\pi i}{b}Q},
     &\qquad \bar{t}&=e^{\frac{2\pi T}{b}},
     &\qquad \bar{z} &= e^{-\frac{2\pi m}{b}},
     &\qquad  \bar{\zeta} &= e^{-\frac{2\pi \xi}{b}}
\end{align}
with $Q = b + 1/b$. The specialised $S^3$ partition function enhances the supersymmetry from $\mathcal{N}=2^*$ to $\mathcal{N}=4$ and corresponds to setting $T\to \frac{i}{2}(b-1/b)$. In this limit the partition function becomes:
\begin{equation}
    \mathcal{Z}_{S^3_b} = \sum_{\lambda} e^{\pi m \xi \sum_{s\in\lambda}c_{\lambda}(s)} e^{\pi(m+\xi) \sum_{s\in \lambda}h_{\lambda}(s)}\sum_{\pi \in \text{RPP}(\lambda)}\zeta^{|\pi|}\sum_{\tilde{\pi} \in \text{RPP}(\lambda)} z^{|\pi|}\,.
\end{equation}
Using the Verma character expression (\ref{eq:vermacharacter}), we can alternatively write this as:
\begin{equation}
    \mathcal{Z}_{S^3_b} = \sum_{\lambda} e^{\pi m \xi \sum_{s\in\lambda}c_{\lambda}(s)} \chi^{C,\lambda}(\zeta) \chi^{H,\lambda}(z)
\end{equation}
and, up to phases, we recover the IR formula expressed in terms of Verma characters in (7.45) of \cite{Gaiotto:2019mmf}.

\subsection{Poincar\'e polynomial limit}\label{subsec:pplimit}
We now consider another limit of the index corresponding to sending $q \to 0$ in the $A$- and $B$-shifted partition functions. As a trace over local operators we can write the $A$-shifted index as:
\begin{equation}
\begin{split}
    \mathcal{I}^A(\mathcal{B}_{\lambda}) &= \text{Tr}(-1)^F q^{J+\frac{R_V}{2}}t^{\frac{R_H-R_C}{2}}z^{F_H}\zeta^{F_C} \\
    &= \text{Tr}(-1)^F q^{D-\frac{R_C}{2}}t^{\frac{R_H-R_C}{2}}z^{F_H}\zeta^{F_C}\,. \\
\end{split}    
\end{equation}
Unitarity bounds allow us to re-write the exponent of $q$ in the above and also imply an expansion in positive powers of $q$ of the index. Sending $q \to 0$ selects the subspace where $D=\frac{1}{2}R_C$ i.e. bosonic Higgs branch operators uncharged under $F_C$. 
\begin{equation}
    \lim_{q\to0}\mathcal{I}^A(\mathcal{B}_{\lambda}) = \text{Tr}_{\mathcal{H}_A} t^{\frac{R_H}{2}-D}z^{F_H}\,.
\end{equation}
These operators are annihilated by the same additional supercharges as those in section \ref{subsec:vermacharacterlimit} and so in this limit the index is a refined Verma character of $\mathcal{A}^C_N$ that keeps track of R-charge graded by $t$. A similar argument holds for the $B$-shifted vortex partition function where instead we count Coulomb branch operators graded with $t \to t^{-1}$.

\paragraph{Geometric interpretation}
Quasimap moduli spaces $\textsf{QM}^{\boldsymbol{d}}$ can occasionally be realised as smooth varieties. For example the quasimap moduli space for the $T_{\rho}[SU(N)]$ theory coincides with Laumon space $\mathfrak{Q}^{\boldsymbol{d}}$ \cite{braverman2014macdonald,Crew:2020jyf}. In this case the equivariant Euler characteristic is the $\chi_t$ genus of Laumon space. In the limit $q\to 0$, the authors argue in \cite{Crew:2020jyf}, that $\chi_t(\mathfrak{Q}^{\boldsymbol{d}})$ becomes the Poincar\'e polynomial of the compact core $\pi^{-1}(0) \subset \mathfrak{Q}^{\boldsymbol{d}}$. In this section we study the same ``Poincar\'e polynomial'' limit of the ADHM theory and, by analogy with $T_{\rho}[SU(N)]$,\footnote{See appendix \ref{appendix:pplimit} for a more detailed review of the $T[SU(N)]$ example.} we expect that we are again computing a Poincar\'e polynomial of the compact core of a putative resolution of the space of quasimaps to $\text{Hilb}^N(\mathbb{C}^2)$:
\begin{equation}
    \lim_{q\to0} \mathcal{Z}_{\text{Vortex}} = \sum_{\boldsymbol{d}}\zeta^{\boldsymbol{d}} P_{t}(\widetilde{\textsf{QM}}^{\boldsymbol{d}}_{\lambda})\,.
\end{equation}
Indeed, in the following we find a finite polynomial in $t$ at each vortex number.\footnote{The existence of this limit is related to the large frame vanishing condition of the Higgs branch of the theory. In this case large frame vanishing holds for every $N$.}

\paragraph{Mirror limit as a generating function}
We first compute the Poincar\'e polynomial limit of the $B$-shifted index. We find
\begin{equation}
     \lim_{q\to0}\mathcal{Z}_{\text{1-loop}}^{B,\lambda } = \prod_{s\in\lambda} \frac{1}{1-z^{a_{\lambda}(s)+l_{\lambda}(s)+1}t^{\frac{1}{2}(-a_{\lambda}(s)+l_{\lambda}(s)+1)}}\,.
\end{equation}
Using the Hillman-Grassl correspondence \cite{GANSNER198171}, it is possible to write a refined generating function of reverse plane partitions that further grades the hook formula by the $t$ deformation above:\footnote{We thank Gjergji Zaimi for drawing our attention to the Hillman-Grassl correspondence.} 
\begin{equation}
    \sum_{\pi \in \text{RPP}(\lambda)} t^{\frac{1}{2}(\text{ht}'(\pi)-\text{ht}(\pi)+b(\pi))} z^{|\pi|} = \prod_{s\in\lambda} \frac{1}{1-z^{a_{\lambda}(s)+l_{\lambda}(s)+1}t^{\frac{1}{2}(-a_{\lambda}(s)+l_{\lambda}(s)+1)}}\,.
\end{equation}
The heights $\text{ht}(\pi)$, $\text{ht}'(\pi)$ and $b(\pi)$ are statistics on reverse plane partitions that we review in appendix \ref{appendix:hillman}. Note that sending $t\to1$ un-grades the character and recovers the Verma limit of the previous subsection.

In the $A$-shifted vortex partition function, the limit only receives contributions from vortices and one can check order-by-order in vortex number that sending $q \to 0$ indeed reproduces the refined generating function. We therefore expect the following identification as a consequence of 3d mirror symmetry
\begin{equation}
\begin{split}
    \lim_{q\to0}\mathcal{Z}^{A,\lambda}_{\text{Vortex}} &= \sum_{\pi \in \text{RPP}(\lambda)} t^{\frac{1}{2}(\text{ht}'(\pi)-\text{ht}(\pi)+b(\pi))} \zeta^{|\pi|} \\ &= \prod_{s\in\lambda} \frac{1}{1-z^{a_{\lambda}(s)+l_{\lambda}(s)+1}t^{\frac{1}{2}(-a_{\lambda}(s)+l_{\lambda}(s)+1)}} \\ &= \lim_{q\to 0}\mathcal{Z}_{\text{1-loop}}^{B,\lambda}\,.
\end{split}    
\end{equation}
In the above the identity includes the mirror symmetry exchange $\zeta \leftrightarrow z$ and $t \leftrightarrow t^{-1}$. Again, we observe that 3d mirror symmetry leads to a combinatorial identity for a generating function of reverse plane partitions.

\paragraph{Refined topological vertex}
We conclude this section with an observation relating the Poincar\'e polynomial limit of the shifted hemisphere partition function with the refined topological vertex \cite{Iqbal:2007ii}. This is consistent with the interpretation discussed previously in section \ref{subsec:quasimaps} of the partition function as the one-legged PT vertex.

In the Poincar\'e polynomial limit, computed via 3d mirror symmetry as above, the $A$-shifted vortex sum can be written
\begin{equation}
\begin{split}
    \lim_{q\to 0}\mathcal{Z}_{\text{Vortex}}^{A,\lambda} = \prod_{s\in\lambda} \frac{1}{1-\zeta^{a_{\lambda}(s)+l_{\lambda}(s)+1}t^{\frac{1}{2}(-a_{\lambda}(s)+l_{\lambda}(s)+1)}}\,,
\end{split}    
\end{equation}
we note that this can also be written as a principally specialised Macdonald polynomial as in the vertex of \cite{Awata:2008ed}.\footnote{We review conventions for symmetric functions in appendix \ref{appendix:symmetricfuns}.}
\begin{equation}\label{eq:Macdovertex}
    \lim_{q\to 0}\mathcal{Z}_{\text{Vortex}}^{A,\lambda} = (\zeta t^{\frac{1}{2}})^{-n(\lambda)}P_{\lambda}(1,(\zeta t^{\frac{1}{2}}),(\zeta t^{\frac{1}{2}})^{2},\ldots;\zeta t^{-\frac{1}{2}},\zeta t^{\frac{1}{2}})\,.
\end{equation}
In this limit, the classical terms also contribute:
\begin{equation}
    \lim_{q\to 0}\mathcal{Z}_{\text{Classical}}^{A,\lambda} = \zeta^{\frac{1}{2} \sum_{s\in \lambda} h_{\lambda}(s)} t^{\frac{1}{4} \sum_{s\in \lambda} c_{\lambda}(s)} \, .
\end{equation}
The 1-loop terms are simply one in this limit. Now, switching to the fugacities $(t_1,t_2)=(\zeta t^{\frac{1}{2}},\zeta^{-1}t^{\frac{1}{2}})$ for the torus action on $\text{Hilb}^N(\mathbb{C}^2)$ rather than the gauge theory parameters $\zeta$ and $t$, we have:
\begin{equation}
    \lim_{q\to0} \mathcal{Z}^{A,\lambda}_{S^1 \times D} = t_1^{\frac{1}{4}||\lambda||^2} t_2^{-\frac{1}{4} ||\lambda^{\vee}||^2} \prod_{s\in \lambda} \frac{1}{1-t_1^{l_{\lambda}(s)+1}t_2^{-a_{\lambda}(s)}}\, ,
\end{equation}
where we have re-written the sum over hook and content in terms of the weight $||\lambda||$ as in (\ref{eq:hookandcontentasnorm}). In conclusion, the Poincar\'e polynomial limit, $q\to0$, of the $A$-shifted hemisphere partition function coincides with the refined topological vertex of \cite{Iqbal:2007ii} with one non-trivial leg on the preferred direction:
\begin{equation}
    \lim_{q\to 0} \mathcal{Z}_{S^1 \times D}^{A,\lambda}  = C^{\text{(IKV)}}_{\emptyset, \emptyset, \lambda}(t=t_2^{-1},q=t_1)\,.
\end{equation}
The classical terms coincide with the framing factors in the topological vertex language. Later, in section \ref{subsec:hilbertseries}, we show that the holomorphic factorisation of the twisted index corresponds to gluing topological vertices.

The unrefined Verma character limit corresponds to setting $t\to1$. This sets equal the parameters of the Macdonald polynomial in (\ref{eq:Macdovertex}) and degenerates it to a Schur polynomial, i.e. the un-refined topological vertex of \cite{Aganagic:2003db}.
\begin{equation}
    \lim_{t\to1}\mathcal{Z}_{\text{Vortex}}^{A,\lambda} = \zeta^{-n(\lambda)}s_{\lambda}(1,\zeta,\zeta^2,\ldots) = \chi^{C,\lambda}\,.
\end{equation}
This specialised Schur formula for the Verma characters $\chi^{C,\lambda}$ of $\mathcal{A}^C_N$ were previously derived in \cite{Gaiotto:2019mmf}---in this work we realise the characters as the specialised vortex partition functions.

Via the correspondence with the vortex sum (\ref{eq:vortexsum}), the refined topological vertex is realised as a refined sum over \textit{reverse} plane partitions. This is in contrast to the usual formulation of the refined topological vertex as a sum over ordinary plane partitions. The apparent discrepancy can be explained by a choice of stability condition, i.e. sign of the real FI parameter, in the quiver construction of $\text{Hilb}^N(\mathbb{C}^2)$.\footnote{We thank Andrey Smirnov for explaining this point.} Our reverse plane partition sums are then analytic continuations in $\zeta\to \zeta^{-1}$ of the more conventional plane partition sums.

The natural conjecture for the 3d ADHM theory with $p \ge 1$ flavours is that the vortex partition functions\footnote{Or, in the opposite twist, the 1-loop contributions to the half-index of the necklace quiver mirror dual.} can be expressed as sums over coloured reverse plane partitions $\pi^{(i)}$ with $i=1,\ldots,p$. In this case we expect the $q\to0$ limit of the vortex sum to reproduce the refined amplitude of the strip geometry \cite{Iqbal:2004ne,Taki:2007dh} with external legs $\{\lambda^{i} \}$ corresponding to the bases of the reverse plane partitions i.e. $\pi^{(i)} \in \text{RPP}(\lambda^{(i)})$.


\section{Twisted Indices}\label{sec:factorisation}
In this section we consider the $A$- and $B$-twisted indices on $S^1\times S^2$ \cite{Closset:2016arn} of the 3d ADHM theory. A 3d $\mathcal{N}=4$ theory admits two fully topological twists  by $U(1)_H$ ($A$-twist) and $U(1)_C$ ($B$-twist) corresponding to a choice of $\mathcal{N}=2$ subalgebra:
\begin{itemize}
    \item The $B$-twist preserves $Q_+^{1\dot{1}}, \enspace Q_{+}^{2\dot{1}}, \enspace Q_{-}^{1\dot{2}}, \enspace Q_{-}^{2\dot{2}}$. 
    \item The $A$-twist preserves $Q_+^{1\dot{1}}, \enspace Q_{-}^{2\dot{2}}, \enspace Q_{+}^{1\dot{2}}, \enspace Q_{-}^{2\dot{1}}$.    
\end{itemize}
Both of the twists preserve a common $Q_{+}^{1\dot{1}}$ and $Q_{-}^{2\dot{2}}$ compatible with turning on a mass parameter $\tau$ for $U(1)_H-U(1)_C$. On $S^2$ we can also grade the index by a fugacity $q$ for angular momentum.
\begin{equation}\label{eq:twistedindextrace}
\begin{split}
    \mathcal{I}_{S^2}^A & = \text{Tr}_{\mathcal{H}_{S^2}^{A}} (-1)^F q^{J+\frac{R_H}{2}} t^{\frac{R_H-R_C}{2}} z^{F_H}\zeta^{F_C} \, , \\
    \mathcal{I}_{S^2}^B & = \text{Tr}_{\mathcal{H}_{S^2}^{B}}  (-1)^F q^{J+\frac{R_C}{2}} t^{\frac{R_H-R_C}{2}} z^{F_H}\zeta^{F_C} \,,
\end{split}
\end{equation}
where $\mathcal{H}_{S^2}^{A,B}$ denote respectively states in each twist on $S^2$ that are annihilated by $Q_{+}^{1\dot1}$ and $Q_{-}^{2\dot2}$. It has been shown in various cases that these indices factorise into the $A$- and $B$-shifted holomorphic blocks \cite{Crew:2020jyf,Cabo-Bizet:2016ars}. We will show our hemisphere partition functions associated to vacua $\lambda$ (\ref{eq:ADHMblock}) provide an \textit{exact} factorisation:
\begin{equation}
\begin{split}
     \mathcal{I}_{S^2}^{A} &= \sum_{\lambda} \left\lVert \mathcal{Z}_{S^1\times D}^{A,\lambda}(q,t,z,\zeta)\right\rVert_{\text{Twisted}}^2\,,\\
     \mathcal{I}_{S^2}^{B} &= \sum_{\lambda} \left\lVert \mathcal{Z}_{S^1\times D}^{B,\lambda}(q,t,z,\zeta)\right\rVert_{\text{Twisted}}^2,
\end{split}
\end{equation}
where the gluing is:
\begin{equation}
    q \rightarrow q^{-1},\quad t\rightarrow t,\quad z\rightarrow z, \quad\zeta \rightarrow \zeta.
\end{equation}
The $A$- and $B$-twisted indices of 3d $\mathcal{N}=4$ theories on $S^1\times S^2$ are in fact independent of $q$ \cite{Crew:2020jyf} and coincide with the $\mathcal{M}_C$ and $\mathcal{M}_H$ Hilbert series of the theory. In particular the $A$- and $B$-twisted indices depend only on the fugacities for those symmetries which act non-trivially on $\mathcal{M}_{C}$ and $\mathcal{M}_{H}$ respectively. Explicitly, 
\begin{eqnarray}
\mathcal{I}^{A}_{S^{2}} & = & \mathcal{I}^{A}_{S^{2}}[\zeta,t] \nonumber \\ 
\mathcal{I}^{B}_{S^{2}} & = & \mathcal{I}^{B}_{S^{2}}[z,t] 
\end{eqnarray}
One can see this from the Coulomb branch localisation where in the resulting Jeffreys-Kirwan contour integral $q$ dependence manifestly drops out, provided residue contributions at infinity vanish, and the resulting expression coincides with the Molien integral counting gauge invariant polynomials in the chiral ring. 

\subsection{$A$- and $B$-twisted indices}
The twisted indices can be computed in the UV with a Coulomb branch localisation scheme. This gives an expression for the indices as contour integrals with a Jeffreys-Kirwan contour prescription \cite{Benini:2015noa}.

\paragraph{$B$-twist}
In this work we focus mainly on the $B$-twisted index that reproduces the Hilbert series of the Higgs branch $\mathcal{M}_H$. It is given by the following JK residue integral:
\begin{equation}\label{eq:Btwistintegral}
\begin{split}
    \mathcal{I}_{S^2}^{B} = &\frac{(-1)^N}{N!} \sum_{\mathfrak{m}\in \mathbb{Z}^N}   \oint_{JK}  \prod_{a=1}^N \frac{dw_a}{2 \pi i w_a} (-\zeta)^{\sum_{a=1}^{N}\mathfrak{m}_a} (t^{\frac{1}{2}}-t^{-\frac{1}{2}})^{N}\\
    & \prod_{a,b=1 \atop a\neq b}^N \frac{\left[\left(\frac{w_a}{w_b}\right)^{\frac{1}{2}}\right]^{\mathfrak{m}_a-\mathfrak{m}_b-1}}{\left(\frac{w_a}{w_b}q^{1-\frac{1}{2}\mathfrak{m}_a+\frac{1}{2}\mathfrak{m}_b};q\right)_{\mathfrak{m}_a-\mathfrak{m}_b-1}}
    \frac{\left[\left(\frac{w_a}{w_b}\right)^{\frac{1}{2}}t^{-\frac{1}{2}}\right]^{\mathfrak{m}_a-\mathfrak{m}_b-1}}{\left(\frac{w_a}{w_b}t^{-1}q^{1-\frac{1}{2}\mathfrak{m}_a+\frac{1}{2}\mathfrak{m}_b};q\right)_{\mathfrak{m}_a-\mathfrak{m}_b-1}}\\
    &\prod_{a=1}^N \frac{\left(w_a^{\frac{1}{2}}t^{\frac{1}{4}}\right)^{\mathfrak{m}_a+1}}{\left(w_a t^{\frac{1}{2}} q^{-\frac{1}{2}\mathfrak{m}_a};q\right)_{\mathfrak{m}_a+1}}
    \frac{\left(w_a^{-\frac{1}{2}}t^{\frac{1}{4}}\right)^{-\mathfrak{m}_a+1}}{\left(w_a^{-1} t^{\frac{1}{2}} q^{\frac{1}{2}\mathfrak{m}_a};q\right)_{-\mathfrak{m}_a+1}}\\
    &\prod_{a,b=1}^N \frac{\left[\left(\frac{w_a}{w_b}\right)^{\frac{1}{2}}z^{\frac{1}{2}}t^{\frac{1}{4}}\right]^{\mathfrak{m}_a-\mathfrak{m}_b+1}}{\left(\frac{w_a}{w_b}zt^{\frac{1}{2}} q^{-\frac{1}{2}\mathfrak{m}_a+\frac{1}{2}\mathfrak{m}_b};q\right)_{\mathfrak{m}_a-\mathfrak{m}_b+1}}
    \frac{\left[\left(\frac{w_a}{w_b}\right)^{-\frac{1}{2}}z^{-\frac{1}{2}}t^{\frac{1}{4}}\right]^{-\mathfrak{m}_a+\mathfrak{m}_b+1}}{\left(\frac{w_b}{w_a}z^{-1}t^{\frac{1}{2}} q^{\frac{1}{2}\mathfrak{m}_a-\frac{1}{2}\mathfrak{m}_b};q\right)_{-\mathfrak{m}_a+\mathfrak{m}_b+1}}.
\end{split}
\end{equation}
The second line (together with $(t^{\frac{1}{2}}-t^{-\frac{1}{2}})^{N}$ in the first line corresponds to the $\mathcal{N}=4$ vector multiplet, the third line the fundamental hypermultiplet, and the last line the adjoint hypermultiplet. The contour prescription encloses poles labelled by reverse plane partitions over a Young diagram base $\lambda$ such that $|\lambda|=N$.\footnote{More accurately, similarly to the hemisphere partition function, only those poles labelled by RPPs over a Young diagram have non-vanishing residues.}  The poles take the form for each $s\in \lambda$:
\begin{equation}
\begin{split}
    w_{s} &= t^{\frac{1}{2}}q^{\frac{1}{2}\left(\tilde{\pi}_s-\pi_s\right)} v_s^{-1},\\
    \mathfrak{m}_{s}&=\tilde{\pi}_{s}+\pi_{s},
\end{split}
\end{equation}
where $v_s=z^{i_s-j_s}t^{\frac{1}{2}(i_s+j_s)}$ and  $\tilde{\pi}_{s}$ and $\pi_{s}$ the heights of reverse plane partitions $\pi$ and $\tilde{\pi}$ above a box $s\in\lambda$. After lengthy cancellation, particularly for the 1-loop piece,\footnote{These cancellations are analogous to the cancellations for the perturbative piece of the hemisphere partition function in appendix \ref{appendix:holoblockdetailedcomputation} so we do not reproduce them here.} we obtain:
\begin{equation}\label{eq:Btwistfactorisation}
    \mathcal{I}_{S^2}^{B} = \sum_{|\lambda|=N}
    \mathcal{I}_{\text{Classical}}^{B,\lambda} \mathcal{I}_{\text{1-loop}}^{B,\lambda}
    \mathcal{Z}^{B,\lambda}_{\text{Vortex}}(q,t,z,\zeta)
    \mathcal{Z}^{B,\lambda}_{\text{Vortex}}(q^{-1},t,z,\zeta),
\end{equation}
where:
\begin{equation}\label{eq:Btwistfactorisation2}
    \mathcal{I}_{\text{Classical}}^{B,\lambda} =  \prod_{s\in \lambda} z^{i_s+j_s-1} t^{\frac{1}{2}(i_s-j_s)}=   \prod_{s \in \lambda} z^{h_{\lambda}(s)} t^{-\frac{1}{2}c_\lambda(s)} = \left\lVert \mathcal{Z}^{B,\lambda}_{\text{Classical}}\right\rVert^2\,.
\end{equation}
The 1-loop piece is:
\begin{equation}\label{eq:Btwistfactorisation3}
\begin{split}
    \mathcal{I}_{\text{1-loop}}^{B,\lambda} &= \prod_{s\in \lambda}\frac{\left(z^{a_{\lambda}(s)+l_{\lambda}(s)+1}t^{\frac{1}{2}\left(-a_{\lambda}(s)+l_{\lambda}(s)-1\right)}q;q\right)_{-1}}{\left(z^{a_{\lambda}(s)+l_{\lambda}(s)+1}t^{\frac{1}{2}\left(-a_{\lambda}(s)+l_{\lambda}(s)+1\right)};q\right)_{1}}\\
    &= \left\lVert \mathcal{Z}_{\text{1-loop}}^{B,\lambda} \right\rVert_{\text{Twisted}}^2,
\end{split}
\end{equation}
where in the last line we have used the fusion rule (\ref{eq:twistedindexfusionrule}) to factorise in terms of perturbative contributions to the blocks. In conclusion the twisted index exactly factorises:
\begin{equation}\label{eq:Btwistfactorisation4}
    \mathcal{I}_{S^2}^{B,\lambda} = \sum_{\lambda} \left\lVert \mathcal{Z}_{S^1\times D}^{B,\lambda}(q,t,z,\zeta)\right\rVert_{\text{Twisted}}^2.
\end{equation}

\paragraph{$A$-twist}
For completeness, in this section we state the results of the factorisation in the $A$-twist.
\begin{equation}
\begin{split}
    \mathcal{I}_{S^2}^{A} = &\frac{(-1)^N}{N!} \sum_{\mathfrak{m}\in \mathbb{Z}^n}   \oint_{JK} \prod_{a=1}^N \frac{dw_a}{2 \pi i w_a}(-\zeta)^{\sum_{a=1}^{N}\mathfrak{m}_a} (t^{\frac{1}{2}}-t^{-\frac{1}{2}})^{-N}\\
    & \prod_{a,b=1 \atop a\neq b}^N \frac{\left[\left(\frac{w_a}{w_b}\right)^{\frac{1}{2}}\right]^{\mathfrak{m}_a-\mathfrak{m}_b-1}}{\left(\frac{w_a}{w_b}q^{1-\frac{1}{2}\mathfrak{m}_a+\frac{1}{2}\mathfrak{m}_b};q\right)_{\mathfrak{m}_a-\mathfrak{m}_b-1}}
    \frac{\left[\left(\frac{w_a}{w_b}\right)^{\frac{1}{2}}t^{-\frac{1}{2}}\right]^{\mathfrak{m}_a-\mathfrak{m}_b+1}}{\left(\frac{s_a}{s_b}t^{-1}q^{-\frac{1}{2}\mathfrak{m}_a+\frac{1}{2}\mathfrak{m}_b};q\right)_{\mathfrak{m}_a-\mathfrak{m}_b+1}}\\
    &\prod_{a=1}^N \frac{\left(s_a^{\frac{1}{2}}t^{\frac{1}{4}}\right)^{\mathfrak{m}_a}}{\left(s_a t^{\frac{1}{2}} q^{\frac{1}{2}-\frac{1}{2}\mathfrak{m}_a};q\right)_{\mathfrak{m}_a}}
    \frac{\left(w_a^{-\frac{1}{2}}t^{\frac{1}{4}}\right)^{-\mathfrak{m}_a}}{\left(w_a^{-1} t^{\frac{1}{2}} q^{\frac{1}{2}+\frac{1}{2}\mathfrak{m}_a};q\right)_{-\mathfrak{m}_a}}\\
    &\prod_{a,b=1}^N \frac{\left[\left(\frac{w_a}{w_b}\right)^{\frac{1}{2}}z^{\frac{1}{2}}t^{\frac{1}{4}}\right]^{\mathfrak{m}_a-\mathfrak{m}_b}}{\left(\frac{w_a}{w_b}zt^{\frac{1}{2}} q^{\frac{1}{2}-\frac{1}{2}\mathfrak{m}_a+\frac{1}{2}\mathfrak{m}_b};q\right)_{\mathfrak{m}_a-\mathfrak{m}_b}}
    \frac{\left[\left(\frac{w_a}{w_b}\right)^{-\frac{1}{2}}z^{-\frac{1}{2}}t^{\frac{1}{4}}\right]^{-\mathfrak{m}_a+\mathfrak{m}_b}}{\left(\frac{w_b}{w_a}z^{-1}t^{\frac{1}{2}} q^{\frac{1}{2}+\frac{1}{2}\mathfrak{m}_a-\frac{1}{2}\mathfrak{m}_b};q\right)_{-\mathfrak{m}_a+\mathfrak{m}_b}}\,.
\end{split}
\end{equation}
 The poles take the form of reverse plane partitions over a base Young diagram $\lambda$ such that $|\lambda|=N$. Those which contribute to the JK residue are indexed by boxes $s \in \lambda$, with $q$-dependence corresponding to vortex and anti-vortex number:
\begin{equation}    
    w_{s} = t^{\frac{1}{2}}q^{\frac{1}{2}\left(\tilde{\pi}_s-\pi_s\right)} v_s^{-1}
\end{equation}
where $\pi_s$ and $\tilde{\pi}_{s}$ are the height of two reverse plane partitions  obeying $\mathfrak{m}_{s} = \tilde{k}_{s}+k_{s}+(i_{s}+j_{s}-1)$. Evaluating the integral at these poles we arrive at:
\begin{equation}\label{eq:Atwistfactorisation}
    \mathcal{I}_{S^2}^{A} = \sum_{|\lambda|=N}
    \mathcal{I}_{\text{Classical}}^{A,\lambda} \mathcal{I}_{\text{1-loop}}^{A,\lambda}
    \mathcal{Z}^{A,\lambda}_{\text{Vortex}}(q,t,z,\zeta)
    \mathcal{Z}^{A,\lambda}_{\text{Vortex}}(q^{-1},t,z,\zeta),
\end{equation}
where:
\begin{equation}
    \mathcal{Z}_{\text{cl}}^\lambda =  \prod_{s}(\zeta)^{i_{s}+j_{s}-1}t^{\frac{1}{2}\left(-i_{s}+j_{s}\right)} =  \left\lVert \mathcal{Z}^{A,\lambda}_{\text{Classical}}\right\rVert^2\,.
\end{equation}
For the one-loop piece we find:
\begin{equation}
\begin{split}
    \mathcal{Z}_{\text{1-loop}}^{A,\lambda} =& \prod_{s\in Y} \frac{\left(z^{a_{\lambda}(s)+l_{\lambda}(s)+1}(tq)^{\frac{1}{2}\left(-a_{\lambda}(s)+l_{\lambda}(s)-1\right)}q;q\right)_{a_{\lambda}(s)-l_{\lambda}(s)}}{\left(z^{a_{\lambda}(s)+l_{\lambda}(s)+1}(tq)^{\frac{1}{2}\left(-a_{\lambda}(s)+l_{\lambda}(s)+1\right)};q\right)_{a_{\lambda}(s)-l_{\lambda}(s)}}\\
    &= \left\lVert \mathcal{Z}_{\text{1-loop}}^{A,\lambda} \right\rVert_{\text{Twisted}}^2\,.
\end{split}
\end{equation}
In conclusion:
\begin{equation}
    \mathcal{I}_{S^2}^{A} = \sum_{\lambda} \left\lVert \mathcal{Z}_{S^1\times D}^{A,\lambda}(q,t,z,\zeta)\right\rVert_{\text{Twisted}}^2.
\end{equation}

\subsection{Hilbert series of the Hilbert scheme}\label{subsec:hilbertseries}
We now focus on the $B$-twisted index. As argued at the beginning of this section, the index coincides with the Higgs branch Hilbert series. Since the index is independent of $q$, we are free to send $q\to0$ in the factorisation (\ref{eq:Btwistfactorisation}--\ref{eq:Btwistfactorisation4}). In this limit only the 1-loop and classical terms survive and we have:
\begin{equation}
\begin{split}
    \lim_{q \to 0}\mathcal{Z}_{\text{1-loop}}^{B,\lambda}(z,\zeta;q,t) &= \prod_{s\in\lambda}\frac{1}{1-z^{a_{\lambda}(s)+l_{\lambda}(s)+1}t^{\frac{1}{2}\left(-a_{\lambda}(s)+l_{\lambda}(s)+1\right)}} \,,\\ 
    \lim_{q \to 0}\mathcal{Z}_{\text{1-loop}}^{B,\lambda}(z,\zeta;q^{-1},t) &= \prod_{s\in\lambda}\frac{1}{1-z^{a_{\lambda}(s)+l_{\lambda}(s)+1}t^{\frac{1}{2}\left(-a_{\lambda}(s)+l_{\lambda}(s)-1\right)}} \,,\\
    \lim_{q\to 0}\mathcal{Z}_{\text{Classical}}^{B,\lambda}(z,\zeta;q,t) &= z^{\frac{1}{2} \sum_{s\in\lambda} h_{\lambda}(s)}t^{-\frac{1}{4} \sum_{s\in\lambda}c_{\lambda}(s)} \, , \\
    \lim_{q\to 0}\mathcal{Z}_{\text{Classical}}^{B,\lambda}(z,\zeta;q^{-1},t) &= z^{\frac{1}{2} \sum_{s\in\lambda} h_{\lambda}(s)}t^{-\frac{1}{4} \sum_{s\in\lambda}c_{\lambda}(s)} \, .
\end{split}
\end{equation}
The $B$-twisted index can then be expressed as,\footnote{We have used the shorthand $\bar{\mathcal{Z}}$ to denote the twisted index gluing $q\to q^{-1}$.}
\begin{eqnarray}
\mathcal{I}_{S^2}^{B}[z,t] & = &\lim_{q \to 0} \mathcal{I}_{S^2}^{B}[z,t]  \nonumber \\ &= & \sum_{\substack{\lambda\\ |\lambda|=N}} \lim_{q \to 0}
    \mathcal{Z}_{\text{Classical}}^{B,\lambda} \bar{\mathcal{Z}}_{\text{Classical}}^{B,\lambda}\mathcal{Z}_{\text{1-loop}}^{B,\lambda} \bar{\mathcal{Z}}_{\text{1-loop}}^{B,\lambda} \nonumber  \\
    &= & \sum_{\substack{\lambda \\ |\lambda|=N}}\prod_{s\in \lambda} \frac{z^{a_{\lambda}(s)+l_{\lambda}(s)+1}t^{\frac{1}{2}(-a_{\lambda}(s)+l_{\lambda}(s)})}{\left(1-z^{a_{\lambda}(s)+l_{\lambda}(s)+1}t^{\frac{1}{2}(-a_{\lambda}(s)+l_{\lambda}(s)+1)}\right) \left(1-z^{a_{\lambda}(s)+l_{\lambda}(s)+1}t^{\frac{1}{2}(-a_{\lambda}(s)+l_{\lambda}(s)-1)} \right)} \nonumber \,.
\end{eqnarray}
We remark that, by mirror symmetry, this is an expression for the Hilbert series of the Higgs branch $\mathcal{M}_H = \text{Hilb}^N(\mathbb{C}^2)$ in terms of Higgs branch Verma denominators since, in the notation of section \ref{subsec:vermacharacterlimit}, we have:
\begin{equation}
     \lim_{q\to0}\mathcal{Z}_{\text{1-loop}}^{B,\lambda} = \lim_{q\to0} \mathcal{Z}_{\text{Vortex}}^{A,\lambda}(t\to t^{-1}, \zeta \to z) = \chi^{H,\lambda}.
\end{equation}
Geometrically, this yields a formula for the Higgs branch hilbert series in terms of the Poincar\'e polynomials of maps to the Coulomb branch.

After changing variables to $(t_1,t_2) = (zt^{\frac{1}{2}},z^{-1}t^{\frac{1}{2}})$ to match to the more conventional symmetry generators for the torus action on $\mathbb{C}^{2}$, 
the twisted index as expressed above recovers the familiar fixed point formula for the Hilbert series of $\text{Hilb}^N(\mathbb{C}^2)$ \cite{nakajima1999lectures}. The latter is also conveniently written via the generating function,\footnote{\text{PE} denotes the plethystic exponential as defined in equation \ref{eq:plethysticexp}.}
\begin{eqnarray}
{\tt Z}[\Lambda,t_{1}, t_{2}]  &  :=& \sum_{N=0}^{\infty}\Lambda^{N}\mathcal{Z}_{\text{H.S.}} \left[\text{Hilb}^N(\mathbb{C}^2) \right](t_1,t_2)= 
\text{PE} \left[\frac{-\sqrt{t_1 t_2} \Lambda}{(1-t_{1})(1-t_{2})}\right]
\label{HS}
\end{eqnarray}
Thus, for the rank $N$ theory we have, 
\begin{eqnarray}
\mathcal{I}^{B}_{S^{2}}[z,t] & =& \mathcal{Z}_{\text{H.S.}}\left[\text{Hilb}^{N}\left(\mathbb{C}^{2}\right)\right](zt^{\frac{1}{2}},z^{-1}t^{\frac{1}{2}})= \text{PE} \left[\frac{-\Lambda t^{\frac{1}{2}}}{(1- zt^{\frac{1}{2}})(1-z^{-1}t^{\frac{1}{2}})}\right]\Bigg|_{O(\Lambda^{N})}
\label{HSH}
\end{eqnarray}
where the final suffix indicates we are extracting the term of order $\Lambda^{N}$. 
Finally, taking account of the self mirror property of the theory we also have, 
\begin{eqnarray}
\mathcal{I}^{A}_{S^{2}}[\zeta ,t] & = & \mathcal{Z}_{\text{H.S.}}\left[\text{Hilb}^{N}\left(\mathbb{C}^{2}\right)\right](\zeta t^{-\frac{1}{2}},\zeta^{-1}t^{-\frac{1}{2}})= \text{PE} \left[\frac{-\Lambda t^{-\frac{1}{2}}}{(1- \zeta t^{-\frac{1}{2}})(1-\zeta^{-1}t^{-\frac{1}{2}})}\right]\Bigg|_{O(\Lambda^{N})}
\label{HSC}
\end{eqnarray}
This is consistent with the identification of both the Higgs and Coulomb branches of the rank $N$ theory as the Hilbert scheme of $N$ points on $\mathbb{C}^{2}$. 
Thus the $A$- and $B$-twisted indices coincide with the Hilbert series of the Coulomb branch and Higgs branch respectively as expected.

\paragraph{Type IIA string theory interpretation}
Recall from section \ref{subsec:quasimaps} that the vortex partition function is expected to coincide with the bare K-theoretic PT vertex with one non-trivial leg.
\begin{equation}
    \mathcal{Z}^{A,\lambda}_{\text{Vortex}} = V_{\text{PT}}^{\emptyset,\emptyset,\lambda}\,.
\end{equation}
Indeed, as discussed in section \ref{subsec:pplimit}, the Poincar\'e polynomial limit reproduces the refined topological vertex \cite{Iqbal:2007ii} with the preferred direction on the non-trivial leg. The classical term yields the topological vertex framing factor:
\begin{equation}
    \lim_{q\to0} \mathcal{Z}_{\text{Classical}}^{A,\lambda}\mathcal{Z}_{\text{Vortex}}^{A,\lambda} = C^{\text{(IKV)}}_{\emptyset,\emptyset,\lambda}(t_2^{-1},t_1)
\end{equation}
where, as above, the gauge theory and vertex parameters are identified as $t_{1}=\zeta t^{-\frac{1}{2}}$ and $t_{2}=\zeta^{-1}t^{-\frac{1}{2}}$. Now we note that the gluing together of the blocks to form the $A$-twisted index described above is identical to the gluing of refined topological vertices\footnote{The block gluing distributes the framing factor equally amongst the vertices, in contrast to \cite{Iqbal:2007ii}, and coincides with the alternative framing factor choice of \cite{Awata:2008ed}.} to get the partition function for the resolved conifold $\mathcal{C}=\mathcal{O}(-1)\oplus\mathcal{O}(-1) \rightarrow \mathbb{P}^{1}$. The conjugation of the vertices coincides with the conjugation on the topologically twisted gluing. In particular, the vertex calculation is given as, 
\begin{eqnarray}
\mathcal{Z}_{\mathcal{C}}[Q,t_{1},t_{2}] & = &  
\sum_{\lambda} (-Q)^{|\lambda|}C_{\emptyset, \emptyset, \lambda}(t_2^{-1},t_1) C_{ \emptyset,\emptyset,\lambda^{\vee}}(t_1,t_2^{-1}) \nonumber \\
&  = &  \prod_{i,j=1}^{\infty} 
\left(1-Q t_{1}^{i-\frac{1}{2}}t_{2}^{-j+\frac{1}{2}}\right) \nonumber \\ 
& = & {\rm PE} \left[ - \frac{\sqrt{t_{1}t_{2}}Q}{(1-t_{1})(1-t_{2}^{-1})}\right]
\label{int}
\end{eqnarray}
As shown in \cite{Kononov:2019fni}, the full K-theoretic PT vertex glues to give the same result. Indeed, the resulting PT partition function is independent of the additional parameters corresponding to the torus action on the conifold. This precisely parallels the cancellation of $q$ and $z$ ($\zeta$)  dependence of the blocks when glued to form the $A$-twisted  ($B$-twisted) partition function.

The refined partition function of the resolved conifold $\mathcal{C}$ also has an interpretation in terms of Type IIA string theory on $\mathcal{C}\times \mathbb{R}^{3,1}$. 
In a particular chamber of the K\"{a}hler moduli space, it corresponds to an index computing the bound states of $D0$ and $D2$ branes in the presence of a single $D6$ brane wrapped on 
$\mathcal{C}$ \cite{Dimofte:2009bv}. Specifically, in each sector of fixed $D$-brane charge $\gamma$, it computes a trace over the Hilbert space $\mathcal{H}(\gamma,u)$ of BPS states weighted by 
their four-dimensional spin $J_{3}$,  
\begin{eqnarray}
\Omega^{\rm ref}(\gamma;u;y) & := &  {\rm Tr}_{\mathcal{H}(\gamma,u)}\, (-y)^{2J_{3}}
\end{eqnarray}
Here $u=B+iJ$ denotes the asymptotic value of complexified K\"{a}hler parameter on which the index has piecewise constant dependence, splitting the K\"{a}hler moduli space into chambers separated by walls of marginal stability. The resulting refined BPS index of the resolved conifold is defined as, 
\begin{eqnarray}
\mathcal{Z}^{\rm ref}_{IIA}\left(\mathfrak{q},Q,y;u\right) &:= & \sum_{m,n\in \mathbb{Z}}\, (-\mathfrak{q})^{n} Q^{m} 
\Omega^{\rm ref}(\gamma_{n,m};u;y)
\end{eqnarray}
where the integers $n$ and $m$ correspond to $D0$ and $D2$ branes charges respectively. 
In a particular region $\mathcal{U}_{PT}$ of moduli space\footnote{This region is described in \cite{Dimofte:2009bv}  as an $n\rightarrow\infty$ limit of a certain sequence $\{\tilde{C}_{n}\}$ of chambers. See the discussion around Eqn (2.18) in this reference for a more detailed explanation.}, the index coincides refined partition function of topological string theory computed above, 
\begin{eqnarray}
\mathcal{Z}^{\rm ref}_{IIA}
\left(\mathfrak{q},Q,y;u\in \mathcal{U}_{PT}\right) & = &  \prod_{i,j=1}^{\infty} 
\left(1- Q (\mathfrak{q}y)^{i-\frac{1}{2}}\left(\frac{\mathfrak{q}}{y}\right)^{j-\frac{1}{2}}\right) \nonumber \\ 
& =& {\rm PE} \left[ - \frac{\mathfrak{q}Q}{(1-\mathfrak{q}/y)(1-\mathfrak{q}y)}\right] \nonumber 
\end{eqnarray}
This coincides with (\ref{int}) above with the identifications $t_{1}=\mathfrak{q}y$, $t_{2}=\mathfrak{q}/y$. To understand the connection to the gauge theory twisted index, note that the infinite product on the RHS of this equation is convergent in the region,  $|\mathfrak{q}y|$, $|\mathfrak{q}/y|<1$. The analytic continuation to the region $|\mathfrak{q}y|$, $|y/\mathfrak{q}|<1$ is given as,  
\begin{eqnarray}
\mathcal{Z}^{\rm ref}_{IIA} 
\left(\mathfrak{q},Q,y;u\in \mathcal{U}_{PT}\right) & = &  \sum_{m\in \mathbb{Z}} 
Q^{m} \mathcal{Z}^{\rm ref}_{m}(\mathfrak{q},y)
\nonumber \\ 
& =& \prod_{i,j=1}^{\infty} 
\left(1- Q (\mathfrak{q}y)^{i-\frac{1}{2}}\left(\frac{y}{\mathfrak{q}}\right)^{j-\frac{1}{2}}\right)^{-1} \nonumber \\ 
& =& {\rm PE} \left[\frac{yQ}{(1-y/\mathfrak{q})(1-y\mathfrak{q})}\right]\,. \nonumber 
\end{eqnarray}
With appropriate identifications of the parameters, this is equal to the  generating function ${\tt Z}[\Lambda,t_{1},t_{2}]$ for the Hilbert series of the Hilbert scheme defined in (\ref{HS}) above, 
\begin{eqnarray}
\mathcal{Z}^{\rm ref}_{IIA} 
\left(\mathfrak{q},Q,y;u\in \mathcal{U}_{PT}\right) & = & 
{\tt Z}[yQ,y\mathfrak{q}, {y}/{\mathfrak{q}}]
\end{eqnarray}
and can thus be related to the gauge theory twisted index using (\ref{HSH}, \ref{HSC}). 

Putting together the various equalities described above, we deduce that the $A$-twisted index of the ADHM quiver theory of rank $N$ 
computes an index for BPS bound states of a configuration consisting of $N$ $D2$ branes wrapped on the compact $\mathbb{P}^{1}$ of the conifold in the presence of a single $D6$ brane and an arbitrary number of $D0$ branes. The vortex counting parameter $\zeta$ corresponds to $\mathfrak{q}$ where $-\mathfrak{q}$ is the fugacity for $D0$ brane charge, while the fugacity $t^{-\frac{1}{2}}$ for the Coulomb branch $R$ symmetry corresponds to the fugacity $y$ for spin on the IIA side. More precisely, 
\begin{eqnarray}
\mathcal{I}^{A}_{S^{2}}\left[\zeta, t \right] & = & t^{\frac{N}{2}}\mathcal{Z}^{\rm ref}_{N}(\zeta,t^{-\frac{1}{2}}) 
\end{eqnarray}
Such a correspondence can be motivated heuristically as follows. The ADHM quiver theory on $\mathbb{R}^{1,2}$ can be realised in Type IIA string theory as the worldvolume  theory of $N$ $D2$ branes in the presence of a single $D6$ on flat ten dimensional space. In this context the vortices of the 3d gauge theory correspond to
$D0$ branes bound to the $D2$s. The vacuum moduli space of the 3d theory corresponds to motion of $N$ identical $D2$s in the eight transverse dimensions.\footnote{One of which arises from the dual photon and becomes a geometrical dimension when lifted to M theory.} Motion in the four tranverse dimensions parallel to the $D6$ corresponds to the Higgs branch of the 3d theory while motion 
in the four remaining transverse dimensions corresponds to the Coulomb branch. The $A$-twisted index arises from compactification of the three-dimensional theory on $\mathbb{P}^{1}$ with a twist involving the Higgs branch $R$-symmetry. Such twisted compactifications of the $D$-brane world volume can indeed be realised in string theory by wrapping the branes on a non-trivial cycle in a Calabi-Yau threefold \cite{Bershadsky:1995qy} where the R-symmetry twist is induced by the non-trivial fibering of the normal bundle over $\mathbb{P}^{1}$. It would be interesting to make this precise in the present context.  

\subsection{Large $N$ limit of the Hilbert series}\label{subsec:largeN}
We now compute the large gauge rank limit of the $B$-twisted index.

\paragraph{Adding flavours} From this point on in the paper we consider the 3d ADHM theory with $p \ge 1$ flavours. This allows us to elucidate some of the structure in this and the following section more clearly. We add hypermultiplets $(I_i,J_i)$ with $i=1,\ldots, p$ in the fundamental representation. The theory is no longer self-mirror and is now mirror dual to an affine quiver theory \cite{Hosseini:2016ume}. The Higgs branch coincides with the moduli space of $\mathbb{C}^2$ instantons $\mathcal{M}_H = \mathcal{M}_{N,p}$ and has a larger global symmetry group $G_H = U(1) \times U(1) \times SU(p)$ for which we introduce the additional flavour fugacities $x_i$. The fixed points under this group action are now labelled by p-coloured Young diagrams $\{\lambda^{(i)}\enspace | \enspace i=1,\ldots,p\}$ such that $\sum_i |\lambda^{(i)}| = N$.

\paragraph{Molien integral and symmetric functions}
The $B$-twisted index coincides with the Hilbert series of the Higgs branch $\mathcal{M}_H$. In the absence of external flux only the sector with $\mathfrak{m}=0$ contributes to the integral (\ref{eq:Btwistintegral}) and the JK prescription picks out poles in the unit circle. In this case $q$ drops out of the integrand and, up to an unimportant $t$ prefactor, we find:
\begin{equation}
\begin{split}
    \mathcal{Z}_{\text{H.S.}}\left[ \mathcal{M}_{N,p}\right] &= \mathcal{I}_{S^2}^{B} = \frac{1}{N!}\oint_{S^1} \prod_{a=1}^N \frac{ds_a}{2 \pi i s_a} \prod_{a\neq b}^N \left(1-\frac{s_a}{s_b} \right) \prod_{a,b=1}^N \frac{\left( 1-t \frac{s_a}{s_b}\right)}{\left(1-zt^{\frac{1}{2}}\frac{s_a}{s_b}\right) \left( 1-z^{-1}t^{\frac{1}{2}} \frac{s_a}{s_b} \right)} \\
    &\qquad\qquad\qquad\qquad \qquad\qquad \prod_{a=1}^N \prod_{i=1}^p \frac{1}{1-t^{\frac{1}{2}}s_a x_i} \frac{1}{1-t^{\frac{1}{2}}s_a^{-1}x_i^{-1}}\,.
\end{split}    
\end{equation}
This integral is a Molien integral counting gauge invariant polynomials in the scalars $(A,B,I,J)$ generating $\mathbb{C}[\mathcal{M}_H]$, it can be evaluated using symmetric function methods. We review the details of (a generalisation of) this calculation in appendix \ref{appendix:molienintegrals} -- the upshot in this example is:
\begin{equation}\label{eq:adhmhs}
    \mathcal{Z}_{\text{H.S.}}\left[ \mathcal{M}_{N,p}\right] = \sum_{\lambda, \mu} \frac{1}{(zt^{\frac{1}{2}};zt^{\frac{1}{2}})_{N-l(\mu)}}\left(z^{-1}t^{\frac{1}{2}}\right)^{|\sigma|}P'_{\mu/\lambda}\left(t^{\frac{1}{2}}X;zt^{\frac{1}{2}} \right) Q'_{\mu/\lambda}\left(t^{\frac{1}{2}}\bar{X};zt^{\frac{1}{2}}\right)
\end{equation}
where $P'_{\mu/\lambda}(X;t)$ and $Q'_{\mu/\lambda}(\bar{X};t)$ are different normalisations of skew Milne polynomials in the flavour fugacities $X=\{x_1,\ldots,x_p\}$ and $\bar{X}=\{x_1^{-1},\ldots,x_p^{-1}\}$. In the large $N$ limit we can use the Cauchy type identity proved in \ref{appendix:molienintegrals} to find:
\begin{equation}
    \lim_{N\to\infty}\mathcal{Z}_{\text{H.S.}}\left[\mathcal{M}_{N,p}\right] =  \prod_{k=0}^{\infty} \frac{1}{1-\left(zt^{\frac{1}{2}}\right)^{k+1}}\frac{1}{1-\left(z^{-1}t^{\frac{1}{2}}\right)^{k+1}} \prod_{i,j=1}^p \prod_{l=0}^{\infty}\frac{1}{1-\left(zt^{\frac{1}{2}}\right)^{l}\left(z^{-1}t^{\frac{1}{2}}\right)^{k}x_i/x_j}.
\end{equation}
This calculation is consistent with the fact that at large $N$ there are no trace relations in the chiral ring $\mathbb{C}[\mathcal{M}_H]$ and it becomes freely generated by the gauge invariant polynomials $\text{Tr}A^i$, $\text{Tr}B^j$ and $IA^iB^jJ$.


\section{Quantum Mechanics and Simple Modules}\label{sec:qm}
In this section we consider an alternative Neumann boundary condition for the half index of the theory in the presence of a line operator. As in the previous subsection \ref{subsec:largeN}, we work with $p \ge 1$ flavours and introduce corresponding fugacities $x_i$ with $i=1,\ldots,p$. In this context, we find a connection to the matrix model of a one dimensional quantum mechanics and discuss a geometric interpretation of the half index of this boundary condition as counting sections of line bundles over a particular Lagrangian in the ADHM moduli space.

We refer the reader to \cite{Bullimore:2016nji} for a detailed construction of the Neumann boundary condition. We note here only that setting the gauge multiplet to Neumann preserves gauge symmetry at the boundary and so the half index is computed by a contour integral that projects onto gauge invariant operators. We choose a particularly simple Lagrangian splitting for the matter of the ADHM theory corresponding to the natural splitting associated to the quiver in figure \ref{fig:adhm}. Specifically, in the notation of section \ref{sec:background}, $J$ and $A$ are set to zero on the boundary whilst $I$ and $B$ are allowed to fluctuate.

In this section we also include a Wilson line of charge $\mathfrak{n}>0$ in the totally symmetric representation of $GL(\mathbb{C}^N)$. The line operator is inserted at $x^2=x^3=0$ and extends perpendicularly out of the boundary---we denote the Wilson line by $\mathcal{W}_{\mathfrak{n}}$. The half index then counts the boundary local operators that transform under the representation corresponding to $\mathcal{W}_{\mathfrak{n}}$. We refer the reader to \cite{Dimofte:2017tpi} for a more detailed discussion of the computation of half indices in the presence of a Wilson line.

\paragraph{Contour integral form}
In the setup discussed above and working with the $B$-shifted R-symmetry convention, the half index can be expressed as the following contour integral
\begin{equation}
\begin{split}
    \mathcal{I}_{N,p}(\mathfrak{n}) = \frac{1}{N!}\oint_{\left(S^1\right)^N} \prod_{a=1}^N \frac{ds_a}{2 \pi i s_a} s_a^{-\mathfrak{n}} & \frac{\prod_{a \neq b}^N (s_a s_b^{-1};q)_{\infty}}{\prod_{a,b=1}^N(s_as_b^{-1} t^{-1}q;q)_{\infty}} \prod_{a,b=1}^N \frac{(s_a s_b^{-1} z t^{-\frac{1}{2}}q;q)_{\infty}}{(s_a s_b^{-1} z t^{\frac{1}{2}};q)_{\infty}} \\
    &\prod_{a=1}^N \prod_{i=1}^p \frac{(s_a x_i t^{-\frac{1}{2}}q;q)_{\infty}}{(s_a x_i  t^{\frac{1}{2}};q)_{\infty}}\,.
\end{split}
\end{equation}
This integral also appears in equation (2.11) of the work \cite{Choi:2019zpz} where (up to a shift of the $R$-symmetry) it arises as a Coulomb branch localisation formula for the disk partition function with a Neumann boundary condition and quantised FI parameter. In the present work, we instead interpret the integral as a count of boundary local operators in the presence of a Wilson line. We discuss the geometric and algebraic interpretation of this operator count in the following subsections.

\subsection{Matrix model limit}
Now we consider the limit $t \to 1$ previously discussed in detail in section \ref{subsec:vermacharacterlimit}. In this limit, the Pochhammer terms in the integrand telescope and the index becomes
\begin{equation}\label{eq:wzwintegral}
    \lim_{t\to 1}\mathcal{I}_{N,p}(\mathfrak{n}) = \frac{1}{N!}\oint_{\left(S^1\right)^N} \prod_{a=1}^N \frac{ds_a}{2 \pi i s_a} s_a^{-\mathfrak{n}}  \frac{\prod_{a \neq b}^N \left( 1-s_a s_b^{-1}\right)}{\prod_{a,b=1}^N \left(1-z s_a s_b^{-1}\right)}  \prod_{a=1}^N \prod_{i=1}^p \frac{1}{1- x_i s_a}\,.
\end{equation}
This integral can be interpreted as the partition function counting gauge invariant states in the following quantum mechanics.

\paragraph{Chern-Simons quantum mechanics}
We consider a gauged quantum mechanics with a $U(N)$ gauge symmetry. The model includes a gauge field $\alpha$, complex adjoint scalar $Z$ and $p$ fundamental scalars $\varphi_i$ with $i=1,\ldots,N$. The following action was first considered by \cite{Polychronakos:2001mi} as a matrix model description of the quantum Hall effect:
\begin{equation}
    S = \int dt \left[ i \text{tr} (Z^{\dagger} \mathcal{D}_t Z) + i \sum_{i=1}^p \varphi_i^{\dagger} \mathcal{D}_t \varphi_i - (\mathfrak{n}+p)\tr \alpha + m \tr Z Z^{\dagger} \right] \,,
\end{equation}
in the above $\mathfrak{n}$ is a positive integer that is identified with our line operator charge, and $m$ corresponds to the ADHM axial mass. The covariant derivatives act by:
\begin{equation}
\begin{split}
    \mathcal{D}_t Z &= \partial_t Z - i [\alpha,Z] \\
    \mathcal{D}_t \varphi &= \partial_t \varphi_i - i \alpha \varphi_i\,.
\end{split}    
\end{equation}
This model was further studied in \cite{Dorey:2016mxm} from the perspective of vortex dynamics in a $2+1$d Chern-Simons-matter theory, we expect this quantum mechanics to be related to the topological quantum mechanics of Higgs/Coulomb branch operators in $\Omega$-deformed $\mathcal{N}=4$ theories, as in \cite{Gaiotto:2019wcc}. In \cite{Dorey:2016hoj} the model is canonically quantised and the partition function gives rise to a contour integral expression that we identify with (\ref{eq:wzwintegral}). 

\subsection{Evaluating the partition function}
The poles of the integral (\ref{eq:wzwintegral}) are parametrised by $p$-coloured vertical Young diagrams with $N$ total boxes, specifically: $s_a=z^{-k_a}x_a^{-1}$ and $k_1+\ldots +k_p=N$. Evaluating the residues at these poles we find:
\begin{equation}
\lim_{t\to 1}\mathcal{I}_{N,p}(\mathfrak{n}) = \sum_{k_1 + \ldots + k_p = N} \prod_{i=1}^p (x_i^{k_i}z^{\frac{1}{2}k_i(k_i-1)})^{\mathfrak{n}} \prod_{i,j=1}^p \frac{1}{(z^{k_i-k_j+1}x_i/x_j;z)_{k_j}}\,.
\end{equation}
Using specialised Macdonald polynomial raising operators, see appendix \ref{appendix:handsawmilne}, one can show that this expression is in fact a polynomial in $x_i$. In particular it is a Milne polynomial in the $x_i$ variables labelled by the partition $(\mathfrak{n}^N)$ and with parameter $z$:
\begin{equation}
    \lim_{t\to 1}\mathcal{I}_{N,p}(\mathfrak{n}) = \frac{1}{(z;z)_N} Q'_{(\mathfrak{n}^N)}(x_1,\ldots,x_p;z)\,.
\end{equation}
Milne polynomials have a positive integral Schur expansion via the Kostka polynomial transition matrix:
\begin{equation}
    Q'_{\mu}(x;q) = \sum_{\lambda}K_{\mu \lambda}(q) S_{\lambda}(x)\,.
\end{equation}
Consequently, we expect this boundary condition to correspond to a simple module of Nakajima and Kodera's \cite{Kodera:2016faj} Coulomb branch algebra $\mathcal{A}_{N,p}^C$ for the ADHM quiver with $p \ge 1$ flavours. It would be interesting to study such modules abstractly.

We further support this claim by analogy with the $T[SU(N)]$ theory. In this example the Higgs branch is the cotangent bundle to a complete flag variety $\mathcal{M}_H = T^*F_N$. The holomorphic block integral for this Neumann boundary condition with an appropriate Wilson line insertion coincides with the contour integral form of a Schur polynomial $s_\lambda$ in the $G_H$ fugacities, i.e. realises a simple module for the chiral ring. Furthermore, the $T[SU(N)]$ analogy to the geometric construction in the following subsection is simply the Borel-Weil-Bott theorem where simple modules of $\mathfrak{sl}_N$ are realised as sections of holomorphic line bundles over the flag $F_N$---this is the analogous Lagrangian in the $T[SU(N)]$ case to the Hanany-Tong Lagrangian that we discuss in more detail below.\footnote{We note a curious difference with in this analogy in that the Hanany-Tong moduli space is a non-compact Lagrangian and is not the core of the Jordan quiver. Indeed, it only contains a subset of the fixed points corresponding to column Young diagrams. In the $T[SU(N)]$ Higgs branch there is one Lagrangian, the compact core, that gives rise to a finite dimensional simple module. It is unclear if the Hanany-Tong Lagrangian is the only Lagrangian in $\mathcal{M}_{N,p}$ that gives rise to simple modules.} Thus we propose that this section of the paper can be viewed as describing a Borel-Weil-Bott analogue for the ADHM chiral algebra acting on the cohomology of line bundles over the Hanany-Tong Lagrangian. Further details of the $T[SU(N)]$ example are discussed in appendix \ref{appendix:tsunexamples}.

\subsection{IR image}\label{subsec:IRimage}
There is a distinguished Lagrangian sub-manifold $\mathcal{V}_{N,p}$ of dimension $2pN = \frac{1}{2}\text{dim}(\mathcal{M}_H)$ in the ADHM quiver that we refer to as the Hanany-Tong \cite{Hanany:2003hp} vortex moduli space.\footnote{Note this has a dual life as the moduli space of vortices in (N)-[N], as studied in \cite{Crew:2020jyf}, but does not correspond to the moduli space of vortices in the ADHM theory. In our context the Hanany-Tong vortex moduli space is a Lagrangian in the Higgs branch of the ADHM theory.} The field content, as shown in figure \ref{fig:ht}, arises from the naive Lagrangian splitting of the hypermultiplets in the ADHM quiver.

This moduli space is a simple example of a handsaw quiver and inherits a group action $G = U(1) \times SU(p)$ from $\mathcal{M}_{N,p}$, we introduce fugacities $(z,x_1,\ldots,x_p)$ for $G$. The fixed points are described in \cite{Nakajima:2011yq} and are in 1-1 correspondence with $p$ column Young diagrams with $\sum k_i = N$ boxes---these fixed points are a subset of the coloured Young diagram $\lambda^{(i)}$ fixed points of $\mathcal{M}_H$ and coincide with the poles of the integral (\ref{eq:wzwintegral}). We denote a fixed point by $x(k) \in \mathcal{V}_{N,p}^G$. As a one node Nakajima quiver variety, $\mathcal{V}_{N,p}$ has a single tautological line bundle $\mathcal{D}$. The character of sections of this line bundle at a fixed point $x(k)$ is given by \cite{Crew:2020jyf}:
\begin{equation}
    \text{ch}_{x(k)}\mathcal{D} = \prod_{i=1}^p \left(x_i^{k_i}z^{\frac{1}{2}k_i(k_i-1)}\right)\,.
\end{equation}
The group action on the tangent bundle at a fixed point has character:\footnote{This can be derived using the fact that the handsaw quiver is a submanifold of instanton moduli space \cite{Nakajima:2011yq}.}
\begin{equation}
    \text{ch} T_{x(k)} \mathcal{V}_{N,p} = \sum_{i,j=1}^p \frac{x_i}{x_j} \sum_{s=0}^{k_{j}-1}q^{k_{i}-k_j+1+s} \,.
\end{equation}
Atiyah-Hirzebruch-Riemann-Roch\footnote{See \cite{Pestun:2016qko} for a review.} localisation then gives a formula for the equivariant Euler characteristic of tensor products of this line bundle:
\begin{equation}
\begin{split}
    \chi(\mathcal{D}^{\otimes \mathfrak{n}};\mathcal{V}_{N,p}) &= \sum_{{x(k)} \in \mathcal{V}_{N,p}^G}  \left(\text{ch}_{x(k)}\mathcal{D}\right)^{\mathfrak{n}} \text{PE} \left[ \text{ch}T_{x(k)}\mathcal{V}_{N,p} \right] \\
    &= \lim_{t\to 1}\mathcal{I}_{N,p}(\mathfrak{n})\,.
\end{split}    
\end{equation} 
We conclude that the half index of this boundary condition in the presence of $\mathcal{W}_{\mathfrak{n}}$ can be interpreted as a count of holomorphic sections of the tautological line bundle on the Hanany-Tong Lagrangian $\mathcal{V}_{N,p}$.
We further note that setting the Wilson line charge $\mathfrak{n}$ to zero leaves a non-trivial half index since $\mathcal{V}_{N,p}$ is non-compact and the half index recieves contributions from operators arising from the non-trivial holomorphic functions on $\mathcal{V}_{N,p}$.

\begin{figure}
\centering

\tikzset{every picture/.style={line width=0.75pt}} 

\begin{tikzpicture}[x=0.75pt,y=0.75pt,yscale=-1,xscale=1]

\draw   (275,112) .. controls (275,98.19) and (286.19,87) .. (300,87) .. controls (313.81,87) and (325,98.19) .. (325,112) .. controls (325,125.81) and (313.81,137) .. (300,137) .. controls (286.19,137) and (275,125.81) .. (275,112) -- cycle ;
\draw   (280,187) -- (320,187) -- (320,224.05) -- (280,224.05) -- cycle ;
\draw    (295,137) -- (295,184) ;
\draw [shift={(295,187)}, rotate = 270] [fill={rgb, 255:red, 0; green, 0; blue, 0 }  ][line width=0.08]  [draw opacity=0] (8.93,-4.29) -- (0,0) -- (8.93,4.29) -- cycle    ;
\draw    (320,97) .. controls (399.53,78.02) and (312.12,9.55) .. (310.04,84.68) ;
\draw [shift={(310,87)}, rotate = 270.24] [fill={rgb, 255:red, 0; green, 0; blue, 0 }  ][line width=0.08]  [draw opacity=0] (8.93,-4.29) -- (0,0) -- (8.93,4.29) -- cycle    ;
\draw    (305,187) -- (305,140) ;
\draw [shift={(305,137)}, rotate = 450] [fill={rgb, 255:red, 0; green, 0; blue, 0 }  ][line width=0.08]  [draw opacity=0] (8.93,-4.29) -- (0,0) -- (8.93,4.29) -- cycle    ;
\draw    (280,97) .. controls (201.46,77.69) and (288.56,9.55) .. (289.98,84.68) ;
\draw [shift={(290,87)}, rotate = 270.24] [fill={rgb, 255:red, 0; green, 0; blue, 0 }  ][line width=0.08]  [draw opacity=0] (8.93,-4.29) -- (0,0) -- (8.93,4.29) -- cycle    ;
\draw   (151.65,111.82) .. controls (151.65,98.01) and (162.84,86.82) .. (176.65,86.82) .. controls (190.46,86.82) and (201.65,98.01) .. (201.65,111.82) .. controls (201.65,125.62) and (190.46,136.82) .. (176.65,136.82) .. controls (162.84,136.82) and (151.65,125.62) .. (151.65,111.82) -- cycle ;
\draw   (156.65,186.82) -- (196.65,186.82) -- (196.65,223.87) -- (156.65,223.87) -- cycle ;
\draw    (175.65,186.82) -- (175.65,139.82) ;
\draw [shift={(175.65,136.82)}, rotate = 450] [fill={rgb, 255:red, 0; green, 0; blue, 0 }  ][line width=0.08]  [draw opacity=0] (8.93,-4.29) -- (0,0) -- (8.93,4.29) -- cycle    ;
\draw    (166,88) .. controls (116.95,39.93) and (224.54,36.38) .. (186.83,84.75) ;
\draw [shift={(185,87)}, rotate = 310.2] [fill={rgb, 255:red, 0; green, 0; blue, 0 }  ][line width=0.08]  [draw opacity=0] (8.93,-4.29) -- (0,0) -- (8.93,4.29) -- cycle    ;

\draw (218,148.4) node [anchor=north west][inner sep=0.75pt]  [font=\LARGE]  {$\subset $};
\draw (217,113.4) node [anchor=north west][inner sep=0.75pt]    {$\omega =0$};
\draw (297,202.4) node [anchor=north west][inner sep=0.75pt]    {$p$};
\draw (293,104.4) node [anchor=north west][inner sep=0.75pt]    {$N$};
\draw (171,199.4) node [anchor=north west][inner sep=0.75pt]    {$p$};
\draw (168.65,104.22) node [anchor=north west][inner sep=0.75pt]    {$N$};

\end{tikzpicture}

		\caption{The Hanany-Tong vortex moduli space Lagrangian $\mathcal{V}_{N,p} \subset \mathcal{M}_{N,p}$.}\label{fig:ht}
\end{figure}
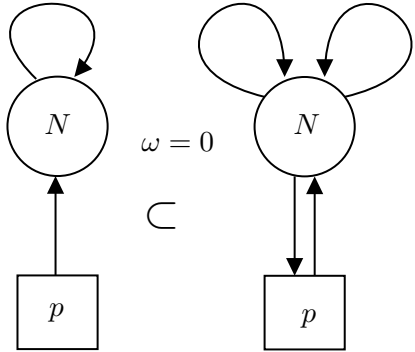


\section{Outlook}\label{sec:conclusion}
In this work we have discussed several combinatoric, geometric and algebraic aspects of the hemisphere partition functions of the worldvolume theory on a stack of $N$ M2-branes described in terms of the UV 3d ADHM theory. We conclude with some directions for further research.

\paragraph{Geometric interpretation of Cardy limit} In this work we looked at the geometric interpretation of a limit $q \to 0$. It would be interesting to investigate the geometric interpretation of the Cardy limit $q \to 1$ as recently studied in \cite{Choi:2019dfu,Choi:2019zpz}. We briefly discuss some ideas in this direction. Consider the $N=1$ case where the adjoint field in 3d ADHM decouples and the theory is $\text{SQED}[1]$. In this case, we can explicitly re-sum the vortex contributions to the block using the $q$-binomial theorem:
\begin{equation}
    \mathcal{Z}_{\text{Vortex}} = \sum_{k}(\zeta t)^k \frac{(tq;q)_{k}}{(q;q)_{k}}=\frac{(t \zeta ;q)_{\infty}}{(\zeta;q)_{\infty}} = \text{PE}\left[ \frac{1-t}{1-q}\zeta\right]\,.
\end{equation}
The $q \to 1$ limit on the right hand side of this equation can be understood as the numerical Donaldson-Thomas invariants of a particular dual quiver as described in e.g. \cite{Ekholm:2018eee}. It would be interesting to upgrade this calculation to general $N$. We also remark that in certain cases, e.g. for the column vacuum $\lambda=(1^N)$, the vortex partition function coincides with a one-point torus block with modular transformation properties in $q$, this is a possible route to relating the geometric interpretation of the $q \to 1$ and $q \to 0$ limits.

\paragraph{Simple modules of the Coulomb branch algebra} It would be interesting to study in more detail the UV Neumann boundary condition of section \ref{sec:qm} that leads to simple modules of the chiral ring. For example, to the authors' knowledge, a careful mathematical understanding of the simple modules of the Nakajima-Kodera algebra is currently lacking in the literature.

\paragraph{Cardy block and Hanany-Tong moduli space} Finally, we remark that in the works \cite{Choi:2019dfu,Choi:2019zpz} the macroscopic black hole states are dominated by a particular holomorphic block associated to the vacuum $\lambda = (1^N)$. In the general $p \ge 1$ flavour case, such column fixed points are the fixed points contained in the Hanany-Tong Lagrangian discussed in section \ref{subsec:IRimage}. The geometry of this Lagrangian appears closely connected to the Cardy block---it would be interesting to combine this observation with the fact that the Hanany-Tong Lagrangian yields simple modules of the Coulomb branch chiral algebra.

\section*{Acknowledgements}
The authors would like to thank Alec Barns-Graham, Andrey Smirnov, Yakov Kononov, Tadashi Okazaki and Gjergji Zaimi for many useful discussions. We especially thank Mathew Bullimore for helpful comments on a draft of the paper. This work has been partially supported by STFC consolidated grants ST/P000681/1, ST/T000694/1.

\appendix

\section{Combinatorics, Polynomials and Characters}\label{appendix:combinatorics}
\subsection{Partitions and Pochhammer symbols}\label{appendix:conventions}
\paragraph{Partitions}
Our conventions for partitions and statistics on partitions coincide with those of \cite{macdonald1998symmetric}. In particular, we use the ``English'' convention for the diagram associated to a partition, see figure \ref{fig:partitionconvention} for an example. Partitions are specified by their parts $\lambda=(\lambda_1,\lambda_2,\ldots)$, we make use of the shorthand $\lambda = (1^{m_1},2^{m_2},...)$ to express a partition in terms of it's multiplicities $m_i= m_{i}(\lambda)$. The transpose partition is denoted $\lambda^{\vee}$. Partitions can be written as Young diagrams in $\mathbb{Z}^2$ with boxes labelled by $s=(i,j) \in \lambda$, where $(i,j)$ run over the rows and columns respectively. The arm and leg lengths of $s \in \lambda$ are defined as follows:
\begin{equation}
\begin{split}
    a_{\lambda}(s) &= \lambda_{i} - j, \\
    l_{\lambda}(s) &= \lambda^\vee_j-i\,.
\end{split}    
\end{equation}
The hook and the content of a box $s$ are given by:
\begin{equation}
\begin{split}
    h_{\lambda}(s) &= a_{\lambda}(s) + l_{\lambda}(s) +1\,, \\
    c_{\lambda}(s) &= j-i\,.
\end{split}
\end{equation}
We use the following shorthand notation for the sums over arm and leg lengths in a partition:
\begin{equation}
\begin{split}
    n(\lambda) &= \sum_{s \in \lambda} l_{\lambda}(s)\,, \\
    n(\lambda^{\vee}) &= \sum_{s\in \lambda} a_{\lambda}(s)\,.
\end{split}    
\end{equation}
In terms of which the sums over hook and content can be expressed as:
\begin{equation}\label{eq:sumofcontentshooks}
\begin{split}
    \sum_{s\in\lambda} c_{\lambda}(s) &= n(\lambda^\vee) -n(\lambda)  = \sum_{(i,j)\in\lambda} j-i\,,\\
    \sum_{s\in \lambda} h_{\lambda}(s) &= n(\lambda)+n(\lambda^\vee)+|\lambda|=\sum_{(i,j)\in\lambda} i+j-1\,.
\end{split}
\end{equation}
We also make use of a particular weight defined by
\begin{equation}
    ||\lambda||^2 \equiv \sum_i \lambda_i^2\, .
\end{equation}
The sums over hook and content can be expressed in terms of this weight:
\begin{equation}\label{eq:hookandcontentasnorm}
\begin{split}
    \sum_{s\in\lambda} c_{\lambda}(s) &= \frac{1}{2}(||\lambda||^2 - ||\lambda^\vee||^2) \\
    \sum_{s\in \lambda} h_{\lambda}(s) &= \frac{1}{2}(||\lambda||^2+||\lambda^\vee||^2) 
\end{split}    
\end{equation}
We write $|\lambda|=\sum_{i} \lambda_i$ for the total number of boxes in a partition $\lambda$ (the weight) and $l(\lambda)$ for the length, i.e. the total number of parts. Partitions are partially ordered by the dominance ordering; we write $\lambda \le \mu$ whenever
\begin{equation}\label{eq:dominanceordering}
    \lambda_1 + \ldots + \lambda_k \le \mu_1 + \ldots + \mu_k
\end{equation}
holds for all $k \ge 1$. The sum of contents $\sum c_{\lambda}(s)$ respects this partial order since whenever $\mu$ dominates $\lambda$ we have
\begin{equation}
    \sum_{s\in\lambda} c_{\lambda}(s) \le \sum_{s\in \mu} c_{\mu}(s).
\end{equation}
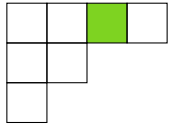
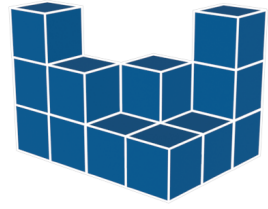
\begin{figure}
\centering
\begin{subfigure}[t]{0.45\textwidth}
\centering
\begin{tikzpicture}[x=0.75pt,y=0.75pt,yscale=-1,xscale=1]

\draw   (270,90) -- (290,90) -- (290,110) -- (270,110) -- cycle ;
\draw   (290,90) -- (310,90) -- (310,110) -- (290,110) -- cycle ;
\draw  [fill={rgb, 255:red, 126; green, 211; blue, 33 }  ,fill opacity=1 ] (310,90) -- (330,90) -- (330,110) -- (310,110) -- cycle ;
\draw   (330,90) -- (350,90) -- (350,110) -- (330,110) -- cycle ;
\draw   (270,110) -- (290,110) -- (290,130) -- (270,130) -- cycle ;
\draw   (290,110) -- (310,110) -- (310,130) -- (290,130) -- cycle ;
\draw   (270,130) -- (290,130) -- (290,150) -- (270,150) -- cycle ;
\end{tikzpicture}
\caption{The partition $\lambda=(4,2,1)$ with the box $s=(1,3)$ highlighted.}\label{fig:partitionconvention}
\end{subfigure}
\hfill
\begin{subfigure}[t]{0.45\textwidth}
\centering
\includegraphics[width=3.5cm]{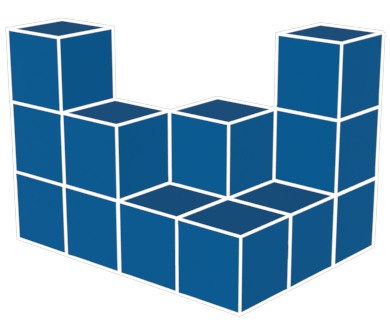}
\caption{An example reverse plane partition with base $\lambda$. Here $\pi_{(1,3)}=2$.}\label{fig:rppexample}
\end{subfigure}
\caption{Conventions for partitions and reverse plane partitions.}
\end{figure}

\paragraph{Skew diagrams}
If $\lambda$ and $\mu$ are two partitions then $\mu \subset \lambda$ means that the diagram for $\mu$ is a subset of the diagram for $\lambda$. The set $\lambda/\mu = \{\theta_i = \lambda_i-\mu_i \enspace | \enspace i=1,2,\ldots \}$ is called a \textit{skew diagram}. A skew diagram $\theta$ is \textit{connected} if all of the boxes in $\theta$ share at least one common side. We say $\theta$ is a \textit{border strip} of a partition $\lambda$ if $\theta$ is contained in $\lambda$ and $\theta$ is connected with no $2\times2$ blocks of boxes. The height $\text{ht}$ of a border strip is defined to be one less than the number of rows it occupies, similarly $\text{ht}'$ is one less than the number of columns it occupies. We further say $\theta$ is a \textit{maximal border strip} of $\lambda$ if the box $s=(i,j)\in \theta$ with maximal content is such that $s=(i+1,j)$ is not in $\lambda$ and the box $s'=(i',j')$ with minimal content is such that $(i',j'+1)$ is not in $\lambda$. 

Every skew diagram $\lambda/\mu$ can be uniquely decomposed into maximal border strips. We define $b(\lambda/\mu)$ to be the number of maximal border strips in this decomposition. The height of a skew diagram $\text{ht}(\lambda/\mu)$ is the sum of the heights of the maximal border strips in the decomposition of $\lambda/\mu$, similarly $\text{ht}(\lambda/\mu)$ is the sum of all the primed heights $\text{ht}'$ of the maximal border strips in the decomposition.

\paragraph{Reverse plane partitions}\label{appendix:hillman}
A reverse plane partition (RPP) $\pi$ with base $\lambda$ is a 3d partition with non-negative integer heights $\pi_s$ above each box $s \in \lambda$ such that $\pi_s$ weakly decrease along the rows and columns of $\lambda$, an example is shown in figure \ref{fig:rppexample}. We write $|\pi|=\sum_{s\in\lambda} \pi_s$ for the total number of boxes in the reverse plane partition.
An RPP can be thought of in terms of layers of skew shapes $\lambda/\mu_i$, with $i=1,2,\ldots$ stacked on top of each other. We define the following statistics on RPPs in terms of their skew counterparts as follows:
\begin{equation}
\begin{split}
    \text{ht}(\pi) &= \sum_{i\ge1} \text{ht}(\lambda/\mu_i), \\
    \text{ht}'(\pi) &= \sum_{i\ge1} \text{ht}'(\lambda/\mu_i), \\
    b(\pi) &= \sum_{i\ge1} b(\lambda/\mu_i)\,.
\end{split}    
\end{equation}
These statistics are involved in the refined sum over RPPs in section \ref{subsec:pplimit}.

\paragraph{Pochhammer symbols}
In this section we summarise the $q$-Pochhammer function definitions and identities used throughout the work. The $q$-Pochhammer symbol, convergent for $|q|<1$, is defined by:
\begin{equation}\label{eq:qpochhammer}
    (x;q)_{\infty} = \prod_{i=0}^{\infty}(1-q^i x).
\end{equation}
The analytic continuation for $|q|>1$ is:
\begin{equation}\label{eq:qpochhameranalyticcontinuation}
    (x;q^{-1})_{\infty} = \prod_{i=0}^{\infty}\frac{1}{1-q^{i+1}x}\,.
\end{equation}
The finite $q$-Pochhammer symbol is defined by:
\begin{equation}
    (x;q)_{n} = \frac{(x;q)_{\infty}}{(x q^n;q)_{\infty}}.
\end{equation}
For integer $n$ this expression reduces to a finite product:
\begin{equation}
    (x;q)_n = 
    \begin{cases} 
    \prod_{i=0}^{n-1}(1-q^i x) &\text{ if } n\geq0, \\
    \prod_{i=1}^{|n|}\frac{1}{1-q^{-j}x} &\text{ if } n < 0.
    \end{cases}
\end{equation}
We note the identity:
\begin{equation}\label{eq:twistedindexfusionrule}
    \left(xq^{\frac{n}{2}};q\right)_{\infty}\left(xq^{-\frac{n}{2}};q^{-1}\right)_{\infty} = \left(xq^{\frac{n}{2}};q\right)_{-n+1}
\end{equation}

\subsection{Symmetric functions and characters}\label{appendix:symmetricfuns}
In this subsection we review symmetric functions and Macdonald polynomials following \cite{macdonald1998symmetric}. We also prove a generalisation of the usual Macdonald Cauchy identity and discuss the connection between Milne polynomials and characters of Kirrilov-Reshetikhin modules.

\paragraph{Symmetric Functions}
We begin with a review of symmetric functions. We denote by $\Lambda_N = \mathbb{Z}[x_1,\ldots,x_N]^{S_N}$ the ring of symmetric functions in $N$ variables, the set of variables is denoted $X = \{ x_1,\ldots,x_N\}$. $\Lambda$ denotes the ring of symmetric functions in infinitely many variables, understood as infinite formal sums of monomials. The ring of symmetric functions over $\mathbb{Q}$, denoted $\Lambda_{\mathbb{Q}}$, is generated by the power sum symmetric functions $p_n$:
\begin{equation}
    p_n(X) \equiv \sum_{i\ge1} x_i^n\,.
\end{equation}
Any function $f\in\Lambda_{\mathbb{Q}}$ can be expanded in power sum symmetric functions. We occasionally make use of the simple plethystic substitution where $f\left(\frac{X}{1-q}\right)$ means in the power sum expansion we replace:
\begin{equation}
    p_n(X) \to \frac{1}{1-q^n}p_{n}(X)\,.
\end{equation}
Monomial symmetric functions are a basis of $\Lambda_N$ labelled by partitions $\lambda$ with $l(\lambda) \le N$ defined by:
\begin{equation}
    m_{\lambda}(X) = \sum x_1^{\lambda_1}\ldots x_N^{\lambda_N}
\end{equation}
where the sum is taken over all permutations of $\lambda=(\lambda_1,\ldots,\lambda_N)$.

A ubiquitous basis for the ring of symmetric functions is given by the Schur polynomials. Schur polynomials are labelled by a partition $\lambda$ with $l(\lambda) \le N$ and are defined by:
\begin{equation}
    s_{\lambda}(X) \equiv \frac{\text{det}_{1\le i,j \le N}\left(x_i^{\lambda_j+N-j}\right)}{\text{det}_{1 \le i,j \le N} \left(x_i^{N-j}\right)}\,.
\end{equation}
Schur polynomials are homogeneous of degree $|\lambda|$. The transition matrix between the Schur basis and the monomial basis defines the Kostka numbers $K_{\lambda \mu}$:
\begin{equation}\label{eq:kostkanumbers}
    s_{\lambda} = \sum_{\mu} K_{\lambda \mu} m_{\mu}\,.
\end{equation}

\paragraph{Macdonald Polynomials}
Macdonald polynomials $P_{\lambda}(X;q,t)$ are two parameter generalisations of the Schur and monomial symmetric functions. They are symmetric functions in $\Lambda_{q,t} = \Lambda \otimes_{\mathbb{Z}} \mathbb{Q}(q,t)$ and homogeneous of degree $|\lambda|$. Macdonald \cite{macdonald1998symmetric} proves existence and uniqueness theorems for these polynomials in terms of their monomial expansion and orthogonality properties. In this subsection we focus on the properties of Macdonald polynomials relevant to the present work. The degeneration limits to Schur and monomial symmetric functions are as follows:
\begin{equation}
\begin{split}
    &P_{\lambda}(X;q,q) = s_{\lambda}\,, \\
    &P_{\lambda}(X;q,1) = m_{\lambda}\,.
\end{split}
\end{equation}
Macdonald polynomials also degenerate to the one parameter Hall-Littlewood polynomials $P_{\lambda}(X;t)$ in the limit $q\to0$. Hall-Littlewood polynomials enjoy an explicit sum formula:
\begin{equation}\label{eq:HLpolys}
    P_{\lambda}(X;0,t)=P_{\lambda}(X;t) = \prod_{i\ge0}\prod_{j=1}^{m_i(\lambda)}\frac{1-t}{1-t^j} \sum_{\sigma \in S_N} x_1^{\lambda_1} \ldots x_N^{\lambda_N} \prod_{i<j}\frac{1-tx_j/x_i}{1-x_j/x_i}\,.
\end{equation}
where the permutations $\sigma \in S_N$ act on the variables $X$.

Macdonald polynomials satisfy a Cauchy identity:
\begin{equation}\label{eq:macdonaldcauchy}
\begin{split}
    \sum_{\lambda}P_{\lambda}(X;q,t)Q_{\lambda}(Y;q,t) &= \prod_{\substack{x \in X \\ y\in Y}} \frac{(txy;q)_{\infty}}{(xy;q)_{\infty}} = \exp(\sum_{n>0}\frac{1}{n}\frac{1-t^n}{1-q^n}p_n(X)p_n(Y)) \\ &\equiv \Pi_{q,t}(X,Y)\,.
\end{split}    
\end{equation}
where $Q_{\lambda}(X;q,t)$ is a modified normalisation of the Macdonald polynomial given by:
\begin{equation}
    Q_{\lambda}(X;q,t)=b_{\lambda}(q,t)P_{\lambda}(X;q,t)
 \end{equation}
and the normalisation constant is defined as follows:
\begin{equation}
\begin{split}
    c_{\lambda}(q,t) &= \prod_{s\in\lambda}\left(1-q^{a_{\lambda}(s)}t^{l_{\lambda}(s)+1}\right)\,, \\
    c_{\lambda^{\vee}}(q,t) &= \prod_{s\in\lambda}\left(1-q^{a_{\lambda}(s)+1}t^{l_{\lambda}(s)}\right) \,,\\
    b_{\lambda}(q,t) &\equiv \frac{c_{\lambda}(q,t)}{c_{\lambda^\vee}(q,t)}\,.
\end{split}    
\end{equation}
These normalisation constants have finite limits as $q \to 0$ and we define the normalised Hall-Littlewood polynomial $Q_{\lambda}(X;t)$ similarly.

One can define an inner product on $\Lambda_{q,t}$ as follows:
\begin{equation}\label{eq:macdonaldinnerproduct}
    \langle f,g \rangle_{q,t} = \oint d\mu\left[X;q,t\right]f(\bar{X})g(X)\,,
\end{equation}
where the contour is a product of unit circles and here and throughout this appendix the variables $\bar{X}$ denote the set of inverse variables $\bar{X}=\{x_1^{-1},\ldots x_N^{-1}\}$. The Macdonald measure $d\mu$ is defined as follows:
\begin{equation}\label{eq:macdonaldmeasure}
    d\mu\left[X;q,t\right] = \frac{1}{N!}\prod_{i=1}^N \frac{dx_i}{2 \pi x_i} \prod_{i\neq j}^N \frac{(x_i/x_j;q)_\infty}{(t x_i/x_j;q)_{\infty}} 
\end{equation}
and the normalisation constant is:
\begin{equation}\label{eq:macdonorm}
    \langle P_{\lambda},P_{\mu}\rangle_{q,t}= \frac{\tilde{c}_N(\lambda;q,t)}{b_{\lambda}(q,t)} \delta_{\lambda \mu}\,.
\end{equation}
We do not discuss $\tilde{c}_N$ for general $q,t$ in this work but we do make use of the Hall-Littlewood limit:
\begin{equation}\label{eq:HLlimit}
    \lim_{q\to 0} \tilde{c}_N(\lambda;q,t)=\frac{(1-t)^N}{(t;t)_{N-l(\lambda)}}\,.
\end{equation}

\paragraph{Integral representation}
Macdonald polynomials can be realised explicitly as iterated contour integrals \cite{Awata:1995eh}. The integral is constructed inductively using the following two observations:
\begin{equation}\label{eq:macdonaldfacts}
\begin{split}
    P_{\lambda+(s^r)}(x_1,\ldots,x_r;q,t) &= x_1^s \ldots x_r^s P_{\lambda}(x_1,\ldots,x_r;q,t) \\
    P_{\lambda}(x_1,\ldots,x_n;q,t) &= C_{\lambda}(q,t) \oint d\mu[w_1,\ldots,w_m] \Pi_{q,t}(\bar{W},X) P_{\lambda}(w_1,\ldots,w_m;q,t)
\end{split}    
\end{equation}
where $C_{\lambda}(q,t)$ is a constant that we do not consider in generality in this work. Using these identities, \cite{Awata:1995eh} prove that Macdonald polynomials can be expressed as:
\begin{equation}\label{eq:macdonaldintegralrepresentation}
    P_{\lambda}(X;q,t)=C(q,t)\oint \prod_{a=1}^N d\mu[W^{(a)}] \Pi_{q,t}(W^{(a+1)},\bar{W}^{(a)}) \prod_{j=1}^{r_a} (w_{j}^{(a)})^{s_a} 
\end{equation}
where each set $W^{(a)}$ consists of $r_a$ integration variables and the last set $W^{(N+1)}$ is identified with the variables $X$. The partition $\lambda$ can be expressed as:
\begin{equation}
    \lambda = (s_N^{r_N}) + \ldots + (s_1^{r_1})\,.
\end{equation}

\paragraph{Skew Macdonald polynomials}
Macdonald polynomials form an algebra with structure constants $f_{\mu \nu}^{\lambda}(q,t)$ defined as follows:
\begin{equation}\label{eq:macdonaldstructureconstants}
    P_{\mu}P_{\nu} = \sum_{\lambda}f_{\mu \nu}^{\lambda} P_{\lambda}\,.
\end{equation}
The structure constants vanish unless $\mu \subset \lambda$ and $\nu \subset \lambda$ and under those conditions, skew Macdonald polynomials are defined by:
\begin{equation}
    Q_{\lambda/\mu}(X;q,t) \equiv \sum_{\nu} f^{\lambda}_{\mu \nu}(q,t)Q_{\nu}(X;q,t)\,.
\end{equation}
These polynomials are homogeneous of degree $|\lambda|-|\mu|$. An alternative normalisation is given by:
\begin{equation}
    P_{\lambda/\mu}(X;q,t) = \frac{b_{\mu}(q,t)}{b_{\lambda}(q,t)}Q_{\lambda/\mu}(X;q,t)\,.
\end{equation}
Skew Macdonald polynomials satisfy a skew Cauchy identity:
\begin{equation}\label{eq:skewcauchy}
    \sum_{\lambda} P_{\rho/\lambda}\left(X;q,t\right)Q_{\rho/\mu}\left(Y;q,t\right) = \Pi_{q,t}\left(X,Y\right) \sum_{\rho} P_{\mu/\rho}\left(X;q,t\right) Q_{\lambda/\rho}\left(Y;q,t\right)\,.
\end{equation}

\begin{lemma}
Skew Macdonald polynomials satisfy the following generalised Cauchy identity:
\begin{equation}
    \sum_{\lambda,\mu}A^{|\lambda|}Q_{\mu/\lambda}(X;q,t)P_{\mu/\lambda}(Y;q,t) = \prod_{k=0}^{\infty}\frac{1}{1-A^{k+1}}\Pi_{q,t}\left(A^kX,Y\right)\,.
\end{equation}
\end{lemma}
\begin{proof}
The method of proof used here is an adaptation of the Schur case found in exercise (28) of Chapter II.5 in Macdonald \cite{macdonald1998symmetric}. We let
\begin{equation}
    F(X,Y;q,t) = \sum_{\lambda,\mu}A^{|\lambda|}Q_{\mu/\lambda}(X;q,t)P_{\mu/\lambda}(Y;q,t)\,.
\end{equation}
Using the identity (\ref{eq:skewcauchy}) and the fact that Macdonald polynomials are homogeneous we can perform the sum over $\mu$ to find:
\begin{equation}
    F(X,Y;q,t) = \Pi_{q,t}(X,Y)\sum_{\lambda,\mu}A^{|\mu|} Q_{\lambda/\mu}\left(AX ;q,t\right)P_{\lambda/\mu}(Y;q,t)\,.
\end{equation}
In other words:
\begin{equation}
    F(X,Y;q,t) = \Pi_{q,t}(X,Y)F(AX,Y;q,t)\,.
\end{equation}
Now, provided $|A|<1$, we can iterate this relation to find:
\begin{equation}
    F(X,Y;q,t) = F(0,Y;q,t)\prod_{k=0}^{\infty}\Pi_{q,t}(A^k X,Y) \,.
\end{equation}
Using the fact that $P_{\lambda/\mu}$ vanishes unless $\mu \subset \lambda$ together with the fact $P_{\lambda/\mu}(0)$ vanishes unless $\mu=\lambda$ (where it equals $1$) we find:
\begin{equation}
    F(0,Y;q;t) = \sum_{\lambda}A^{|\lambda|} = \prod_{k=1}^{\infty}\frac{1}{1-A^k}\,.
\end{equation}
The lemma then follows.
\end{proof}
Later, in appendix \ref{appendix:molienintegrals}, we make use of a plethystically substituted form of this result. Under the power sum replacements:
\begin{equation}
    X \to \frac{X}{1-t}, \quad Y \to \frac{Y}{1-t},
\end{equation}
The Cauchy kernel becomes:
\begin{equation}
    \Pi_{q,t}\left(\frac{X}{1-t},\frac{Y}{1-t}\right) = \exp(\sum_{n>0}\frac{1}{n}\frac{1}{(1-t^n)(1-q^n)}p_n(X)p_n(Y)) = \prod_{\substack{x\in X\\y\in Y}}\prod_{k,l=0}^{\infty}\frac{1}{1-q^k t^l x y}
\end{equation}
and the analogous generalised Cauchy identity is then:
\begin{equation}\label{eq:generalisedcauchyplethy}
    \sum_{\lambda,\mu}A^{|\lambda|}Q_{\mu/\lambda}\left(\frac{X}{1-t};q,t\right)P_{\mu/\lambda}\left(\frac{Y}{1-t};q,t\right) = \prod_{k=0}^{\infty}\frac{1}{1-A^{k+1}}\prod_{\substack{x\in X\\y\in Y}}\prod_{l,m=0}^{\infty} \frac{1}{1-q^lt^m A^k xy}\,.
\end{equation}

\paragraph{Principal specialisation}
When the parameter $|t|<1$ Macdonald polynomials have a principal specialisation formula:
\begin{equation}
    P_{\lambda}(1,t,t^2,\ldots;q,t) = t^{n(\lambda)} \prod_{s\in \lambda} \frac{1}{1-q^{a_{\lambda}(s)}t^{l_{\lambda}(s)+1}}
\end{equation}

\paragraph{Milne polynomials}
We consider an additional normalisation of Macdonald polynomials:
\begin{equation}
    J_{\lambda}(X;q,t) = c_{\lambda}(q,t)P_{\lambda}(X;q,t)\,.
\end{equation}
Together with a plethystic substitution, these Macdonald polynomials have a positive integral Schur expansion in terms of $(q,t)$ Kostka polynomials:
\begin{equation}
    J_{\lambda}\left(\frac{1}{1-t}X;q,t\right) = \sum_{\mu} K_{\mu \lambda}(q,t) s_{\mu}(X)\,.
\end{equation}
Setting $q=0$ degenerates $K_{\mu \lambda}(q,t)$ to the Kostka polynomial $K_{\mu \lambda}(t)$, which itself is a $t$-deformation of the Kostka numbers (\ref{eq:kostkanumbers}), and we recover the one parameter Milne polynomials:
\begin{equation}
    Q'_{\lambda}(X;t) \equiv \sum_{\mu} K_{\mu \lambda}(t) s_{\mu}(X)\,.
\end{equation}
Equivalently, the Milne polynomials $P_{\lambda}'$ and $Q'_{\lambda}$ can be understood as plethystic substitutions $X \to \frac{X}{1-t}$ in the appropriately normalised Hall-Littlewood polynomials (\ref{eq:HLpolys}).

\paragraph{Kirrilov-Reshetikhin characters}
In this section we review the difference operators of \cite{di2018difference} and discuss the connection of this work to Milne polynomials. Milne polynomials realise graded characters of Kirrilov-Reshetikhin modules of $U_{q}(\hat{\mathfrak{sl}}_{N+1})$, these modules are specified by a set of non-negative integers:
\begin{equation}
    \vec{n} = \{ n_{l}^{(\alpha)} : 1 \le l \le k, 1 \le \alpha \le N\}
\end{equation}
and they have a tensor decomposition into $\mathfrak{sl}_{N+1}$ modules as follows:
\begin{equation}
    V = \bigotimes_{1 \le \alpha \le N} \bigotimes_{1 \le l \le k} V(l \omega_\alpha)^{n_l^{(\alpha)}}
\end{equation}
where $\omega_i$ are fundamental weights of $\mathfrak{sl}_{N+1}$.
In this work we consider only the case $\alpha=1$,\footnote{This is the opposite case considered in \cite{di2018difference} where in this work $k=1$ and they find $q$-Whittaker functions which can be realised instead as involution Milne polynomials $\iota Q'$ where $\iota$ acts on power sums by $p_n \to -p_n$} the highest weight is specified by a partition formed from ordering the $\vec{n}=n_l^{(1)}$, we denote this partition by $\lambda$. In this case, the Kostka polynomial $K_{\lambda \mu}(q)$ gives the graded multiplicity of the $\mathfrak{sl}_{N+1}$ representation associated to $\mu$ in the tensor decomposition and the graded character is then identified with the Milne polynomial:
\begin{equation}
    Q'_{\lambda}(X;q) = \chi_{\vec{n}}(X;q).
\end{equation}
The characters of Kirrilov-Reshetikhin modules satisfy the quantum Q-system relations. It follows that graded characters can be constructed iteratively from the raising operators introduced in \cite{di2018difference}:
\begin{equation}\label{eq:raisingop}
    D_{\alpha,n} = \sum_{\substack{I \subset [1,N+1] \\ |I|=\alpha}} X_I^n \prod_{\substack{i \in I \\ j \notin I}}\frac{1}{1-x_j/x_i} \Gamma_{q;I}
\end{equation}
where $X_I$ denotes the multiset $\{x_{i_1} \ldots x_{i_{|I|}}\}$ with $i_k \in I$ and $\Gamma_{q;i}$ is the shift operator acting on the variables $X$ as follows:
\begin{equation}
    \Gamma_{q;i}(x_1,\ldots, x_{N+1}) = (x_1,\ldots,q x_i,\ldots, x_{N+1})\,.
\end{equation}
Milne polynomials can then be constructed, up to a constant in $q$, from the $\alpha=1$ raising operators as follows:
\begin{equation}
    Q'_{\lambda}(X;q) = D_{1,k}^{n_k}D_{1,k-1}^{n_{k-1}}\ldots D_{1,1}^{n_1}.1\,.
\end{equation}
Indeed, these raising operators correspond to the Milne degeneration of the Macdonald polynomial raising operators of \cite{kirillov1996affine}. We use this formalism to understand the Euler characteristic of line bundles on the Hanany-Tong moduli space in appendix \ref{appendix:symmetricfunctionmethods}.

\paragraph{Plethystic exponential}
The plethystic exponential of a function $f(t_1,\ldots,t_N)$ is defined formally by
\begin{equation}\label{eq:plethysticexp}
    \text{PE} \left[f\right](t_1,\ldots,t_n) = \exp \left( \sum_{n=0}^{\infty} \frac{f(t_1^n,\ldots, t_N^n)}{n}\right) \,.
\end{equation}

\section{Symmetric Function Methods}\label{appendix:symmetricfunctionmethods}
In this appendix we discuss symmetric function methods to evaluate particular enumerative invariants of quiver varieties. We consider the equivariant Euler characteristic of line bundles over a simple handsaw quiver and the Hilbert series of an arbitrary chainsaw quiver, both expressed as Molien integrals.

\subsection{Handsaw quivers and Milne polynomials}\label{appendix:handsawmilne}
In section \ref{subsec:IRimage} we found the following form for the Euler characteristic of line bundles on the Hanany-Tong moduli space:
\begin{equation}
    \chi(\mathcal{D}^{\otimes \mathfrak{n}};\mathcal{V}_{N,p}) = \mathcal{I}_{N,p}(\mathfrak{n}) = \sum_{k_1 + \ldots + k_p = N} \prod_{i=1}^p (x_i^{k_i}z^{\frac{1}{2}k_i(k_i-1)})^{\mathfrak{n}} \prod_{i,j=1}^p \frac{1}{(z^{k_i-k_j+1}x_i/x_j;z)_{k_j}} \,.
\end{equation}

\begin{lemma}
The Euler characteristic of the tautological line bundle over the Hanany-Tong Lagrangian is a particular Milne polynomial with highest weight $(\mathfrak{n}^N)$. In terms of repeated iteration of the raising operators (\ref{eq:raisingop}) we have:
\begin{equation}
    \mathcal{I}_{N,p} = \frac{1}{(z;z)_N} \left( D_{1,\mathfrak{n}}\right)^N.1 = \frac{1}{(z;z)_N} Q'_{(\mathfrak{n}^N)}(x_1,\ldots,x_p;z) \,.
\end{equation}
\end{lemma}

\begin{proof}
We proceed inductively on $N$. When $N=1$ the sum over $\{k_i\}$ in $\mathcal{I}_{N,p}$ is a choice of which $k_i$ is set equal to $1$. Further, the first product over $i$ becomes simply $x_i^{\mathfrak{n}}$ for this choice $k_i$ and the second product receives contributions only from terms involving the non-zero $k_i$, and brings out a factor of $1/(1-z)$---thus we find:
\begin{equation}
    \mathcal{I}_{1,p} = \frac{1}{1-z}\sum_{i=1}^p x_i^{\mathfrak{n}} \prod_{\substack{j=1 \\ j\neq i}}^{p}\frac{1}{1-x_i/x_j}\,.
\end{equation}
This coincides with the raising operator $D_{1,\mathfrak{n}}$ divided by $(z;z)_1$ as required. 

Now we act with $D_{1,\mathfrak{n}}$ on $\mathcal{I}_{N,p}$. Firstly, we consider the action of the shift operator $\Gamma_{z;i}$ on the summand. We denote the summand by $\mathcal{I}_{N,p}^{\{k\}}$ so that:
\begin{equation}
    \mathcal{I}_{N,p} = \sum_{k_1+\ldots+k_p=N} \mathcal{I}_{N,p}^{\{k\}}\,.
\end{equation}
The shift operator acts on the summand as follows:
\begin{equation}
    \Gamma_{q,i}\left(\mathcal{I}_{N,p}^{\{k\}} \right) = x_i^{-\mathfrak{n}}\mathcal{I}_{N,p}^{\{\tilde{k}^{(i)}\}} \prod_{j=1}^p \left( 1-z^{\tilde{k}_j^{(i)}}x_j/x_i\right)\,.
\end{equation}
The set of integers $\{\tilde{k}^{(i)}\}$ is the same as the set $\{k\}$ except the $i^{\text{th}}$ integer is shifted by $1$ i.e. $\tilde{k}_i = k_i+1$. Now applying the whole raising operator (\ref{eq:raisingop}) we have:
\begin{equation}
\begin{split}
    D_{1,\mathfrak{n}} \mathcal{I}_{N,p} = \sum_{i=1}^p x_i^{\mathfrak{n}} \prod_{\substack{j=1 \\ j\neq i}}^p \frac{1}{1-x_j/x_i}  \sum_{\{k\}} \left[ x_i^{-\mathfrak{n}}\mathcal{I}_{N,p}^{\{\tilde{k}^{(i)}\}} \prod_{j=1}^p \left( 1-z^{\tilde{k}_j^{(i)}}x_j/x_i\right) \right]\,.
\end{split}    
\end{equation}
Now we seek to change variable in the sum over $\{k\}$. We can reparametrise the sum as a sum over $\{k'\}$ with $\sum_{i=1}^p k'_i = N+1$ but with $k'_i \ge 1$. Now the term in square brackets vanishes if $k'_i=0$ so we can write the expression as a sum over all $\{k' \}$ with $\sum_{i}k'_i=N+1$. The result is then an expression:\footnote{Relabelling $\{k'\} \to \{k\}$ for ease of notation.}
\begin{equation}
    D_{1,\mathfrak{n}} \mathcal{I}_{N,p} = \sum_{k_1+\ldots+k_p=N+1} \mathcal{I}_{N+1,p}^{\{k\}} \left[ \sum_{i=1}^p (1-z^{k_i}) \prod_{\substack{j=1 \\ j \neq i}} \frac{1-z^{k_j}x_j/x_i}{1-x_j/x_i} \right]\,.
\end{equation}
One can verify that the term in square brackets is in fact independent of $x_i$ and gives simply $1-z^{\sum_{i=1}^p k_i} = 1-z^{N+1}$ thus completing the proof:
\begin{equation}
    D_{1,\mathfrak{n}} \mathcal{I}_{N,p} = \mathcal{I}_{N+1,p}.
\end{equation}
\end{proof}

\subsection{Molien integral symmetric functionology}\label{appendix:molienintegrals}

In this appendix we use symmetric function methods to evaluate the following integral:
\begin{equation}
\begin{split}
    &\mathcal{I}^N_{\{N_a\} \{k_a\}}(Z^{(a)};\zeta_a;\alpha_a,\beta_a,\gamma_a;q,t) \\= &\prod_{a=1}^N \frac{1}{k_a!(t;q)_{\infty}^{k_a}}\oint \prod_{a=1}^N\prod_{i=1}^{k_a} \frac{d w^{(a)}_i}{2 \pi i w^{(a)}_i} \left( w_i^{(a)}\right)^{-\zeta_a} \prod_{a=1}^N \prod_{i \neq j}^{k_a}  \frac{(w_i^{(a)}/w_j^{(a)};q)_{\infty}}{(tw_i^{(a)}/w_j^{(a)};q)_{\infty}}\\&\prod_{a=1}^N \prod_{i=1}^{k_a} \prod_{j=1}^{k_{a+1}}\frac{(\gamma_a t w_i^{(a)}/w_j^{(a+1)};q)_{\infty}}{(\gamma_a  w_i^{(a)}/w_j^{(a+1)};q)_{\infty}}
     \prod_{a=1}^N  \prod_{i=1}^{k_a}\prod_{m=0}^{N_a} \prod_{n=0}^{N_{a+1}} \frac{1}{(\alpha_a w_i^{(a)}z_m^{(a)};q)_{\infty}}\frac{1}{(\beta_a \frac{1}{w_i^{(a)}z_n^{(a+1)}};q)_{\infty}}\,.
\end{split}    
\end{equation}
The data in this integral corresponds to the chainsaw quiver in figure \ref{fig:chainsaw}. The quiver has $N$ gauge nodes with gauge ranks $k_a$ and $N$ flavour nodes $N_a$, the $(N+1)^{\text{th}}$ flavour is identified with the $1^{\text{st}}$ flavour, and the flavour fugacities are grouped into $N$ sets $Z^{(a)}$ each with $N_a$ variables. We take the contour to be a product of unit circles and $W^{(a)}$ are $N$ sets of $k_a$ integration variables parametrising unit circles. The integral has positive integer parameters $\zeta_a$ that specify tensor powers of tautological line bundles over the chainsaw quiver. The integral also depends on the set of auxiliary parameters $\{q,t; \alpha_a,\beta_a,\gamma_a\}$. The variables $W^{(N+1)}$ are identified with $W^{(1)}$ and $Z^{(N+1)}$ are similarly identified with $Z^{(1)}$.

\begin{figure}
\centering
\begin{tikzpicture}[x=0.75pt,y=0.75pt,yscale=-1,xscale=1]

\draw   (95,95) .. controls (95,81.19) and (106.19,70) .. (120,70) .. controls (133.81,70) and (145,81.19) .. (145,95) .. controls (145,108.81) and (133.81,120) .. (120,120) .. controls (106.19,120) and (95,108.81) .. (95,95) -- cycle ;
\draw   (195,95) .. controls (195,81.19) and (206.19,70) .. (220,70) .. controls (233.81,70) and (245,81.19) .. (245,95) .. controls (245,108.81) and (233.81,120) .. (220,120) .. controls (206.19,120) and (195,108.81) .. (195,95) -- cycle ;
\draw    (145,95) -- (193,95) ;
\draw [shift={(195,95)}, rotate = 180] [color={rgb, 255:red, 0; green, 0; blue, 0 }  ][line width=0.75]    (10.93,-3.29) .. controls (6.95,-1.4) and (3.31,-0.3) .. (0,0) .. controls (3.31,0.3) and (6.95,1.4) .. (10.93,3.29)   ;
\draw   (100,160) -- (140,160) -- (140,200) -- (100,200) -- cycle ;
\draw   (200,160) -- (240,160) -- (240,200) -- (200,200) -- cycle ;
\draw    (120,160) -- (120,122) ;
\draw [shift={(120,120)}, rotate = 450] [color={rgb, 255:red, 0; green, 0; blue, 0 }  ][line width=0.75]    (10.93,-3.29) .. controls (6.95,-1.4) and (3.31,-0.3) .. (0,0) .. controls (3.31,0.3) and (6.95,1.4) .. (10.93,3.29)   ;
\draw    (140,110) -- (218.3,158.94) ;
\draw [shift={(220,160)}, rotate = 212.01] [color={rgb, 255:red, 0; green, 0; blue, 0 }  ][line width=0.75]    (10.93,-3.29) .. controls (6.95,-1.4) and (3.31,-0.3) .. (0,0) .. controls (3.31,0.3) and (6.95,1.4) .. (10.93,3.29)   ;
\draw    (220,160) -- (220,122) ;
\draw [shift={(220,120)}, rotate = 450] [color={rgb, 255:red, 0; green, 0; blue, 0 }  ][line width=0.75]    (10.93,-3.29) .. controls (6.95,-1.4) and (3.31,-0.3) .. (0,0) .. controls (3.31,0.3) and (6.95,1.4) .. (10.93,3.29)   ;
\draw  [dash pattern={on 0.84pt off 2.51pt}]  (245,95) -- (285,95) ;
\draw  [dash pattern={on 0.84pt off 2.51pt}]  (55,95) -- (95,95) ;
\draw    (230,70) .. controls (260.2,30.4) and (181.11,30) .. (209.12,68.81) ;
\draw [shift={(210,70)}, rotate = 232.67000000000002] [color={rgb, 255:red, 0; green, 0; blue, 0 }  ][line width=0.75]    (10.93,-3.29) .. controls (6.95,-1.4) and (3.31,-0.3) .. (0,0) .. controls (3.31,0.3) and (6.95,1.4) .. (10.93,3.29)   ;
\draw    (130,70) .. controls (160.2,30.4) and (81.11,30) .. (109.12,68.81) ;
\draw [shift={(110,70)}, rotate = 232.67000000000002] [color={rgb, 255:red, 0; green, 0; blue, 0 }  ][line width=0.75]    (10.93,-3.29) .. controls (6.95,-1.4) and (3.31,-0.3) .. (0,0) .. controls (3.31,0.3) and (6.95,1.4) .. (10.93,3.29)   ;

\draw (113,86.4) node [anchor=north west][inner sep=0.75pt]    {$k_{a}$};
\draw (207,85.4) node [anchor=north west][inner sep=0.75pt]    {$k_{a+1}$};
\draw (110,171.4) node [anchor=north west][inner sep=0.75pt]    {$N_{a}$};
\draw (204,171.4) node [anchor=north west][inner sep=0.75pt]    {$N_{a+1}$};

\end{tikzpicture}

		\caption{Periodic chainsaw quiver with gauge nodes $\{k_a\}$ and flavour nodes $\{N_a\}$.}\label{fig:chainsaw}
\end{figure}
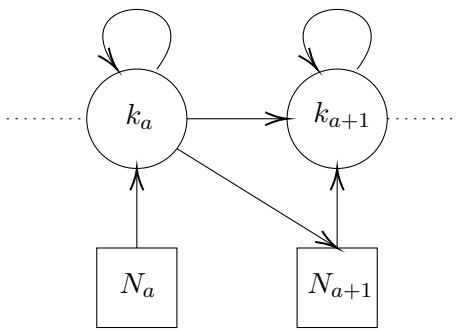

Using the Macdonald measure (\ref{eq:macdonaldmeasure}) and the Macdonald Cauchy identity (\ref{eq:macdonaldcauchy})\footnote{We use a plethystically substituted form $Z \to \frac{Z}{1-t}$ of this identity for the flavour terms.}, we can re-write the integrand in terms of symmetric functions:
\begin{equation}
\begin{split}
    &\mathcal{I}^N_{\{N_a\} \{k_a\}}(Z^{(a)};\zeta_a;\alpha_a,\beta_a,\gamma_a;q,t) = \\
    &\prod_{a=1}^N \frac{1}{(t;q)_{\infty}^{k_a}}\oint \prod_{a=1}^N d\mu[W^{(a)};q,t]  \prod_{i=1}^{k_a}\left( w^{(a)}_i\right)^{-\zeta_a} \sum_{\{\lambda^{(a)}\}} P_{\lambda^{(a)}}\left( W^{(a)};q,t\right)Q_{\lambda^{(a)}}\left( \frac{\alpha_a Z^{(a)}}{1-t};q,t\right) \\
    &\sum_{\{\mu^{(a)}\}} P_{\mu^{(a)}}\left( \bar{W}^{(a)};q,t\right)Q_{\mu^{(a)}}\left( \frac{\beta_a \bar{Z}^{(a+1)}}{1-t};q,t\right) \\
    &\sum_{\{\sigma^{(a)}\}} \gamma_a^{|\sigma|}b_{\sigma^{(a)}}(q,t) P_{\sigma^{(a)}}\left(W^{(a)};q,t\right)P_{\sigma^{(a)}}\left(\bar{W}^{(a+1)};q,t\right)\,.
\end{split}    
\end{equation}
Now we can use (\ref{eq:macdonaldfacts}) to absorb the factors of $\left(w^{(a)}\right)^{-\zeta_a}$. Further, using the Macdonald algebra structure constants (\ref{eq:macdonaldstructureconstants}) we can write the integral as:
\begin{equation}
\begin{split}
    &\mathcal{I}^N_{\{N_a\} \{k_a\}}(Z^{(a)};\zeta_a;\alpha_a,\beta_a,\gamma_a;q,t) = \\
     &\prod_{a=1}^N \frac{1}{(t;q)_{\infty}^{k_a}}\oint \prod_{a=1}^N d\mu[W^{(a)};q,t] \\ &\sum_{\substack{\{\lambda^{(a)},\mu^{(a)},\sigma^{(a)}\} \\ \{\nu^{(a)},\rho^{(a)}\}}}\prod_{a=1}^N f^{\nu^{(a)}}_{\lambda^{(a)} \sigma^{(a)}}(q,t) P_{\nu^{(a)}}\left( W^{(a)};q,t\right) f^{\rho^{(a)}}_{\mu^{(a)} \tilde{\sigma}^{(a-1)}} P_{\rho^{(a)}}\left( \bar{W}^{(a)};q,t\right) \\
    &\gamma_a^{|\sigma^{(a)}|} Q_{\lambda^{(a)}}\left( \frac{\alpha_a Z^{(a)}}{1-t};q,t \right)Q_{\mu^{(a)}}\left( \frac{\beta_a\bar{Z}^{(a+1)}}{1-t};q,t \right) b_{\sigma^{(a)}}(q,t)\,.
\end{split}    
\end{equation}
In the above we write $\tilde{\sigma}^{(a)}$ to denote the partition shifted by $(\zeta^k)$, and $\sigma^{(0)}$ is identified with $\sigma^{(N)}$. Precisely:
\begin{equation}
    \tilde{\sigma}^{(a)} = \sigma^{(a)} + (\zeta_{a+1}^{k_{a+1}})
\end{equation}
where again $k_{N+1}$ and $\zeta_{N+1}$ are identified with $k_1$ and $\zeta_1$ respectively. In the next step of the calculation, we use the orthogonality of the Macdonald polynomials with respect to the inner product (\ref{eq:macdonaldinnerproduct}), this introduces a normalisation factor (\ref{eq:macdonorm}). Finally, we use the definition of skew Macdonald polynomials to write the integral as:
\begin{equation}\label{eq:generalchainsaw}
\begin{split}
    &\mathcal{I}^N_{\{N_a\} \{k_a\}}(Z^{(a)};\zeta_a;\alpha_a,\beta_a,\gamma_a;q,t) = \\
    &\sum_{\{\nu^{(a)},\sigma^{(a)}\}}\prod_{a=1}^N \frac{\tilde{c}_{k_a}(\nu^{(a)};q,t)}{(t;q)_{\infty}^{k_a}} \gamma_a^{|\sigma^{(a)}|}P_{\nu^{(a)}/ \sigma^{(a)}}\left( \frac{\alpha_a Z^{(a)}}{1-t} ;q,t\right) Q_{\nu^{(a)}/ \tilde{\sigma}^{(a-1)}}\left( \frac{\beta_a \bar{Z}^{(a+1)}}{1-t} ;q,t\right)\,.
\end{split}    
\end{equation}
In the above, we have also used part of the normalisation of the inner product (the $b_{\nu}$ term) combined with the $b_{\sigma}$ term to re-normalise the first skew Macdonald polynomial.

Now, if we send $q \to 0$ then $\mathcal{I}$ becomes the Molien integral for the Hilbert series of the chainsaw quiver of figure \ref{fig:chainsaw}---if we include line bundle charge $\zeta$ this is the equivariant Euler characteristic of the corresponding tautological line bundle. This limit degenerates the plethystically substituted Macdonald polynomials to Milne polynomials and the normalisation constant simplifies using (\ref{eq:HLlimit}).

In fact the ADHM formula of section \ref{subsec:hilbertseries} is a special case of a chainsaw quiver with one gauge node and $\zeta_1=0$, identifying the parameters as follows recovers the expression (\ref{eq:adhmhs}):
\begin{equation}
    \mathcal{I}_{p, N}\left(\zeta=0,\alpha=\beta=t_{\text{There}}^{\frac{1}{2}}, \gamma=z^{-1}t^{\frac{1}{2}}_{\text{There}};q=0,t_{\text{Here}}=zt^{\frac{1}{2}}_{\text{There}}\right) = \mathcal{Z}_{\text{H.S.}}\left[\mathcal{M}_{N,p}\right]\,.
\end{equation}
Explicitly, for ADHM with $p \ge 1$ flavours we find:
\begin{equation}
    \mathcal{Z}_{\text{H.S.}}\left[\mathcal{M}_{N,p}\right] = \sum_{\lambda, \mu} \frac{1}{(zt^{\frac{1}{2}};zt^{\frac{1}{2}})_{N-l(\mu)}}\left(z^{-1}t^{\frac{1}{2}}\right)^{|\lambda|}P'_{\mu/\lambda}\left(t^{\frac{1}{2}}X;zt^{\frac{1}{2}} \right) Q'_{\mu/\lambda}\left(t^{\frac{1}{2}}\bar{X};zt^{\frac{1}{2}}\right)\,.
\end{equation}
where $X$ is the set of fugacities for the flavour symmetry $X = \{ x_1,\ldots,x_{p}\}$.

\paragraph{Large rank limit}
We now consider the limit $N\to \infty$. In this limit the normalisation constant is independent of $\mu$ and becomes simply:
\begin{equation}
    (zt^{\frac{1}{2}};zt^{\frac{1}{2}})_{\infty} = \prod_{k=0}^{\infty}\frac{1}{1-\left(zt^{\frac{1}{2}}\right)^{k+1}}\,.
\end{equation}
Combining this with the plethystic form of the generalised Cauchy identity (\ref{eq:generalisedcauchyplethy}) in the Hall-Littlewood limit $q \to 0$ we find:
\begin{equation}
    \lim_{N\to\infty}\mathcal{Z}_{\text{H.S.}}\left[\mathcal{M}_{N,p}\right] =  \prod_{k=0}^{\infty} \frac{1}{1-\left(zt^{\frac{1}{2}}\right)^{k+1}}\frac{1}{1-\left(z^{-1}t^{\frac{1}{2}}\right)^{k+1}} \prod_{i,j=1}^N \prod_{l=0}^{\infty}\frac{1}{1-\left(zt^{\frac{1}{2}}\right)^{l}\left(z^{-1}t^{\frac{1}{2}}\right)^{k}x_i/x_j}\,.
\end{equation}
Before concluding this appendix we remark that in fact a large rank limit of the more general expression (\ref{eq:generalchainsaw}) is also possible in the case $\zeta=0$. Iterating the identity (\ref{eq:skewcauchy}) allows us to concatenate the variables before applying (\ref{eq:generalisedcauchyplethy}), although in this work we focus on the ADHM case rather than the more general chainsaw.

\section{$T[SU(N)]$ Examples}\label{appendix:tsunexamples}
In this appendix we consider the $T[SU(N)]$ theory. This theory has a product gauge group $U(1)\times \ldots \times U(N-1)$ with $\mathcal{N}=4$ vector multiplets for each gauge node, bifundamental hypermultiplets transforming in $U(k)\times U(k+1)$ for $k=1,\ldots,N-2$ and $N$ fundamental hypermultiplets of $U(N-1)$. The field content of this theory is summarised by the $\mathcal{N}=4$ quiver diagram \ref{fig:tsunquiver}.

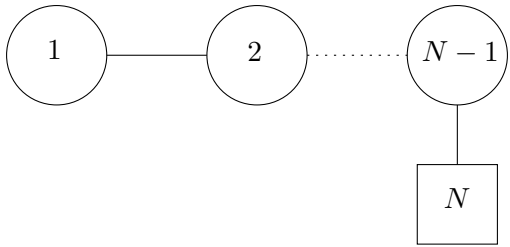
\begin{figure}
\centering

\begin{tikzpicture}[x=0.75pt,y=0.75pt,yscale=-1,xscale=1]

\draw   (60,95) .. controls (60,81.19) and (71.19,70) .. (85,70) .. controls (98.81,70) and (110,81.19) .. (110,95) .. controls (110,108.81) and (98.81,120) .. (85,120) .. controls (71.19,120) and (60,108.81) .. (60,95) -- cycle ;
\draw   (160,95) .. controls (160,81.19) and (171.19,70) .. (185,70) .. controls (198.81,70) and (210,81.19) .. (210,95) .. controls (210,108.81) and (198.81,120) .. (185,120) .. controls (171.19,120) and (160,108.81) .. (160,95) -- cycle ;
\draw    (110,95) -- (160,95) ;
\draw  [dash pattern={on 0.84pt off 2.51pt}]  (210,95) -- (260,95) ;
\draw   (260,95) .. controls (260,81.19) and (271.19,70) .. (285,70) .. controls (298.81,70) and (310,81.19) .. (310,95) .. controls (310,108.81) and (298.81,120) .. (285,120) .. controls (271.19,120) and (260,108.81) .. (260,95) -- cycle ;
\draw    (285,120) -- (285,150) ;
\draw   (265,150) -- (305,150) -- (305,190) -- (265,190) -- cycle ;

\draw (79,86.4) node [anchor=north west][inner sep=0.75pt]    {$1$};
\draw (179,87.4) node [anchor=north west][inner sep=0.75pt]    {$2$};
\draw (266,86.4) node [anchor=north west][inner sep=0.75pt]    {$N-1$};
\draw (277,160.4) node [anchor=north west][inner sep=0.75pt]    {$N$};

\end{tikzpicture}

		\caption{$\mathcal{N}=4$ Quiver diagram for the  $T[SU(N)]$ theory.}\label{fig:tsunquiver}
\end{figure}

The theory is self mirror dual and has resolved Higgs/Coulomb branches given by the cotangent bundle to the complete flag variety in $\mathbb{C}^N$:
\begin{equation}
    \mathcal{M}_H = \mathcal{M}_C = T^*F_N.
\end{equation}
The Higgs/Coulomb branch admits an action of the $SU(N)$ flavour symmetry and an anti-diagonal combination of the $\mathcal{N}=4$ R-symmetry that acts by contracting the cotangent directions. In the presence of generic real mass and FI parameters, the fixed points under this group action on the Higgs/Coulomb branches are labelled by permutations of $N$ which we denote by $\sigma \in S_N$. 

We denote the mass parameters/flavour fugacities as $z_i=e^{-m_i}$ for $i=1,\ldots,N$, the FI parameters/topological fugacities as $\zeta_i = e^{-\xi_i}$ and $t$ denotes the usual diagonal combination of R-symmetries.

We now turn to the Neumann half index of $T[SU(N)]$, we follow the recipe of \cite{Dimofte:2017tpi} to compute this object. In the $B$-shifted R-symmetry convention and with a Wilson line insertion $\mathcal{W}_{\mathfrak{n}}$ analogous to the setup in section \ref{sec:qm} we have
\begin{equation}
\begin{split}
    &\mathcal{I}_{\mathfrak{n}}[z_i;q,t] =\\ &\oint_{\Gamma} \prod_{a=1}^{N-1}\prod_{i=1}^{a} \frac{dx_a^{(i)}}{x_i^{(a)}} e^{\log(x_{i}^{(a)}) \log(\mathfrak{n}_{i}^{(a)})} \prod_{a=1}^{N-1} \frac{\prod_{i \neq j}^a \left(x_j^{(a)}/x_i^{(a)};q\right)_{\infty}}{\prod_{i,j}^N\left(tqx_j^{(a)}/x_i^{(a)};q\right)_{\infty}} \prod_{a=1}^{N-1}\prod_{i=1}^a \prod_{j=1}^{a+1} \frac{\left(tqx_j^{(a+1)}/x_i^{(a)};q\right)_{\infty}}{\left(x_j^{(a+1)}/x_i^{(a)};q\right)_{\infty}} \,,
\end{split}    
\end{equation}
where in the above we identify $x_j^{(N)}=z_j$ for $j=1,\ldots, N$.

In the work of \cite{Zenkevich:2017ylb}, these integrals realise holomorphic blocks of the $T[SU(N)]$ theory and in that context there is a basis of contours $\Gamma_\sigma$ in 1-1 correspondence with vacua $\sigma \in S_N$. In this appendix we replicate the half index setup of section \ref{sec:qm} and instead take all poles in the unit circle. This corresponds to a linear combination of holomorphic blocks $\Gamma = \sum_{\sigma \in S_N} \Gamma_\sigma$. 

The index coincides, up to a $(q,t)$ constant, with the integral representation of a Macdonald polynomial (\ref{eq:macdonaldintegralrepresentation}) with the line operator charge setting the highest weight $\lambda$:
\begin{equation}
    \mathcal{I}_{\mathfrak{n}_i^{(a)}=\lambda_a}[z_i;q,t] = P_{\lambda}\left(z;q,tq\right)\,.
\end{equation}
Now we consider the $\mathcal{N}=4$ limit. Sending $t\to1$ we have:
\begin{equation}
    \mathcal{I}_{\mathfrak{n}_i^{(a)}=\lambda_a}[z_i;q,t\to1] = P_{\lambda}\left(z;q,q\right) = s_{\lambda}\left(z\right)\,.
\end{equation}
The Macdonald polynomial degenerates to a Schur polynomial and we recover the finite dimensional simple module of the Coulomb branch algebra. Physically, the $t\to1$ limit suppresses the vortex contributions and the holomorphic block integral receives contributions only from fixed points on the Higgs branch.

The geometry of this example is simply the Borel-Weil-Bott theorem. Given a weight $\lambda$ for $\mathfrak{sl}_N$ one can define a line bundle $L_{\lambda}$ over $F_N$, it is well-known that the higher cohomology groups vanish and $H^0(F_N,L_{\lambda})$ forms the highest weight irreducible $\mathfrak{sl}_N$ module corresponding to $\lambda$. The equivariant Euler characteristic with respect to the maximal torus $T\subset \mathfrak{sl}_N$ then yields the character:
\begin{equation}
    \chi_T(F_N,L_{\lambda}) = s_{\lambda}.
\end{equation}
That is, the half index in the presence of a Wilson line insertion counts sections of appropriate tensor products of tautological line bundles over $F_N$.

\subsection{$T[SU(N)]$ Poincar\'e polynomial limit}\label{appendix:pplimit}
In this appendix we give an explicit example where the $q \to 0$ limit of the vortex partition function is identified with the Poincar\'e polynomial of vortex moduli space. This section is essentially a review of results in section 6 of the work \cite{Crew:2020jyf}.

The vortex contributions to the holomorphic block in the $A$-twist can be identified with the $\chi_t$ genus of local Laumon space $\mathfrak{Q}_{\boldsymbol{d}}$. The vortex number is identified with the degree of the Laumon space and we have: 
\begin{equation}
    \mathcal{Z}_{S^1 \times D}^{\text{Vortex}} = \sum_{\boldsymbol{d}} \prod_{s=1}^{N-1} \left(t^{-1/2} \zeta_{s}/ \zeta_{s+1}\right)^{d_s} \chi_t\left( \mathfrak{Q}_{\boldsymbol{d}}\right)\,.
\end{equation}
The action corresponding to the $q$ fugacity is a Reeb vector on $\mathfrak{Q}_{\boldsymbol{d}}$. Sending $q\to 0$ computes the $\chi_t$ genus of the compact fixed point submanifold, in this case the compact core of the Laumon space denoted $\pi^{-1}(o) \subset \mathfrak{Q}_{\boldsymbol{d}}$. Since the core is compact the $\chi_t$ genus coincides with the Poincar\'e polynomial and we have
\begin{equation}
    \lim_{q\to0}\chi_t\left({\mathfrak{Q}_{\boldsymbol{d}}}\right) = P_t\left(\pi^{-1}(o)\right)\,.
\end{equation}
A more careful discussion of this argument can be found in \cite{Dorey:2019kaf}. Nakajima \cite{nakajima2012handsaw} derives a generating function for these Poincar\'e polynomials and thus the vortex sum can be computed explicitly in this limit:
\begin{equation}
    \mathcal{Z}_{S^1\times D}^{\text{Vortex}} = \sum_{\boldsymbol{d}} \prod_{s=1}^{N-1} \left(t^{-1/2} \zeta_{s}/ \zeta_{s+1}\right)^{d_s} P_t\left(\pi^{-1}(o)\right) = \prod_{i<j}^N \frac{1}{1-t^{\frac{1}{2}} \zeta_i/\zeta_j}\,.
\end{equation}
The right hand side is a $t$-graded Verma denominator for $\mathfrak{sl}_N$. This is the $T[SU(N)]$ analogue of the result discussed in section \ref{sec:vpfn}.

\section{Localisation and Boundary Condition}\label{appendix:localisation}

We give the detailed computation of the fundamental blocks for the ADHM theory, which fuse \textit{exactly} to the twisted and superconformal indices, and the squashed ellipsoid partition function.\footnote{The latter after introducing boundary 2d matter to cancel the boundary 't Hooft anomaly.} To our knowledge this is the first example of a first principles derivation of a ``block'' for a theory with adjoint matter without appealing to holomorphic factorisation. In \cite{Bullimore:2020jdq} the authors propose the blocks associated to vacua $\{\alpha\}$ for a 3d $\mathcal{N}=4$ theory should be given by a hemisphere partition function on $S^1\times D$ with $\mathcal{N}=(2,2)$ exceptional Dirichlet boundary conditions on the boundary $\partial(S^1 \times D) = T^2$. Two limits are shown to correspond to characters of Verma modules of the quantised Coulomb and Higgs branch algebras. We refer to \cite{Bullimore:2020jdq, Bullimore:2016nji} for details, including the precise form of the boundary conditions, but give a brief introduction here. 

To associate a boundary condition to each vacuum we work in the half-space operator picture -- the count of BPS states on $S^1\times D$ is the same as counting BPS operators on $\mathbb{R}_{\geq 0} \times \mathbb{R}^2$ inserted at the origin. On the boundary, the $\mathcal{N}=(2,2)$ BPS equations on the Higgs (Coulomb) branch become gradient flow with respect to the Morse function given by contracting the real mass (FI) parameter generating the flavour symmetry with the real moment map of the flavour group action, i.e. $m_{\mathbb{R}} \cdot \mu_{H,\mathbb{R}}$ and  $\xi_{\mathbb{R}} \cdot \mu_{C,\mathbb{R}}$. To avoid repetition we focus on the Higgs branch picture -- analogous statements can be made on the Coulomb branch with the obvious replacements. The boundary condition is specified by:
\begin{itemize}
    \item Dirichlet and Neumann boundary conditions for the $\mathcal{N}=2$ vector and adjoint chiral multiplets comprising the $\mathcal{N}=4$ vector multiplet respectively.
    \item A holomorphic Lagrangian splitting $\mathcal{R}=L\oplus L^*$ for the linear quaternionic representation $\mathcal{R} \simeq \mathbb{H}^N$ of the gauge group $G$ specifying the gauge representation of $N$ hypermultiplets. This specifies a splitting $(X_L,Y_L)$ of the hypermultiplet scalars, such that the scalars in $L^*$ are set to some (matrix) of constant values $Y_L|_{\partial} = c$. The remaining boundary conditions for the components of the hypermultiplets are fixed by supersymmetry. $X_L$ are allowed to fluctuate at the boundary and the image of this under the hyperk\"ahler quotient by $G$ automatically defines a holomorphic Lagrangian submanifold of the Higgs branch.
\end{itemize}
The matrix of constants $c$ is chosen such that the gauge group $G$ is completely broken at the boundary, but a maximal torus of the Higgs and Coulomb branch flavour symmetries, $T_H\times T_C \subset G_H\times G_C$, is preserved. The Lagrangian splitting $L$ is chosen such that its image on $\mathcal{M}_H$ under the quotient gives the holomorphic attracting submanifold $\mathcal{L}_{\alpha} \subset \mathcal{M}_H$ (under Morse flow) associated to the vacuum $\alpha$.\footnote{There is a subtlety for non-abelian theories in that $L$ must also be chosen such that there are no non-trivial orbits of the complexified gauge group $G_{\mathbb{C}}$, else there would be noncompact 2d degrees of freedom on the boundary. This is discussed in section 4.4 of \cite{Bullimore:2016nji} and also dealt with in section \ref{sec:vpfn} and appendix \ref{appendix:geometricboundarycondition} in this work for the theory of interest. } In this way the boundary condition at $x^1=0$ on $\mathbb{R}_{\geq 0 }$ mimics a vacuum at infinity on the full $\mathbb{R}$.  We work in the convention where the Morse function increases along the flow. 

The Morse function evaluated at each vacuum (fixed point) provides a partial ordering on the vacua:
\begin{equation}
    \alpha \in \overline{\mathcal{L}_{\beta}} \Rightarrow \alpha \leq \beta.
\end{equation}

The values of both functions at each vacua (which are the critical values of the Morse functions) coincide with a single central charge, labelled by the vacuum. 
\begin{equation}
\begin{split}
    &\kappa_{\alpha}\,:\quad  \mathfrak{t}_H \times \mathfrak{t}_C \rightarrow \mathbb{R}, \\
    \kappa_{\alpha}(m_{\mathbb{R}},\xi_{\mathbb{R}})& = m_{\mathbb{R}} \cdot \mu_{H,\mathbb{R}}(\alpha) = \xi_{\mathbb{R}} \cdot \mu_{C,\mathbb{R}}(\alpha)\,.
\end{split}
\end{equation} 
The central charge coincides with the value of the effective $T_H\times T_C$  mixed Chern-Simons coupling in the vacuum.

One can compute the half-index $\mathcal{I}^{\alpha}$ for such boundary conditions \cite{Dimofte:2017tpi}. As shown in \cite{Bullimore:2020jdq}, the half index on $\mathbb{R}_{\geq 0} \times \mathbb{R}^2$ and the hemisphere partition function $\mathcal{Z}^{\alpha}_{S^1\times D}$ differ precisely by a factor $e^{\phi_{\alpha}}$ corresponding to the Casimir energy:
\begin{equation}
    \mathcal{Z}_{S^1\times D}^{\alpha} = e^{\phi_{\alpha}}\mathcal{I}^{\alpha}
\end{equation}
where $\phi_{\alpha}$ is determined by boundary 't Hooft anomalies and correspond to effective (mixed) Chern-Simons couplings. For an $\mathcal{N}=4$ theory with $\mathcal{N}=(2,2)$ boundary conditions, the possible boundary 't Hooft anomalies are mixed $U(1)_V - U(1)_A$ , $T_H  - U(1)_A$, $T_C - U(1)_V$ and  $T_H - T_C$ anomalies. The latter corresponds exactly to the central charge. In the limits $t\rightarrow q^{\pm \frac{1}{2}}$, the hemisphere partition function specialises to the characters of Verma modules of quantised Higgs and Coulomb branch algebras with $\log q$ playing the role of the $\Omega$-deformation parameter, and where the module is that of boundary Higgs and Coulomb branch operators. The Casimir energy specialises to the lowest weights of these modules.
\begin{equation}
\begin{split}
    \lim_{t\rightarrow q^{-\frac{1}{2}}} \mathcal{Z}_{S^1\times D}^{\alpha} = \chi^{H,\alpha}\,,\\
    \lim_{t\rightarrow q^{\frac{1}{2}}} \mathcal{Z}_{S^1\times D}^{\alpha} = \chi^{C,\alpha}.\\
\end{split}
\end{equation}

\subsection{Detailed computation of hemisphere partition functions}\label{appendix:holoblockdetailedcomputation}

We give the details of the computation of the hemisphere partition function from section \ref{sec:vpfn}. Before deformation, for Dirichlet boundary conditions with $c=0$ we have:
\begin{equation}
\begin{split}
    \tilde{\mathcal{Z}}^{\lambda}_{S^1\times D} &= \sum_{k\in \mathbb{Z}^N}  I(k,q,t,\{u^{-1}v_a\},z,\zeta)\\
    &\equiv 
    \sum_{k\in \mathbb{Z}^N} e^{\frac{\log\zeta}{\log q}\left(\sum\limits_{a\in\lambda} \log\left(s_aq^{k_a}\right)\right)} \prod_{a,b\,\in \lambda}  \frac{\left(u^2\frac{s_a}{s_b}q^{k_a-k_b};q\right)'_{\infty}}{\left(q\frac{s_a}{s_b}q^{k_a-k_b};q\right)'_{\infty}}
    \prod_{a\in \lambda }  \frac{\left(q u^{-1} s_a q^{k_a};q\right)'_{\infty}}{\left(us_a q^{k_a};q\right)'_{\infty}}\\
    &\prod_{\substack{a,b\in \lambda \text{ s.t. } \\  (b \in \lambda^{B} )\cap( i_{a}> i_b) \\ \text{ or } (b \notin \lambda^{B})}}  \frac{\left(qz^{-1}u^{-1}\frac{s_a}{s_b}q^{k_a-k_b};q\right)'_{\infty}}{\left(z^{-1}u\frac{s_a}{s_b}q^{k_a-k_b};q\right)'_{\infty}}
    \prod_{\substack{a,b\in \lambda \text{ s.t. } \\ (a \in \lambda^{B} )\cap( i_{b} \leq i_a)}}  \frac{\left(qzu^{-1}\frac{s_a}{s_b}q^{k_a-k_b};q\right)'_{\infty}}{\left(zu\frac{s_a}{s_b}q^{k_a-k_b};q\right)'_{\infty}}, \\
\end{split}
\end{equation}
where recall
\begin{equation}\label{eq:zetafunctionreg}
    (a;q)' = e^{-\mathcal{E}[-\log(a)]} (a;q)\,,\qquad \mathcal{E}[x]=\frac{\beta}{12}-\frac{x}{4}+\frac{1}{8\beta}x^2
\end{equation}
and $q=e^{-2\beta}$. We deform to the partition function $\mathcal{Z}^{\lambda}_{S^1\times D}$ for the exceptional Dirichlet boundary condition $\mathcal{B}_{\lambda}$ by setting to 1 the product of fugacities dual to the charges of the chirals whose scalars are fixed to non-zero values at the boundary. This corresponds to the non-zero value of scalars at the boundary breaking the combination of gauge, flavour and R-symmetry under which they are charged. These are the components of $(A,B,I)$ given by the particular tree $T_{\lambda}$. Thus:
\begin{equation}
\begin{split}
    s_{(1,1)} = u\,,\quad s_a = u (z u)^{i_a-1} (z^{-1} u )^{j_a-1} \equiv u^{-1} v_{a}\,.
\end{split}
\end{equation}

\paragraph{The vortex partition function.} We can isolate the dependence on monopole charge as the vortex partition function:
\begin{equation}
    \mathcal{Z}^{\lambda}_{S^1\times D} = I(0,q,t,\{u^{-1}v_a\},z,\zeta)\, \mathcal{Z}^{\lambda}_{\text{Vortex}}
\end{equation}
by using the identity: $(aq^n;q)_{\infty} = \frac{(a;q)_{\infty}}{(a;q)_n}$, and the form of the $\mathcal{E}(x)$ function. We have:
\begin{equation}
    \mathcal{Z}^{\lambda}_{\text{Vortex}} = \sum_{k_a} \left( \zeta t^{\frac{1}{2}}q^{-\frac{1}{4}}\right)^{\sum k_a} \prod_{a\in Y} \frac{\left(u^2 v_a^{-1};q\right)_{-k_a}}{\left(q v_a^{-1};q\right)_{-k_a}} \prod_{a\neq b} \frac{\left(qu^{-2}\frac{v_b}{v_a};q\right)_{k_b-k_a}}{\left(\frac{v_b}{v_a};q\right)_{k_b-k_a}}\frac{\left(zu\frac{v_b}{v_a};q\right)_{k_b-k_a}}{\left(qzu^{-1}\frac{v_b}{v_a};q\right)_{k_b-k_a}}\,.
\end{equation}
We now prove that the only non-zero contributions to $\mathcal{Z}^{\lambda}_{\text{Vortex}}$ are when $\{k\}$ form a reverse plane partition. To do this, first note when various q-Pochhamers could develop poles or zeros:
\begin{itemize}
    \item For $a=(1,1)$: $\left(u^2 v_a^{-1};q\right)_{-k_a}$ is 0 if $k_a <0 $ and non-zero if $k_a \geq 0$.
    \item For $a$ directly below $b$: 
        \begin{tikzpicture}[x=0.2pt,y=0.2pt,yscale=-1,xscale=1]
        \draw  [fill={rgb, 255:red, 74; green, 144; blue, 226 }  ,fill opacity=0.5 ] (196,114) -- (246,114) -- (246,164) -- (196,164) -- cycle ;
        \draw  [fill={rgb, 255:red, 255; green, 77; blue, 77 }  ,fill opacity=0.5 ] (196,164) -- (246,164) -- (246,214) -- (196,214) -- cycle ;
        \draw (221,189) node    {$a$};
        \draw (221,139) node    {$b$};
        \end{tikzpicture} 
        , $\left(zu\frac{v_b}{v_a};q\right)_{k_b-k_a}$ is a zero if $k_a<k_b$, and non-zero if $k_a \geq k_b$.
    \item If $b$ directly right of $a$: 
            \begin{tikzpicture}[x=0.2pt,y=0.2pt,yscale=-1,xscale=1]
            \draw  [fill={rgb, 255:red, 74; green, 144; blue, 226 }  ,fill opacity=0.5 ] (196,114) -- (246,114) -- (246,164) -- (196,164) -- cycle ;
            \draw  [fill={rgb, 255:red, 255; green, 77; blue, 77 }  ,fill opacity=0.5 ] (146,114) -- (196,114) -- (196,164) -- (146,164) -- cycle ;
            \draw (171,139) node    {$a$};
            \draw (221,139) node    {$b$};
            \end{tikzpicture}
            , $\left(qzu^{-1}\frac{v_b}{v_a};q\right)_{k_b-k_a}^{-1}$ is zero if $k_b<k_a$ and non-zero if $k_b\geq k_a$.
    \item For $b$ diagonally to the right of $a$: 
    \begin{tikzpicture}[x=0.2pt,y=0.2pt,yscale=-1,xscale=1]
        \draw  [fill={rgb, 255:red, 74; green, 144; blue, 226 }  ,fill opacity=0.5 ] (196,164) -- (246,164) -- (246,214) -- (196,214) -- cycle ;
        \draw  [fill={rgb, 255:red, 255; green, 77; blue, 77 }  ,fill opacity=0.5 ] (146,114) -- (196,114) -- (196,164) -- (146,164) -- cycle ;
        \draw  [fill={rgb, 255:red, 184; green, 233; blue, 134 }  ,fill opacity=1 ] (196,114) -- (246,114) -- (246,164) -- (196,164) -- cycle ;
        \draw  [fill={rgb, 255:red, 184; green, 233; blue, 134 }  ,fill opacity=1 ] (146,164) -- (196,164) -- (196,214) -- (146,214) -- cycle ;
        \draw (171,139) node    {$a$};
        \draw (221,189) node    {$b$};
    \end{tikzpicture}
    , $\left(qu^{-2}\frac{v_b}{v_a};q\right)_{k_b-k_a}$ can develop a pole if $k_b<k_a$. But if this is the case, there are at least two zeros coming from considering $a,b$ in relation to the two green boxes corresponding to the cases above. Such a configuration of $\{k\}$ is always zero. We must have $k_b \geq k_a$.
\end{itemize}
None of the other q-Pochhammers can develop poles or zeros. Therefore $k_a$ must increase along the rows (going to the right) and columns (going down) of the Young diagram, and forms a reverse plane partition.

\paragraph{Perturbative contribution} We now describe the perturbative piece of the hemisphere partition function $I(0,q,t,\{u^{-1}v_a\},z,\zeta)$. This undergoes significant cancellations. Naively we see that in the contribution of the vector multiplet, $\left(u^2\frac{s_a}{s_b};q\right)_{\infty}$ gives a zero whenever there is a pair of boxes $a,b$ in the configuration \tikzset{every picture/.style={line width=0.75pt}}: \begin{tikzpicture}[x=0.2pt,y=0.2pt,yscale=-1,xscale=1]

\draw  [fill={rgb, 255:red, 74; green, 144; blue, 226 }  ,fill opacity=0.5 ] (246,214) -- (296,214) -- (296,264) -- (246,264) -- cycle ;
\draw  [fill={rgb, 255:red, 255; green, 77; blue, 77 }  ,fill opacity=0.5 ] (196,164) -- (246,164) -- (246,214) -- (196,214) -- cycle ;

\draw (221,189) node    {$a$};
\draw (271,239) node    {$b$};

\end{tikzpicture} . We will see that these are cancelled. First consider the terms:
\begin{equation}
\begin{split}
     &\prod_{a,b\,\in \lambda} \frac{1}{\left(q\frac{v_a}{v_b};q\right)'_{\infty}}
    \prod_{a\in \lambda }  \left(q u^{-2} v_a ;q\right)'_{\infty} \\
    &\prod_{\substack{ a,b\in \lambda \text{ s.t. } \\  (b \in \lambda^{B} )\cap( i_{a}> i_b) \\ \text{ or } (b \notin \lambda^{B}) }} \left(qz^{-1}u^{-1}\frac{v_a}{v_b};q\right)'_{\infty}
    \prod_{\substack{a, b \in \lambda \text{ s.t. } \\(a \in \lambda^{B} )\cap( i_{b} \leq i_a)}}  \left(qzu^{-1}\frac{v_a}{v_b};q\right)'_{\infty}\\
    \equiv&   
    \prod\limits_{\substack{a, b \in \lambda \text{ s.t. }\\ b \neq (1,1)}} \mathbf{k}(a,b)^{-1}
    \prod\limits_{\substack{a, b \in \lambda \text{ s.t. } \\(a \in \lambda^{B} )\cap( i_{b} \leq i_a)}} \mathbf{f}(a,b)   
    \prod\limits_{\substack{ a, b \in \lambda \text{ s.t. } \\  (b \in \lambda^{B} )\cap( i_{a}> i_b) \\ \text{ or } (b \notin \lambda^{B}) }} \mathbf{h}(a,b).
\end{split}
\end{equation}
where hopefully the definition of the functions $\mathbf{h},\mathbf{k},\mathbf{f}$ are obvious. We note the following identity:
\begin{equation}
    \mathbf{h}(a,b) = \mathbf{k}( Ta, T\downarrow b) = \mathbf{k}( T\uparrow a, T  b),
\end{equation}
where T is any translation, and $\uparrow,\downarrow$ are shifts up and down by one box in the Young diagram. Using this, we can see upon making successive cancellations:
\begin{equation}
    \frac{\prod\limits_{a, b \in \lambda} \mathbf{h}(a,b)}{\prod\limits_{\substack{a, b \in \lambda\text{ s.t. }\\b\neq (1,1) }} \mathbf{k}(a,b)} = \frac{\prod\limits_{\substack{a, b \in \lambda \text{ s.t. }\\ b \in \lambda^B}} \mathbf{h}(a,b)}{\prod\limits_{\substack{a, b \in \lambda\text{ s.t. }\\b\ \in (1,\cdot)'}} \mathbf{k}(a,b)} = \frac{\prod\limits_{\substack{a, b \in \lambda \text{ s.t. }\\ (b \in \lambda^B)\cap(i_a \leq i_b) }} \mathbf{h}(a,b)}{\prod\limits_{\substack{a, b \in \lambda\text{ s.t. }\\ (b \in (1,\cdot)')\cap (a\notin \tilde{\lambda}_b)}} \mathbf{k}(a,b)}\,.
\end{equation}
Here $(1,\cdot)'$ is the set of boxes in the first row of $\lambda$ except for $(1,1)$. $\lambda^B$ the bottom-most boxes in the diagram.  $\tilde{\lambda}_b$ for some $b 
\in (1,\cdot)'$ is defined as the set of boxes obtained by shifting all boxes $c$ such that $i_c \geq \lambda_{j_b}^{\vee}$ all the way to the top. See figure \ref{fig:youngdiagramtildeYb} for an example.
\begin{figure}
    \centering
        \begin{tikzpicture}[x=0.3pt,y=0.3pt,yscale=-1,xscale=1]
        
        \draw  [fill={rgb, 255:red, 184; green, 233; blue, 134 }  ,fill opacity=1 ] (171,62) -- (221,62) -- (221,112) -- (171,112) -- cycle ;
        \draw  [fill={rgb, 255:red, 184; green, 233; blue, 134 }  ,fill opacity=1 ] (221,62) -- (271,62) -- (271,112) -- (221,112) -- cycle ;
        \draw   (271,62) -- (321,62) -- (321,112) -- (271,112) -- cycle ;
        \draw   (321,62) -- (371,62) -- (371,112) -- (321,112) -- cycle ;
        \draw   (371,62) -- (421,62) -- (421,112) -- (371,112) -- cycle ;
        \draw  [fill={rgb, 255:red, 184; green, 233; blue, 134 }  ,fill opacity=1 ] (171,112) -- (221,112) -- (221,162) -- (171,162) -- cycle ;
        \draw   (171,162) -- (221,162) -- (221,212) -- (171,212) -- cycle ;
        \draw   (171,212) -- (221,212) -- (221,262) -- (171,262) -- cycle ;
        \draw   (171,262) -- (221,262) -- (221,312) -- (171,312) -- cycle ;
        \draw   (221,212) -- (271,212) -- (271,262) -- (221,262) -- cycle ;
        \draw   (221,162) -- (271,162) -- (271,212) -- (221,212) -- cycle ;
        \draw   (271,162) -- (321,162) -- (321,212) -- (271,212) -- cycle ;
        \draw   (321,162) -- (371,162) -- (371,212) -- (321,212) -- cycle ;
        \draw   (221,112) -- (271,112) -- (271,162) -- (221,162) -- cycle ;
        \draw   (271,112) -- (321,112) -- (321,162) -- (271,162) -- cycle ;
        \draw   (321,112) -- (371,112) -- (371,162) -- (321,162) -- cycle ;
        \draw   (421,62) -- (471,62) -- (471,112) -- (421,112) -- cycle ;
        
        \draw (346,87) node    {$b$};
        \draw (296,87) node    {$b$};

        \end{tikzpicture}
    \caption{The green highlighted boxes forming $\tilde{\lambda}_b$ for either of the boxes marked $b$.}
    \label{fig:youngdiagramtildeYb}
\end{figure}
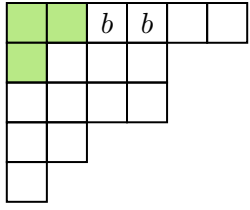
In particular this implies that:
\begin{equation}
    \prod\limits_{\substack{ a, b \in \lambda \text{ s.t. } \\  (b \in \lambda^{B} )\cap( i_{a}> i_b) \\ \text{ or } (b \notin \lambda^{B}) }} \mathbf{h}(a,b) \prod\limits_{\substack{a, b \in \lambda \text{ s.t. }\\ b \neq (1,1)}} \mathbf{k}(a,b)^{-1} = \prod\limits_{\substack{a, b \in \lambda\text{ s.t. }\\ (b \in (\cdot,1)')\cap (a\notin \tilde{\lambda}_b)}} \mathbf{k}(a,b)^{-1}.
\end{equation}
Now we claim that:
\begin{equation}\label{eq:ADHMholoblockcancellation}
    \prod\limits_{\substack{a, b \in \lambda \text{ s.t. } \\(a \in \lambda^{B} )\cap( i_{b} \leq i_a)}} \mathbf{f}(a,b)
    \prod\limits_{\substack{a, b \in \lambda\text{ s.t. }\\ (b \in (\cdot,1)')\cap (a\notin \tilde{\lambda}_b)}} \mathbf{k}(a,b)^{-1}
    = \prod_{\substack{a, b \in \lambda \text{ s.t. } \\ (a\in \lambda^B)\cap(b\in\lambda^R) \\ \cap(i_b \leq i_a) }} \mathbf{f}(a,b),
\end{equation}
where $\lambda^{R}$ are the rightmost boxes in $\lambda$. To justify this, note we have:
\begin{equation}
    \mathbf{f}(a,b)=\mathbf{k}(T\leftarrow a, Tb)= \mathbf{k}( T a, T\rightarrow b).
\end{equation}
Define the subset of pairs of boxes in $\lambda$:
\begin{equation}
\begin{split}
    S&=\left\{ (a,b)\,\rvert\, (a\in \lambda^B)\cap(b\notin \lambda^{R})\cap(i_b\leq i_a)\right\}\\
    S'&=\left\{ (a,b)\,\rvert\, (b \in (1,\cdot)')\cap (a\notin \tilde{\lambda}_b) \right\}
\end{split}
\end{equation}
where
\begin{equation}
    |S| = |S'| = \left(\sum_{j=1}^{\lambda_1} \sum_{i=1}^{\lambda_j^{\vee}} \lambda_i \right) - N,
\end{equation}
and the map:
\begin{equation}
    M: S \rightarrow S',\quad M: (a,b) \mapsto (a',b')\equiv ( \uparrow^{i_b-1} a, \uparrow^{i_b-1} \rightarrow b).
\end{equation}
We can show that $M$ truly maps $S$ into $S'$. Take $(a,b)\in S$. For $j_b<j_a$ it is clear. For $j_b\geq j_a$ note that $i_a - i_{\rightarrow b } \geq \lambda_{j_a}^{\vee} - \lambda_{j_{\rightarrow b}}^{\vee}$, which implies $i_{a'} - i_{b'} \geq \lambda_{j_{a'}}^{\vee} - \lambda_{j_{b'}}^{\vee}$. The latter implies that $a' \notin \tilde{\lambda}_{b'}$. Noting that $\mathbf{f}(a,b) = \mathbf{k}(a',b')$, and that $M$ clearly injects, we arrive at (\ref{eq:ADHMholoblockcancellation}). So all together:
\begin{equation}\label{eq:1-looppiececancellatinos}
     \prod\limits_{\substack{ a, b \in \lambda \text{ s.t. } \\  (b \in \lambda^{B} )\cap( i_{a}> i_b) \\ \text{ or } (b \notin \lambda^{B}) }} \mathbf{h}(a,b)
    \prod\limits_{\substack{a, b \in \lambda \text{ s.t. } \\(a \in \lambda^{B} )\cap( i_{b} \leq i_a)}} \mathbf{f}(a,b) \prod\limits_{\substack{a, b \in \lambda \text{ s.t. }\\ b \neq (1,1)}} \mathbf{k}(a,b)^{-1} =\prod_{\substack{a, b \in \lambda \text{ s.t. } \\ (a\in \lambda^{B})\cap(b\in\lambda^{R}) \\ \cap(i_b \leq i_a) }} \mathbf{f}(a,b) \,.
\end{equation}
From the arguments of the other q-Pochhamers appearing in $I(0,q,t,\{u^{-1}v_a\},z,\zeta)$, we see there are identical cancellations from the remaining terms (cancelling any zeros from the denominator of the 1-loop determinant for the $\mathcal{N}=4$ vector multiplet), and at the end of the day we have:
\begin{equation}
    I(0,q,t,\{u^{-1}v_a\},z,\zeta)
    = e^{\frac{\log\zeta}{\log q} \left(\sum_{a\in\lambda }\log u^{-1}v_a\right)} 
    \prod_{\substack{a, b \in \lambda \text{ s.t. } \\ (a\in \lambda^{B})\cap(b\in\lambda^{R}) \\ \cap(i_b \leq i_a) }} 
    \frac{e^{-\mathcal{E}\left[-\log\left(qzu^{-1}\frac{v_a}{v_b}\right)\right]} \left(qzu^{-1}\frac{v_a}{v_b};q\right)_{\infty} }{e^{-\mathcal{E}\left[-\log\left( zu \frac{v_a}{v_b}\right)\right]} \left( zu \frac{v_a}{v_b};q\right)_{\infty}}.
\end{equation}
Note $a,b \in \lambda$ such that  $(a\in \lambda^{B})\cap(b\in\lambda^{R})  \cap(i_b \leq i_a)$ uniquely defines a box $s$ in the same column as $a$ and row as $b$, and this identification is 1-1. In particular $i_a-i_b = l_{\lambda}(s)$ and $j_a-j_b = -a_{\lambda}(s)$ in terms of arm and leg lengths. In total then, we can write:
\begin{equation}
    \mathcal{Z}^{\lambda}_{S^1\times D} =  \mathcal{Z}^{\lambda}_{\text{Classical}}  \mathcal{Z}^{\lambda}_{\text{1-loop}} \mathcal{Z}^{\lambda}_{\text{Vortex}}
\end{equation}
where:
\begin{equation}
    \mathcal{Z}^{\lambda}_{\text{1-loop}} = \prod_{s\in \lambda} \frac{\left(q z^{a_{\lambda}(s)+l_{\lambda}(s)+1}u^{-a_{\lambda}(s)+l_{\lambda}(s)-1};q\right)_{\infty}}{\left(z^{a_{\lambda}(s)+l_{\lambda}(s)+1}u^{-a_{\lambda}(s)+l_{\lambda}(s)+1};q\right)_{\infty}}
\end{equation}
and:
\begin{equation}
    \mathcal{Z}^{\lambda}_{\text{Classical}} = e^{-\left[\sum_{s\in\lambda}c(s) \right]\frac{\log\zeta\log z}{\log q}}
    e^{\left[\sum_{s\in\lambda}h(s) \right]\frac{\log v \log z}{\log q}}
    e^{\left[\sum_{s\in\lambda}h(s) \right]\frac{\log u \log \zeta}{\log q}}
    e^{-\left[\sum_{s\in\lambda}c(s) \right]\frac{\log u \log v}{\log q}}
\end{equation}
where we defined $v=q^{\frac{1}{4}}t^{-\frac{1}{2}}$. In writing the classical piece, the explicit form of $\mathcal{E}$ (\ref{eq:zetafunctionreg}) and identities (\ref{eq:sumofcontentshooks}) have been used.

\paragraph{Superconformal index}

We demonstrate that the above hemisphere partition function fuses \textit{exactly} to the superconformal index computed in \cite{Choi:2019zpz}, which we rewrite in our notation as:
\begin{equation}
    \mathcal{Z}^{\text{S.C.}}_{S^1\times S^2}(q,t,x,\zeta) = \sum_{\lambda} \mathcal{Z}^{\text{S.C.},\lambda}_{\text{Pert}} \mathcal{Z}^{\lambda}_{\text{Vortex}}(q,t,z,\zeta) \mathcal{Z}^{\lambda}_{\text{Vortex}}(q^{-1},t^{-1},z^{-1},\zeta^{-1}) 
\end{equation}
where $\mathcal{Z}^{\text{S.C.}}_{\text{pert}}$ is given in (2.49) of \cite{Choi:2019zpz} as:
\begin{equation}
\begin{split}
    \mathcal{Z}^{\text{S.C.},\lambda}_{\text{Pert}} = &\prod_{a\in\lambda} \frac{\left(q v_a^{-1};q\right)_{\infty}}{\left(u^2 v_a^{-1};q\right)_{\infty} }
    \frac{\left(q u^{-2} v_a;q\right)_{\infty}}{\left( v_a ;q\right)_{\infty}}
    \prod_{a,b \in \lambda} 
    \frac{\left(\frac{v_b}{v_a} ;q\right)_{\infty}}{\left(q\frac{v_a}{v_b} ;q\right)_{\infty}} 
     \frac{\left(u^2\frac{v_a}{v_b} ;q\right)_{\infty}}{\left(qu^{-2}\frac{v_b}{v_a} ;q\right)_{\infty}}\\
    &\prod_{a,b \in \lambda}
    \frac{\left(qzu^{-1}\frac{v_a}{v_b} ;q\right)_{\infty}}{\left(z^{-1}u\frac{v_b}{v_a} ;q\right)_{\infty}} 
    \frac{\left(qz^{-1}u^{-1}\frac{v_a}{v_b} ;q\right)_{\infty}}{\left(z u\frac{v_b}{v_a} ;q\right)_{\infty}}
\end{split}
\end{equation}
and implicitly any vanishing factors in the q-Pochhammers are discarded. In fact, this assumption can be dropped as we can write:
\begin{equation}
\begin{split}
    \mathcal{Z}^{\text{S.C.},\lambda}_{\text{Pert}}  = & \left\lVert \prod_{a,b\,\in \lambda}  \frac{\left(u^2\frac{v_a}{v_b};q\right)_{\infty}}{\left(q\frac{v_a}{v_b};q\right)_{\infty}}
    \prod_{a\in \lambda }  \frac{\left(q u^{-2} v_a ;q\right)_{\infty}}{\left(v_a ;q\right)_{\infty}}\right.\\
    &\prod_{\substack{a,b\in \lambda \text{ s.t. } \\  (b \in \lambda^{B} )\cap( i_{a}> i_b) \\ \text{ or } (b \notin \lambda^{B})}} \frac{\left(qz^{-1}u^{-1}\frac{v_a}{v_b};q\right)_{\infty}}{\left(z^{-1}u\frac{v_a}{v_b};q\right)_{\infty}}
    \left.\prod_{\substack{a,b\in \lambda \text{ s.t. } \\ (a \in \lambda^{B} )\cap( i_{b} \leq i_a)}}  \frac{\left(qzu^{-1}\frac{v_a}{v_b};q\right)_{\infty}}{\left(zu\frac{v_a}{v_b};q\right)_{\infty}} \right\rVert_{\text{S.C.}}^2 \\
     = &  \left\Vert \mathcal{Z}^{\lambda}_{\text{1-loop}} \right\rVert_{\text{S.C.}}^2\,.
\end{split}
\end{equation}
where the gluing $\left\lVert \cdot \right\rVert_{\text{S.C.}}^2$ is given by (\ref{eq:SCindexgluing}) and we have used the analytic continuation of the q-Pochhammer (\ref{eq:qpochhameranalyticcontinuation}). In the second line identical cancellations to the perturbative piece of the hemisphere partition function have been made. We also have $\left\lVert\mathcal{Z}^{\lambda}_{\text{Classical}}\right\rVert_{\text{S.C.}}^2  = 1$. In total then:
\begin{equation}
    \mathcal{Z}^{\text{S.C.}}_{S^1\times S^2}(q,t,x,\zeta) =  \sum_{\lambda} \left\lVert
    \mathcal{Z}^{\lambda}_{\text{Classical}}  \mathcal{Z}^{\lambda}_{\text{1-loop}} \mathcal{Z}^{\lambda}_{\text{Vortex}}
    \right\rVert_{\text{S.C.}}^2
\end{equation}
The inclusion of fluxes in the superconformal index can be achieved as usual by shifting fugacities $z\rightarrow zq^{-\frac{\mathfrak{n}_z}{2}}$, $\zeta \rightarrow \zeta q^{-\frac{\mathfrak{n}_\zeta}{2}}$ in one block, and the opposite in the other.  The factorisation is still exact but now the classical piece glues non-trivially. 

\subsection{Geometry of the boundary condition}\label{appendix:geometricboundarycondition}
Here we give evidence to support the claim in section \ref{sec:hemispherelocalisation} that the holomorphic Lagrangian defined by the image of the boundary conditions \label{eq:exceptionalDirichletBCADHM} on the Higgs branch coincides with the attracting Lagrangian $\mathcal{L}_{\lambda}$ of the fixed point labelled by $\lambda$. 

We first show that the gauge group is completely broken in a neighbourhood of any fixed point $\lambda$ of the ADHM theory described in section \ref{sec:background}, i.e. there are no non-trivial $GL(N,\mathbb{C})$ orbits. More precisely, if $\theta \in \mathfrak{gl}(N,\mathbb{C})\simeq \text{Hom}(V,V)$ is such that:
\begin{equation}\label{eq:boundaryconditionbreaking}
    \big( A+[\theta,A], B+[\theta,B], I+\theta I, J-J\theta\big) \in \mathcal{B}_{\lambda} \quad \Rightarrow \quad \theta = 0.
\end{equation}
for the values $(A, B, I)$ at the fixed point $\lambda$ (we will use this notation throughout the remainder of this section). The left hand side is an infinitesimal gauge transformation. This is line with the argument in \cite{Bullimore:2016nji} that there should be no non-trivial gauge orbits in $\mathcal{B}_{\lambda}$, at least locally around the fixed point. For the left side of (\ref{eq:boundaryconditionbreaking}) to hold we must have:
\begin{equation}\label{eq:boudnaryconditionbreakingequations}
\begin{split}
P_{ab} \equiv [\theta,A]_{ab} = 0 \qquad &\forall\,a, b \in \lambda \text{ s.t. }  (b \in \lambda^{B} )\cap( i_{a}> i_b) \text{ or } (b \notin \lambda^{B})\\
Q_{ab} \equiv [\theta,B]_{ab}  = 0 \qquad &\forall\,a, b \in \lambda \text{ s.t. }  (a \in \lambda^{B} )\cap( i_{b} \leq i_a) \\
(\theta i)_{a}=0 \qquad &\forall\, a.
\end{split}
\end{equation}
The last equation implies that, $\theta_{a(1,1)} = 0$ $\forall\, a$. If $\lambda$ has more than one row i.e. $\lambda_{1}^{\vee} > 0 $ then $\forall a$:
\begin{equation}
\begin{split}
    P_{a (1,1)} = \theta_{a (2,1)} &{A}_{(2,1)(1,1)} - {A}_{a b}\theta_{b(1,1)} = \theta_{a (2,1) } = 0\\
    \Rightarrow\, &P_{a (2,1)} = \theta_{a (3,1) } = 0\, ...
\end{split}
\end{equation}
so:
\begin{equation}\label{eq:Youngdiagramfirstcolumn}
    \theta_{ab} = 0 \quad \forall \,b \in (\cdot,1).
\end{equation}
where $(\cdot,1)$ is the first column. Next, we consider a general $\theta_{uv}$ and show that it must be zero. We enumerate the equations in (\ref{eq:boudnaryconditionbreakingequations}) that it can appear in, also noting the restrictions on indices on the RHS:
\begin{enumerate}[label=(\alph*)]
    \item If $ v \notin  (1,\cdot)$, $ u \in (1,\cdot)$ then $P_{u, \uparrow v} = \theta_{uv} = 0$.
    \item If $ v \notin  (1,\cdot)$, $ u \notin (1,\cdot)$ then $P_{u, \uparrow v} = \theta_{uv} - \theta_{\uparrow u, \uparrow v}= 0$.
    \item If $v \notin \lambda^{B}$, $u \notin \lambda^{B}$ then  $P_{\downarrow u,v} = \theta_{\downarrow u,\downarrow v } - \theta_{uv} = 0$
    \item If $v\in  \lambda^{B}$,  $u \notin \lambda^{B}$ and $i_u \geq i_v$ then $P_{\downarrow u,v} = - \theta_{uv} = 0$.
    \item If $u \in \lambda^{B}$, $u \notin (\cdot,1)$, $v \notin (\cdot,1)$ and $i_u\geq i_v$ then $Q_{u,\leftarrow v} = \theta_{uv} - \theta_{\leftarrow u,\leftarrow v}=0$.
    \item If $u \in \lambda^{B}$, $u \in (\cdot,1)$, $v \notin (\cdot,1)$ and $i_u \geq i_v$ then $Q_{u,\leftarrow v} = \theta_{uv} =0$.
    \item If $u\notin \lambda^{R}$, $(\rightarrow u) \in \lambda^{B}$, $v \notin \lambda^{R}$ and $i_u \geq i_v$ then $Q_{\rightarrow u, v} = \theta_{\rightarrow u, \rightarrow v} - \theta_{uv} = 0$.
    \item If $u\notin \lambda^{R}$, $(\rightarrow u) \in \lambda^{B}$, $v \in \lambda^{R}$ and $i_u\geq i_v$ then $Q_{\rightarrow u, v} = - \theta_{uv} = 0$.
\end{enumerate}
All the non-zero $(A,B,I)$ are normalised to 1 here. Note that the above already demand some of the $\{\theta_{uv}\}$ to be $0$. We consider different cases, assuming $v\notin (\cdot,1)$:
\begin{itemize}
    \item If $i_v > i_u$, repeated applying (b):
            \begin{equation}
                \theta_{uv} - \theta_{\uparrow^{i_u-1}u, \uparrow^{i_u-1} v} = \theta_{uv}  = 0
            \end{equation}
            where the last equality comes from (a). So:
            \begin{equation}
                \theta_{uv} = 0 \quad \forall \quad u,v\in \lambda\,\rvert\,i_u < i_v.
            \end{equation}
            
    \item If $i_u \geq i_v$ and also $l(u) >  l(v)$ (necessarily requiring that $j_v > j_u$) then repeated application of (c) gives:
        \begin{equation}\label{eq:Youngdiagramdownshift}
            \theta_{\downarrow^{l(v)}u, \downarrow^{l(v)}v} - \theta_{uv}= -\theta_{uv} =  0,
        \end{equation}
        where the last equality follows from (d). 
        
        \item It remains to show that $\theta_{uv}=0$ where $i_u \geq i_v$ and $l(u) \leq l(v)$. Repeated application of (c) gives:
        \begin{equation}
            \theta_{uv} = \theta_{u' v'}
        \end{equation}
        where $u'\equiv (\downarrow^{l(u)} u) \, \in \lambda^{B}$, and $v'\equiv (\downarrow^{l(u)}v)$. If $u' \in (\cdot,1)$ i.e. $u'$ is the lowest box in the first column of the Young diagram, we are done since $\theta_{u'v'} = \theta_{uv} = 0$ by (f) (if $v \in (\cdot,1)$ then (\ref{eq:Youngdiagramfirstcolumn}) implies $\theta_{uv}=0$ already). Otherwise, applying (e) we have:
        \begin{equation}
            \theta_{uv} = \theta_{u'v'} = \theta_{\leftarrow u',\leftarrow v'} \equiv \theta_{u'', v''}.
        \end{equation}
        If $v'' \in (\cdot,1)$ we are done. Otherwise, we still clearly have $i_{u''}\geq i_{v''}$ and if $l(u'') > l(v'')$, then from (\ref{eq:Youngdiagramdownshift}) we have:
        \begin{equation}
            \theta_{u''v''} = \theta_{uv} = 0.
        \end{equation}
        If instead $l(u'') \leq l(v'')$, we can iterate what we just did to $u''$ and $v''$, decreasing the $j$ coordinate of $u$ and $v$ each time. Eventually we must have that:
        \begin{equation}
            \theta_{uv} = \theta_{\tilde{u}\tilde{v}}
        \end{equation}
        where either $l(\tilde{u}) > l(\tilde{v})$, or either $\tilde{u}$ or $\tilde{v} \in (\cdot,1)$, and thus $\theta_{uv}= 0$. 
\end{itemize}
Therefore we have shown (\ref{eq:boundaryconditionbreaking}), i.e. that $\theta_{uv} = 0$ $\forall \,u,v$.

\paragraph{Tangent space character} 
With this information, we now prove:
\begin{equation}\label{eq:appendixtangentspacecharacter}
    T_{\lambda}\mathcal{B}_{\lambda}^{\text{Bulk}} =\sum_{s\in\lambda} z^{a_{\lambda}(s)+l_{\lambda}(s)+1} t^{\frac{1}{2}(-a_{\lambda}(s)+l_{\lambda}(s)+1)} =  T_{\lambda}(\mathcal{L}_{\lambda}),
\end{equation}
\text{i.e.} $T_{\lambda}\mathcal{B}_{\lambda}^{\text{Bulk}} \subset T_{\lambda}\text{Hilb}^N(\mathbb{C}^2) $ is precisely the positive weight subspace under the $z$-action (since hook-length is always positive), which is locally the attracting Lagrangian submanifold. We can describe the tangent space to $\lambda$ in $\mathcal{B}_{\lambda}^{\text{Bulk}}$ (which we have shown to coincide with $\mathcal{B}_{\lambda}\cap \mu_{\mathbb{C}}^{-1}(0)$ in a neighbourhood of the fixed point), as the kernel of the differential of the complex moment map $d\mu$:
\begin{equation}\label{eq:tangentspacecomplex}
    \begin{split}
    \bigoplus\limits_{\substack{a, b \in \lambda \text{ s.t. } \\(a \in \lambda^{B} )\cap( i_{b} \leq i_a)}} \text{Hom}&(V_a,Q_1 \otimes V_b)\\[-15pt]
    &\bigoplus\\
    \bigoplus\limits_{\substack{ a, b \in \lambda \text{ s.t. } \\  (b \in \lambda^{B} )\cap( i_{a}> i_b) \\ \text{ or } (b \notin \lambda^{B}) }} \text{Hom}&(V_a,Q_2 \otimes V_b)\\[-15pt]
    &\bigoplus\\
    \text{Hom}(V, &\bigwedge^2 Q \otimes W)
    \end{split}
    \quad\overset{d\mu}{\longrightarrow}\quad
    \begin{split}
    \text{Hom}(V,V)\otimes\bigwedge^{2}Q
    \end{split}
\end{equation}
where $Q_i$ is the 1-dimensional module of $T_i$, $i=1,2$ (with fugacities $t_1,t_2$), and $Q$ is the 2-dimensional module of $T^2 = T_1\times T_2$. They have been inserted so that $d\mu$ is an equivariant map. No quotient by the gauge action is needed by the result just proven. That is, for $(\tilde{A}, \tilde{B}, \tilde{J})$ in the vector space on the left hand side of (\ref{eq:tangentspacecomplex}):
\begin{equation}
    d\mu (\tilde{A},\tilde{B}, \tilde{J}) =\left[A,\tilde{B} \right] + \left[ \tilde{A},B\right] + I \tilde{J}.
\end{equation}
and $T_{\lambda}\mathcal{B}_{\lambda}^{\text{Bulk}} = \text{Ker}(d\mu)$. $d\mu$ must be surjective since $\text{Ker}(d\mu)$ is $N$ (complex) dimensional as the tangent space to $\mathcal{B}_{\lambda}^{\text{Bulk}}$, there are no-nontrivial gauge orbits in it, and $\mu^{-1}(0)$ provides at most $N^2$ constraints.  Notice that $\lambda(t_1,t_2)$ makes $V$ (and in fact each $V_a$) into a $T^2$-module. We abuse notation by representing the weight of an eigenspace with the eigenspace itself, e.g. $V_a = t_1^{-i_a+1} t_2^{-j_a+1} = z^{-i_a+j_b}t^{\frac{1}{2}(-i_a-j_a+2)} \equiv t \tilde{v}_a^{-1}$, and compute the $T^2$-character of the tangent space:
\begin{equation}
\begin{split}
    T_{\lambda}\mathcal{B}_{\lambda}^{\text{Bulk}} =& \sum\limits_{\substack{a, b \in \lambda \text{ s.t. } \\(a \in \lambda^{B} )\cap( i_{b} \leq i_a)}} t_1 V_a^{*} \otimes V_b +
    \sum \limits_{\substack{ a, b \in \lambda \text{ s.t. } \\  (b \in \lambda^{B} )\cap( i_{a}> i_b) \\ \text{ or } (b \notin \lambda^{B}) }}  t_2 V_a^{*} \otimes V_b\\
    &\quad +t_1t_2 \sum_{a\in \lambda}V_a^{*} - t_1t_2 \sum_{a, b \in \lambda} V_a^{*}\otimes V_b \\
    =&\quad t\, \Bigg( \sum\limits_{\substack{a, b \in \lambda \text{ s.t. } \\(a \in \lambda^{B} )\cap( i_{b} \leq i_a)}} zt^{-\frac{1}{2}} \frac{\tilde{v}_a}{\tilde{v}_b}+
    \sum \limits_{\substack{ a, b \in \lambda \text{ s.t. } \\  (b \in \lambda^{B} )\cap( i_{a}> i_b) \\ \text{ or } (b \notin \lambda^{B}) }}  z^{-1}t^{-\frac{1}{2}} \frac{\tilde{v}_a}{\tilde{v}_b} -   \sum\limits_{\substack{a, b \in \lambda \text{ s.t. }\\ b \neq (1,1)}} \frac{\tilde{v}_a}{\tilde{v}_b}  \Bigg)\\
    =&\quad \sum_{\substack{a, b \in \lambda \text{ s.t. } \\ (a\in \lambda^{B})\cap(b\in\lambda^{R}) \\ \cap(i_b \leq i_a) }} zt^{\frac{1}{2}} \frac{\tilde{v}_a}{\tilde{v}_b}
\end{split}
\end{equation}
where we have performed precisely the same cancellations as in (\ref{eq:1-looppiececancellatinos}). Using the aforementioned correspondence between a pair of boxes $(a,b)$ where $a\in \lambda^{B}$ and $ b \in \lambda^{R}$ with a single box $s$, we arrive at (\ref{eq:appendixtangentspacecharacter}).

\bibliographystyle{JHEP}
\bibliography{blocksandvortices}

\providecommand{\href}[2]{#2}\begingroup\raggedright\begin{thebibliography}{10}

\bibitem{Bullimore:2016hdc}
M.~Bullimore, T.~Dimofte, D.~Gaiotto, J.~Hilburn and H.-C. Kim, \emph{{Vortices
  and Vermas}}, \href{https://doi.org/10.4310/ATMP.2018.v22.n4.a1}{\emph{Adv.
  Theor. Math. Phys.} {\bfseries 22} (2018) 803}
  [\href{https://arxiv.org/abs/1609.04406}{{\ttfamily 1609.04406}}].

\bibitem{Aganagic:2017smx}
M.~Aganagic, E.~Frenkel and A.~Okounkov, \emph{{Quantum $q$-Langlands
  Correspondence}}, \href{https://doi.org/10.1090/mosc/278}{\emph{Trans. Moscow
  Math. Soc.} {\bfseries 79} (2018) 1}
  [\href{https://arxiv.org/abs/1701.03146}{{\ttfamily 1701.03146}}].

\bibitem{Bullimore:2015lsa}
M.~Bullimore, T.~Dimofte and D.~Gaiotto, \emph{{The Coulomb Branch of 3d
  ${\mathcal{N}= 4}$ Theories}},
  \href{https://doi.org/10.1007/s00220-017-2903-0}{\emph{Commun. Math. Phys.}
  {\bfseries 354} (2017) 671}
  [\href{https://arxiv.org/abs/1503.04817}{{\ttfamily 1503.04817}}].

\bibitem{Rozansky:1996bq}
L.~Rozansky and E.~Witten, \emph{{HyperKahler geometry and invariants of three
  manifolds}}, \href{https://doi.org/10.1007/s000290050016}{\emph{Selecta
  Math.} {\bfseries 3} (1997) 401}
  [\href{https://arxiv.org/abs/hep-th/9612216}{{\ttfamily hep-th/9612216}}].

\bibitem{Nakajima:2015txa}
H.~Nakajima, \emph{{Towards a mathematical definition of Coulomb branches of
  $3$-dimensional $\mathcal{N}=4$ gauge theories, I}},
  \href{https://doi.org/10.4310/ATMP.2016.v20.n3.a4}{\emph{Adv. Theor. Math.
  Phys.} {\bfseries 20} (2016) 595}
  [\href{https://arxiv.org/abs/1503.03676}{{\ttfamily 1503.03676}}].

\bibitem{Braverman:2016wma}
A.~Braverman, M.~Finkelberg and H.~Nakajima, \emph{{Towards a mathematical
  definition of Coulomb branches of $3$-dimensional $\mathcal{N} = 4$ gauge
  theories, II}}, \href{https://doi.org/10.4310/ATMP.2018.v22.n5.a1}{\emph{Adv.
  Theor. Math. Phys.} {\bfseries 22} (2018) 1071}
  [\href{https://arxiv.org/abs/1601.03586}{{\ttfamily 1601.03586}}].

\bibitem{Cremonesi:2013lqa}
S.~Cremonesi, A.~Hanany and A.~Zaffaroni, \emph{{Monopole operators and Hilbert
  series of Coulomb branches of $3d$ $\mathcal{N} = 4$ gauge theories}},
  \href{https://doi.org/10.1007/JHEP01(2014)005}{\emph{JHEP} {\bfseries 01}
  (2014) 005} [\href{https://arxiv.org/abs/1309.2657}{{\ttfamily 1309.2657}}].

\bibitem{Bullimore:2018jlp}
M.~Bullimore, A.~Ferrari and H.~Kim, \emph{{Twisted Indices of 3d ${\mathcal N}
  = 4$ Gauge Theories and Enumerative Geometry of Quasi-Maps}},
  \href{https://doi.org/10.1007/JHEP07(2019)014}{\emph{JHEP} {\bfseries 07}
  (2019) 014} [\href{https://arxiv.org/abs/1812.05567}{{\ttfamily
  1812.05567}}].

\bibitem{Bullimore:2019qnt}
M.~Bullimore, A.~E. Ferrari and H.~Kim, \emph{{The 3d Twisted Index and
  Wall-Crossing}},  \href{https://arxiv.org/abs/1912.09591}{{\ttfamily
  1912.09591}}.

\bibitem{Costello:2018swh}
K.~Costello, T.~Creutzig and D.~Gaiotto, \emph{{Higgs and Coulomb branches from
  vertex operator algebras}},
  \href{https://doi.org/10.1007/JHEP03(2019)066}{\emph{JHEP} {\bfseries 03}
  (2019) 066} [\href{https://arxiv.org/abs/1811.03958}{{\ttfamily
  1811.03958}}].

\bibitem{Kapustin:2010xq}
A.~Kapustin, B.~Willett and I.~Yaakov, \emph{{Nonperturbative Tests of
  Three-Dimensional Dualities}},
  \href{https://doi.org/10.1007/JHEP10(2010)013}{\emph{JHEP} {\bfseries 10}
  (2010) 013} [\href{https://arxiv.org/abs/1003.5694}{{\ttfamily 1003.5694}}].

\bibitem{Gang:2011xp}
D.~Gang, E.~Koh, K.~Lee and J.~Park, \emph{{ABCD of 3d ${\cal N}=8$ and 4
  Superconformal Field Theories}},
  \href{https://arxiv.org/abs/1108.3647}{{\ttfamily 1108.3647}}.

\bibitem{Benini:2015eyy}
F.~Benini, K.~Hristov and A.~Zaffaroni, \emph{{Black hole microstates in
  AdS$_{4}$ from supersymmetric localization}},
  \href{https://doi.org/10.1007/JHEP05(2016)054}{\emph{JHEP} {\bfseries 05}
  (2016) 054} [\href{https://arxiv.org/abs/1511.04085}{{\ttfamily
  1511.04085}}].

\bibitem{Benini:2016rke}
F.~Benini, K.~Hristov and A.~Zaffaroni, \emph{{Exact microstate counting for
  dyonic black holes in AdS4}},
  \href{https://doi.org/10.1016/j.physletb.2017.05.076}{\emph{Phys. Lett. B}
  {\bfseries 771} (2017) 462}
  [\href{https://arxiv.org/abs/1608.07294}{{\ttfamily 1608.07294}}].

\bibitem{Choi:2019dfu}
S.~Choi and C.~Hwang, \emph{{Universal 3d Cardy Block and Black Hole Entropy}},
  \href{https://doi.org/10.1007/JHEP03(2020)068}{\emph{JHEP} {\bfseries 03}
  (2020) 068} [\href{https://arxiv.org/abs/1911.01448}{{\ttfamily
  1911.01448}}].

\bibitem{Choi:2019zpz}
S.~Choi, C.~Hwang and S.~Kim, \emph{{Quantum vortices, M2-branes and black
  holes}},  \href{https://arxiv.org/abs/1908.02470}{{\ttfamily 1908.02470}}.

\bibitem{Nakajima:1994nid}
H.~Nakajima, \emph{{Instantons on ALE spaces, quiver varieties, and Kac-Moody
  algebras}}, \href{https://doi.org/10.1215/S0012-7094-94-07613-8}{\emph{Duke
  Math. J.} {\bfseries 76} (1994) 365}.

\bibitem{Braden:2014iea}
T.~Braden, A.~Licata, N.~Proudfoot and B.~Webster, \emph{{Quantizations of
  conical symplectic resolutions II: category $\mathcal O$ and symplectic
  duality}},  \href{https://arxiv.org/abs/1407.0964}{{\ttfamily 1407.0964}}.

\bibitem{MR3594663}
T.~Braden, A.~Licata, N.~Proudfoot and B.~Webster, \emph{Quantizations of
  conical symplectic resolutions}, {\emph{Ast{\'e}risque} (2016)
  iii{\textendash}iv}.

\bibitem{nakajima1998}
H.~Nakajima, \emph{Quiver varieties and kac-moody algebras},
  \href{https://doi.org/10.1215/S0012-7094-98-09120-7}{\emph{Duke Math. J.}
  {\bfseries 91} (1998) 515}.

\bibitem{Aganagic:2017gsx}
M.~Aganagic and A.~Okounkov, \emph{{Quasimap counts and Bethe eigenfunctions}},
  {\emph{Moscow Math. J.} {\bfseries 17} (2017) 565}
  [\href{https://arxiv.org/abs/1704.08746}{{\ttfamily 1704.08746}}].

\bibitem{Smirnov:2016cqz}
A.~Smirnov, \emph{{Quantum difference equations for quiver varieties}}, Ph.D.
  thesis, Columbia U., 2016.
\newblock 10.7916/D8RN37T6.

\bibitem{Koroteev:2017nab}
P.~Koroteev, P.~P. Pushkar, A.~Smirnov and A.~M. Zeitlin, \emph{{Quantum
  K-theory of Quiver Varieties and Many-Body Systems}},
  \href{https://arxiv.org/abs/1705.10419}{{\ttfamily 1705.10419}}.

\bibitem{Beem:2012mb}
C.~Beem, T.~Dimofte and S.~Pasquetti, \emph{{Holomorphic Blocks in Three
  Dimensions}}, \href{https://doi.org/10.1007/JHEP12(2014)177}{\emph{JHEP}
  {\bfseries 12} (2014) 177} [\href{https://arxiv.org/abs/1211.1986}{{\ttfamily
  1211.1986}}].

\bibitem{Pasquetti:2011fj}
S.~Pasquetti, \emph{{Factorisation of N = 2 Theories on the Squashed
  3-Sphere}}, \href{https://doi.org/10.1007/JHEP04(2012)120}{\emph{JHEP}
  {\bfseries 04} (2012) 120} [\href{https://arxiv.org/abs/1111.6905}{{\ttfamily
  1111.6905}}].

\bibitem{Hwang:2012jh}
C.~Hwang, H.-C. Kim and J.~Park, \emph{{Factorization of the 3d superconformal
  index}}, \href{https://doi.org/10.1007/JHEP08(2014)018}{\emph{JHEP}
  {\bfseries 08} (2014) 018} [\href{https://arxiv.org/abs/1211.6023}{{\ttfamily
  1211.6023}}].

\bibitem{Cabo-Bizet:2016ars}
A.~Cabo-Bizet, \emph{{Factorising the 3D Topologically Twisted Index}},
  \href{https://doi.org/10.1007/JHEP04(2017)115}{\emph{JHEP} {\bfseries 04}
  (2017) 115} [\href{https://arxiv.org/abs/1606.06341}{{\ttfamily
  1606.06341}}].

\bibitem{Crew:2020jyf}
S.~Crew, N.~Dorey and D.~Zhang, \emph{{Factorisation of 3d $\mathcal{N}=4$
  Twisted Indices and the Geometry of Vortex Moduli Space}},
  \href{https://doi.org/10.1007/JHEP08(2020)015}{\emph{JHEP} {\bfseries 08}
  (2020) 015} [\href{https://arxiv.org/abs/2002.04573}{{\ttfamily
  2002.04573}}].

\bibitem{Benini:2013yva}
F.~Benini and W.~Peelaers, \emph{{Higgs branch localization in three
  dimensions}}, \href{https://doi.org/10.1007/JHEP05(2014)030}{\emph{JHEP}
  {\bfseries 05} (2014) 030} [\href{https://arxiv.org/abs/1312.6078}{{\ttfamily
  1312.6078}}].

\bibitem{Fujitsuka:2013fga}
M.~Fujitsuka, M.~Honda and Y.~Yoshida, \emph{{Higgs branch localization of 3d
  \ensuremath{\mathscr{N}} = 2 theories}},
  \href{https://doi.org/10.1093/ptep/ptu158}{\emph{PTEP} {\bfseries 2014}
  (2014) 123B02} [\href{https://arxiv.org/abs/1312.3627}{{\ttfamily
  1312.3627}}].

\bibitem{Bullimore:2020jdq}
M.~Bullimore, S.~Crew and D.~Zhang, \emph{{Boundaries, Vermas, and
  Factorisation}},  \href{https://arxiv.org/abs/2010.09741}{{\ttfamily
  2010.09741}}.

\bibitem{Bullimore:2016nji}
M.~Bullimore, T.~Dimofte, D.~Gaiotto and J.~Hilburn, \emph{{Boundaries, Mirror
  Symmetry, and Symplectic Duality in 3d $\mathcal{N}=4$ Gauge Theory}},
  \href{https://doi.org/10.1007/JHEP10(2016)108}{\emph{JHEP} {\bfseries 10}
  (2016) 108} [\href{https://arxiv.org/abs/1603.08382}{{\ttfamily
  1603.08382}}].

\bibitem{Atiyah:1978ri}
M.~Atiyah, N.~J. Hitchin, V.~Drinfeld and Y.~Manin, \emph{{Construction of
  Instantons}}, \href{https://doi.org/10.1016/0375-9601(78)90141-X}{\emph{Phys.
  Lett. A} {\bfseries 65} (1978) 185}.

\bibitem{Nekrasov:1998ss}
N.~Nekrasov and A.~S. Schwarz, \emph{{Instantons on noncommutative R**4 and
  (2,0) superconformal six-dimensional theory}},
  \href{https://doi.org/10.1007/s002200050490}{\emph{Commun. Math. Phys.}
  {\bfseries 198} (1998) 689}
  [\href{https://arxiv.org/abs/hep-th/9802068}{{\ttfamily hep-th/9802068}}].

\bibitem{Aharony:2008ug}
O.~Aharony, O.~Bergman, D.~L. Jafferis and J.~Maldacena, \emph{{N=6
  superconformal Chern-Simons-matter theories, M2-branes and their gravity
  duals}}, \href{https://doi.org/10.1088/1126-6708/2008/10/091}{\emph{JHEP}
  {\bfseries 10} (2008) 091} [\href{https://arxiv.org/abs/0806.1218}{{\ttfamily
  0806.1218}}].

\bibitem{Costello:2017fbo}
K.~Costello, \emph{{Holography and Koszul duality: the example of the $M2$
  brane}},  \href{https://arxiv.org/abs/1705.02500}{{\ttfamily 1705.02500}}.

\bibitem{Kodera:2016faj}
R.~Kodera and H.~Nakajima, \emph{{Quantized Coulomb branches of Jordan quiver
  gauge theories and cyclotomic rational Cherednik algebras}}, {\emph{Proc.
  Symp. Pure Math.} {\bfseries 98} (2018) 49}
  [\href{https://arxiv.org/abs/1608.00875}{{\ttfamily 1608.00875}}].

\bibitem{Yoshida:2014ssa}
Y.~Yoshida and K.~Sugiyama, \emph{{Localization of 3d $\mathcal{N}=2$
  Supersymmetric Theories on $S^1 \times D^2$}},
  \href{https://arxiv.org/abs/1409.6713}{{\ttfamily 1409.6713}}.

\bibitem{nakajima1999lectures}
H.~Nakajima, \emph{Lectures on Hilbert Schemes of Points on Surfaces},
  University lecture series. American Mathematical Society, 1999.

\bibitem{Smirnov:2018drm}
A.~Smirnov, \emph{{Elliptic stable envelope for Hilbert scheme of points in the
  plane}},  \href{https://arxiv.org/abs/1804.08779}{{\ttfamily 1804.08779}}.

\bibitem{macdonald1998symmetric}
I.~G. Macdonald, \emph{Symmetric functions and Hall polynomials}. Oxford
  University Press, 1998.

\bibitem{deBoer:1996mp}
J.~de~Boer, K.~Hori, H.~Ooguri and Y.~Oz, \emph{{Mirror symmetry in
  three-dimensional gauge theories, quivers and D-branes}},
  \href{https://doi.org/10.1016/S0550-3213(97)00125-9}{\emph{Nucl. Phys. B}
  {\bfseries 493} (1997) 101}
  [\href{https://arxiv.org/abs/hep-th/9611063}{{\ttfamily hep-th/9611063}}].

\bibitem{Porrati:1996xi}
M.~Porrati and A.~Zaffaroni, \emph{{M theory origin of mirror symmetry in
  three-dimensional gauge theories}},
  \href{https://doi.org/10.1016/S0550-3213(97)00061-8}{\emph{Nucl. Phys. B}
  {\bfseries 490} (1997) 107}
  [\href{https://arxiv.org/abs/hep-th/9611201}{{\ttfamily hep-th/9611201}}].

\bibitem{Intriligator:1996ex}
K.~A. Intriligator and N.~Seiberg, \emph{{Mirror symmetry in three-dimensional
  gauge theories}},
  \href{https://doi.org/10.1016/0370-2693(96)01088-X}{\emph{Phys. Lett. B}
  {\bfseries 387} (1996) 513}
  [\href{https://arxiv.org/abs/hep-th/9607207}{{\ttfamily hep-th/9607207}}].

\bibitem{Bobev:2015kza}
N.~Bobev, M.~Bullimore and H.-C. Kim, \emph{{Supersymmetric Casimir Energy and
  the Anomaly Polynomial}},
  \href{https://doi.org/10.1007/JHEP09(2015)142}{\emph{JHEP} {\bfseries 09}
  (2015) 142} [\href{https://arxiv.org/abs/1507.08553}{{\ttfamily
  1507.08553}}].

\bibitem{Dimofte:2017tpi}
T.~Dimofte, D.~Gaiotto and N.~M. Paquette, \emph{{Dual boundary conditions in
  3d SCFT\textquoteright{}s}},
  \href{https://doi.org/10.1007/JHEP05(2018)060}{\emph{JHEP} {\bfseries 05}
  (2018) 060} [\href{https://arxiv.org/abs/1712.07654}{{\ttfamily
  1712.07654}}].

\bibitem{Dinkins:2019pwj}
H.~Dinkins and A.~Smirnov, \emph{{Characters of tangent spaces at torus fixed
  points and $3d$-mirror symmetry}},
  \href{https://arxiv.org/abs/1908.01199}{{\ttfamily 1908.01199}}.

\bibitem{hunterpending}
S.~Crew, H.~Dinkins and D.~Zhang, \emph{{Hemisphere Blocks and Mirror Symmetry
  of Twisted Indices (in preparation)}}, .

\bibitem{Okounkov:2015spn}
A.~Okounkov, \emph{{Lectures on K-theoretic computations in enumerative
  geometry}},  \href{https://arxiv.org/abs/1512.07363}{{\ttfamily 1512.07363}}.

\bibitem{CIOCANFONTANINE2012268}
I.~Ciocan-Fontanine, M.~Konvalinka and I.~Pak, \emph{Quantum cohomology of
  hilbn(c2) and the weighted hook walk on young diagrams},
  \href{https://doi.org/https://doi.org/10.1016/j.jalgebra.2011.10.011}{\emph{Journal
  of Algebra} {\bfseries 349} (2012) 268 }.

\bibitem{Smirnov:2016vaw}
A.~Smirnov, \emph{{Rationality of capped descendent vertex in $K$-theory}},
  \href{https://arxiv.org/abs/1612.01048}{{\ttfamily 1612.01048}}.

\bibitem{Okounkov:2016sya}
A.~Okounkov and A.~Smirnov, \emph{{Quantum difference equation for Nakajima
  varieties}},  \href{https://arxiv.org/abs/1602.09007}{{\ttfamily
  1602.09007}}.

\bibitem{Okounkov:2018huu}
A.~Okounkov, \emph{{On the crossroads of enumerative geometry and geometric
  representation theory}},  \href{https://arxiv.org/abs/1801.09818}{{\ttfamily
  1801.09818}}.

\bibitem{Kononov:2019fni}
Y.~Kononov, A.~Okounkov and A.~Osinenko, \emph{{The 2-leg vertex in K-theoretic
  DT theory}},  \href{https://arxiv.org/abs/1905.01523}{{\ttfamily
  1905.01523}}.

\bibitem{Iqbal:2007ii}
A.~Iqbal, C.~Kozcaz and C.~Vafa, \emph{{The Refined topological vertex}},
  \href{https://doi.org/10.1088/1126-6708/2009/10/069}{\emph{JHEP} {\bfseries
  10} (2009) 069} [\href{https://arxiv.org/abs/hep-th/0701156}{{\ttfamily
  hep-th/0701156}}].

\bibitem{Bonelli:2013mma}
G.~Bonelli, A.~Sciarappa, A.~Tanzini and P.~Vasko, \emph{{Vortex partition
  functions, wall crossing and equivariant Gromov-Witten invariants}},
  \href{https://doi.org/10.1007/s00220-014-2193-8}{\emph{Commun. Math. Phys.}
  {\bfseries 333} (2015) 717}
  [\href{https://arxiv.org/abs/1307.5997}{{\ttfamily 1307.5997}}].

\bibitem{Bonelli:2013rja}
G.~Bonelli, A.~Sciarappa, A.~Tanzini and P.~Vasko, \emph{{The Stringy Instanton
  Partition Function}},
  \href{https://doi.org/10.1007/JHEP01(2014)038}{\emph{JHEP} {\bfseries 01}
  (2014) 038} [\href{https://arxiv.org/abs/1306.0432}{{\ttfamily 1306.0432}}].

\bibitem{Gaiotto:2019wcc}
D.~Gaiotto and J.~Oh, \emph{{Aspects of $\Omega$-deformed M-theory}},
  \href{https://arxiv.org/abs/1907.06495}{{\ttfamily 1907.06495}}.

\bibitem{Gaiotto:2019mmf}
D.~Gaiotto and T.~Okazaki, \emph{{Sphere correlation functions and Verma
  modules}}, \href{https://doi.org/10.1007/JHEP02(2020)133}{\emph{JHEP}
  {\bfseries 02} (2020) 133}
  [\href{https://arxiv.org/abs/1911.11126}{{\ttfamily 1911.11126}}].

\bibitem{braverman2014macdonald}
A.~{Braverman}, M.~{Finkelberg} and J.~{Shiraishi}, \emph{{Macdonald
  polynomials, Laumon spaces and perverse coherent sheaves}}, {\emph{arXiv
  e-prints} (2012) arXiv:1206.3131}
  [\href{https://arxiv.org/abs/1206.3131}{{\ttfamily 1206.3131}}].

\bibitem{GANSNER198171}
E.~R. Gansner, \emph{The hillman-grassl correspondence and the enumeration of
  reverse plane partitions},
  \href{https://doi.org/https://doi.org/10.1016/0097-3165(81)90041-8}{\emph{Journal
  of Combinatorial Theory, Series A} {\bfseries 30} (1981) 71 }.

\bibitem{Awata:2008ed}
H.~Awata and H.~Kanno, \emph{{Refined BPS state counting from Nekrasov's
  formula and Macdonald functions}},
  \href{https://doi.org/10.1142/S0217751X09043006}{\emph{Int. J. Mod. Phys. A}
  {\bfseries 24} (2009) 2253}
  [\href{https://arxiv.org/abs/0805.0191}{{\ttfamily 0805.0191}}].

\bibitem{Aganagic:2003db}
M.~Aganagic, A.~Klemm, M.~Marino and C.~Vafa, \emph{{The Topological vertex}},
  \href{https://doi.org/10.1007/s00220-004-1162-z}{\emph{Commun. Math. Phys.}
  {\bfseries 254} (2005) 425}
  [\href{https://arxiv.org/abs/hep-th/0305132}{{\ttfamily hep-th/0305132}}].

\bibitem{Iqbal:2004ne}
A.~Iqbal and A.-K. Kashani-Poor, \emph{{The Vertex on a strip}},
  \href{https://doi.org/10.4310/ATMP.2006.v10.n3.a2}{\emph{Adv. Theor. Math.
  Phys.} {\bfseries 10} (2006) 317}
  [\href{https://arxiv.org/abs/hep-th/0410174}{{\ttfamily hep-th/0410174}}].

\bibitem{Taki:2007dh}
M.~Taki, \emph{{Refined Topological Vertex and Instanton Counting}},
  \href{https://doi.org/10.1088/1126-6708/2008/03/048}{\emph{JHEP} {\bfseries
  03} (2008) 048} [\href{https://arxiv.org/abs/0710.1776}{{\ttfamily
  0710.1776}}].

\bibitem{Closset:2016arn}
C.~Closset and H.~Kim, \emph{{Comments on twisted indices in 3d supersymmetric
  gauge theories}}, \href{https://doi.org/10.1007/JHEP08(2016)059}{\emph{JHEP}
  {\bfseries 08} (2016) 059}
  [\href{https://arxiv.org/abs/1605.06531}{{\ttfamily 1605.06531}}].

\bibitem{Benini:2015noa}
F.~Benini and A.~Zaffaroni, \emph{{A topologically twisted index for
  three-dimensional supersymmetric theories}},
  \href{https://doi.org/10.1007/JHEP07(2015)127}{\emph{JHEP} {\bfseries 07}
  (2015) 127} [\href{https://arxiv.org/abs/1504.03698}{{\ttfamily
  1504.03698}}].

\bibitem{Dimofte:2009bv}
T.~Dimofte and S.~Gukov, \emph{{Refined, Motivic, and Quantum}},
  \href{https://doi.org/10.1007/s11005-009-0357-9}{\emph{Lett. Math. Phys.}
  {\bfseries 91} (2010) 1} [\href{https://arxiv.org/abs/0904.1420}{{\ttfamily
  0904.1420}}].

\bibitem{Bershadsky:1995qy}
M.~Bershadsky, C.~Vafa and V.~Sadov, \emph{{D-branes and topological field
  theories}}, \href{https://doi.org/10.1016/0550-3213(96)00026-0}{\emph{Nucl.
  Phys. B} {\bfseries 463} (1996) 420}
  [\href{https://arxiv.org/abs/hep-th/9511222}{{\ttfamily hep-th/9511222}}].

\bibitem{Hosseini:2016ume}
S.~M. Hosseini and N.~Mekareeya, \emph{{Large $N$ topologically twisted index:
  necklace quivers, dualities, and Sasaki-Einstein spaces}},
  \href{https://doi.org/10.1007/JHEP08(2016)089}{\emph{JHEP} {\bfseries 08}
  (2016) 089} [\href{https://arxiv.org/abs/1604.03397}{{\ttfamily
  1604.03397}}].

\bibitem{Polychronakos:2001mi}
A.~P. Polychronakos, \emph{{Quantum Hall states as matrix Chern-Simons
  theory}}, \href{https://doi.org/10.1088/1126-6708/2001/04/011}{\emph{JHEP}
  {\bfseries 04} (2001) 011}
  [\href{https://arxiv.org/abs/hep-th/0103013}{{\ttfamily hep-th/0103013}}].

\bibitem{Dorey:2016mxm}
N.~Dorey, D.~Tong and C.~Turner, \emph{{Matrix model for non-Abelian quantum
  Hall states}}, \href{https://doi.org/10.1103/PhysRevB.94.085114}{\emph{Phys.
  Rev. B} {\bfseries 94} (2016) 085114}
  [\href{https://arxiv.org/abs/1603.09688}{{\ttfamily 1603.09688}}].

\bibitem{Dorey:2016hoj}
N.~Dorey, D.~Tong and C.~Turner, \emph{{A Matrix Model for WZW}},
  \href{https://doi.org/10.1007/JHEP08(2016)007}{\emph{JHEP} {\bfseries 08}
  (2016) 007} [\href{https://arxiv.org/abs/1604.05711}{{\ttfamily
  1604.05711}}].

\bibitem{Hanany:2003hp}
A.~Hanany and D.~Tong, \emph{{Vortices, instantons and branes}},
  \href{https://doi.org/10.1088/1126-6708/2003/07/037}{\emph{JHEP} {\bfseries
  07} (2003) 037} [\href{https://arxiv.org/abs/hep-th/0306150}{{\ttfamily
  hep-th/0306150}}].

\bibitem{Nakajima:2011yq}
H.~Nakajima, \emph{{Handsaw quiver varieties and finite W-algebras}},
  {\emph{Moscow Math. J.} {\bfseries 12} (2012) 633}
  [\href{https://arxiv.org/abs/1107.5073}{{\ttfamily 1107.5073}}].

\bibitem{Pestun:2016qko}
V.~Pestun, \emph{{Review of localization in geometry}},
  \href{https://doi.org/10.1088/1751-8121/aa6161}{\emph{J. Phys. A} {\bfseries
  50} (2017) 443002} [\href{https://arxiv.org/abs/1608.02954}{{\ttfamily
  1608.02954}}].

\bibitem{Ekholm:2018eee}
T.~Ekholm, P.~Kucharski and P.~Longhi, \emph{{Physics and geometry of
  knots-quivers correspondence}},
  \href{https://arxiv.org/abs/1811.03110}{{\ttfamily 1811.03110}}.

\bibitem{Awata:1995eh}
H.~Awata, S.~Odake and J.~Shiraishi, \emph{{Integral representations of the
  Macdonald symmetric functions}},
  \href{https://doi.org/10.1007/BF02100101}{\emph{Commun. Math. Phys.}
  {\bfseries 179} (1996) 647}
  [\href{https://arxiv.org/abs/q-alg/9506006}{{\ttfamily q-alg/9506006}}].

\bibitem{di2018difference}
P.~Di~Francesco and R.~Kedem, \emph{Difference equations for graded characters
  from quantum cluster algebra}, {\emph{Transformation Groups} {\bfseries 23}
  (2018) 391}.

\bibitem{kirillov1996affine}
A.~N. Kirillov and M.~Noumi, \emph{Affine hecke algebras and raising operators
  for macdonald polynomials}, {\emph{arXiv preprint q-alg/9605004} (1996) }.

\bibitem{Zenkevich:2017ylb}
A.~Nedelin, S.~Pasquetti and Y.~Zenkevich, \emph{{T[SU(N)] duality webs: mirror
  symmetry, spectral duality and gauge/CFT correspondences}},
  \href{https://doi.org/10.1007/JHEP02(2019)176}{\emph{JHEP} {\bfseries 02}
  (2019) 176} [\href{https://arxiv.org/abs/1712.08140}{{\ttfamily
  1712.08140}}].

\bibitem{Dorey:2019kaf}
N.~Dorey and D.~Zhang, \emph{{Superconformal quantum mechanics on K\"ahler
  cones}}, \href{https://doi.org/10.1007/JHEP05(2020)115}{\emph{JHEP}
  {\bfseries 05} (2020) 115}
  [\href{https://arxiv.org/abs/1911.06787}{{\ttfamily 1911.06787}}].

\bibitem{nakajima2012handsaw}
H.~Nakajima, \emph{{Handsaw quiver varieties and finite W-algebras}},
  {\emph{Moscow Math.J.} {\bfseries 12} (2012) 633}
  [\href{https://arxiv.org/abs/1107.5073}{{\ttfamily 1107.5073}}].

\end{thebibliography}\endgroup

\end{document}